\documentclass{article}[11pt]

\usepackage[english]{babel} 
\usepackage[autostyle, english=american]{csquotes}
\MakeOuterQuote{"}
\usepackage{graphicx}
\graphicspath{{figures/}}
\usepackage{amsmath}
\usepackage{amsthm}
\usepackage{amssymb}
\usepackage[margin=2.5cm]{geometry}
\usepackage{hyperref}
\usepackage{xcolor}
\usepackage{dsfont}
\usepackage{rotating}
\usepackage{comment}
\usepackage{float}
\usepackage{morefloats}
\usepackage{placeins}
\usepackage{multirow}

 \usepackage[round]{natbib}
\definecolor{links}{RGB}{0,0,128}

\newtheorem{proposition}{Proposition}[section]
\newtheorem{theorem}{Theorem}[section]
\newtheorem{definition}{Definition}[section]
\newtheorem{lemma}{Lemma}[section]
\newtheorem{corollary}{Corollary}[section]

\newtheorem{remark}{Remark}[section]

\definecolor{todocolor}{RGB}{220,20,60}

\def\sD{{\scriptsize D}}
\newcommand{\sL}{\mbox{ {\scriptsize L}}}
\newcommand{\sLM}{\mbox{ {\scriptsize LM}}}
\newcommand{\sM}{\mbox{ {\scriptsize M}}}
\newcommand{\sE}{\mbox{{\scriptsize E}}}

\def\loss{\ell}

\def\leqas{\stackrel{a.s.}{\leq}}
\def\Exp{\mathbb{E}}
\def\eqas{\stackrel{a.s.}{=}}
\newcommand{\ind}[1]{\mathds{1}_{#1}}

\newcommand{\ds}{\displaystyle}
\newcommand\Tau{\mathrm{T}}
\DeclareMathOperator*{\argmin}{arg\,min\,}

\def\btau{\boldsymbol{\tau}}

\title{Generalized extremiles and risk measures of distorted random variables}

\newcommand{\footremember}[2]{%
    \footnote{#2}
    \newcounter{#1}
    \setcounter{#1}{\value{footnote}}%
}
\newcommand{\footrecall}[1]{%
    \footnotemark[\value{#1}]%
}

\usepackage{color}
\definecolor{dgreen}{rgb}{0.,0.6,0.}

\def\boxit#1{\vbox{\hrule\hbox{\vrule\kern6pt
          \vbox{\kern6pt#1\kern6pt}\kern6pt\vrule}\hrule}}

\author{
Dieter Debrauwer\footremember{KUL}{Department of Mathematics, KU Leuven, Leuven, Belgium},
Irène Gijbels\footnote{Corresponding Author}\, \footrecall{KUL} ,
Klaus Herrmann\footremember{UdeS}{D\'{e}partement de math\'{e}matiques, Universit\'{e} de Sherbrooke, Sherbrooke, Canada}
}

\begin{document}

\maketitle

\begin{abstract}
Quantiles, expectiles and extremiles can be seen as concepts defined via an optimization problem, where this optimization problem is driven by two important ingredients: the loss function as well as a distributional weight function. 
This leads to the formulation of a general class of functionals that contains next to the above concepts many interesting quantities, including also a subclass of distortion risks. The focus of the paper is on developing estimators for such functionals and to establish asymptotic consistency and asymptotic normality of these estimators. The advantage of the general framework is that it allows application to a very broad range of concepts, providing as such estimation tools and tools for statistical inference (for example for construction of confidence intervals) for all involved concepts. After developing the theory for the general functional we apply it to various settings, illustrating the broad applicability. 
In a real data example the developed tools are used in an analysis of  natural disasters. 
\end{abstract}

\section{Introduction}\label{sec: Introduction} 

There is a vast literature on quantiles and expectiles, and their applications in a variety of settings, in particular as risk measures. See \cite{ArtznerEtAl1999}, \cite{Taylor2008}, \cite{BelliniEtAl2017},  among others.  For a random variable $X$, its $\delta$th quantile (with $\delta \in (0,1)$) can be defined as 
\begin{equation}\label{Quantile1}
\argmin_{c\in \mathbb{R}} \Exp[\ell_{\delta}(X,c)] 
\qquad \mbox{with} \quad \ell_{\delta}(x,c) = \left|\delta-1_{\{x\le c\}}\right| |x-c| ,  
\end{equation}
see for example \cite{Ferguson1967} and \cite{KoenkerAndBassett1978}. Herein $1_{A}$ denotes the indicator function on a set $A$, i.e.,  $1_{A}=1$, if $A$ holds, and zero otherwise. 
Expectiles, introduced and first studied by \cite{AignerEtAl1976} and  \cite{NeweyAndPowel1987}, are defined as follows: the $\delta$th expectile (with $\delta\in (0,1)$) of $X$ is given by
\begin{equation}\label{Expectile}
\argmin_{c\in \mathbb{R}} \Exp[\ell_{\delta}(X,c)] 
\qquad \mbox{with} \quad \ell_{\delta}(x,c) = \left|\delta-1_{\{x\le c\}}\right| |x-c|^2  . 
\end{equation}
Quantiles can also be viewed via an alternative optimization problem, namely the $\tau$th quantile (with $\tau\in (0, 1)$) of $X$ (with cumulative distribution function $F_X$) is obtained via 
\begin{equation}\label{Quantile2}
\argmin_{c\in \mathbb{R}} \Exp[J_{\tau} (F_X(X)) \ell_{\delta}(X,c)] 
\qquad \mbox{with} \quad \ell_{\delta}(x,c) =  |x-c| , 
\end{equation}
and  
\begin{equation} \label{densitytauExtremiles} 
J_{\tau}(u) = \begin{cases}
       s(\tau)(1-u)^{s(\tau)-1} &\text{ if } 0<\tau <  \frac{1}{2}\\
       r(\tau) u^{r(\tau)-1} &\text{ if } \frac{1}{2}\leq \tau<1 , 
    \end{cases} 
\end{equation}
with $r(\tau)=s(1-\tau)=\ln(1/2)/ \ln(\tau)$. 
This alternative formulation of quantiles can be found in  
\cite{NeweyAndPowel1987}. Consequently \eqref{Quantile1} and \eqref{Quantile2} (taking in the latter $\tau=\delta$)  
provides two approaches to obtain the $\delta$th quantile of a random variable $X$. Recently \cite{DaouiaAndGijbels2019} studied the class of extremiles, where the $\tau$th extremile of $X$ is obtained via 
\begin{equation}\label{Extremile}
\argmin_{c\in \mathbb{R}} \Exp[J_{\tau} (F_X(X)) \ell_{\delta}(X,c)] 
\qquad \mbox{with} \quad \ell_{\delta}(x,c) =  (x-c)^2 \quad \mbox{and} \quad J_{\tau} \, \mbox{as in \eqref{densitytauExtremiles}} . 
\end{equation}
Of interest is further to note that the function $J_{\tau}$ is in fact the derivative of a function $K_{\tau}$, which is a specific cumulative distribution function defined on the interval $[0,1]$ (see expression \eqref{DtauExtremiles} and Figure \ref{fig: KandJFunctionPlots}).  

The above observations inspire to look at a more general form of an optimization problem, and to study functionals defined via 
\[
\argmin_{c \in \mathbb{R}}  \mathbb{E}[d_{\tau}(F_X(X))\, \ell_{\delta}(X,c)], 
\]
where $d_{\tau}$ is the density associated to a cumulative distribution function $D_{\tau}$, such  that $D_{\tau}(0)=0$ and $D_{\tau}(1)=1$.  
As shown in Section \ref{sec:argmin:Exp}, this minimization problem can also be looked upon as $
\argmin_{c \in \mathbb{R}} \mathbb{E}[\ell_{\delta}(X_{\sD_{\tau}},c)] $, where $X_{\sD_{\tau}}$ denotes the random variable which has as  cumulative distribution  function the transformed function $D_{\tau}\circ F_X$, and the expectation is taken with respect to this distributional distorted (shortly distorted) random variable $X_{\sD_{\tau}}$. 
Moreover, with $\ell_{\delta}(x,c)=(x-c)^2$, i.e. square loss, this leads to 
$\argmin_{c \in \mathbb{R}}  \mathbb{E}[\ell_{\delta}(X_{\sD_{\tau}},c)] 
=  \mathbb{E} \left [ X_{\sD_{\tau}} \right ]$ and subsequently  $\min_{c \in \mathbb{R}} \mathbb{E}[\ell_{\delta}(X_{\sD_{\tau}},c)]= \mbox{Var} (X_{\sD_{\tau}})$.

In the special case of square loss, the functional thus relates to the mean of the transformed random variable $X_{\sD_{\tau}}$. This is linked  to the probability distortions studied by \cite{LiuSchiedWang2021}, as a particular setting of distributional transforms. The paper investigates probabilistic properties of these distributional transforms, and discusses (new) risk measures generated from such distributional transforms. In Section \ref{sec:risks of distorted rvs} we rely on some results of \cite{LiuSchiedWang2021} to study risk measures of the distorted random variable $X_{\sD_{\tau}}$.
Distortion risk measures have found widespread application in the literature, for example, in finance, economics, insurance; see \cite{Wang1996}, \cite{WangEtAl1997}, \cite{FollmerAndSchied2002}, and \cite{ChernyAndMadan2009}, 
among others.
Our general functional involves a loss function $\ell_{\delta}(x,c)$. There are numerous examples of loss functions appearing in the literature. Tables  \ref{LossFunctions1A} and \ref{LossFunctions2A}, as well as Section \ref{sec: ExamplesLossFunctions} review some of these.  

As our main contribution we introduce and study a general class of functionals, which cover and generalize these in \eqref{Quantile1}, \eqref{Expectile}, \eqref{Quantile2}, and \eqref{Extremile}, as  special cases. This generalization allows an interesting link to distortion risks as well as to risks of non-distortion type. We call this general class of functionals \emph{generalized extremiles}. The main contribution of this paper consists of developing estimators for generalized extremiles, and establish consistency and asymptotic normality of these estimators, with explicit expression for the asymptotic variance. This general theory can then be applied to a variety of population target quantities. The asymptotic normality result allows to construct asymptotic confidence intervals of estimators in the general framework.

Specific optimization problems capture  a number of situations studied in actuarial science, financial mathematics or statistics which are a special case of our framework.
For example, our framework applies if a functional $\rho$ can be defined as the expected loss minimizer
\begin{align}\label{eq:argmin:Exp}
	\rho(Y) = \argmin_{c\in \mathbb{R}} \,  \Exp[\loss(Y,c)]
\end{align}
for a suitable loss function $\loss \colon \mathbb{R}\times \mathbb{R}\to\mathbb{R}$ (where we surpressed the index $\delta$), and a random variable $Y$. 
In this case $\rho(X_{\sD_{\tau}})$ can be the mean, quantile or expectile of the distorted random variable $X_{\sD_{\tau}}$. 
Equally, one could for example consider functionals defined as the expected loss minimum
\begin{align}\label{eq:min:Exp}
	\rho(Y) = \min_{c\in\mathbb{R}} \Exp[\loss(Y,c)],
\end{align}
in which case $\rho(X_{\sD_{\tau}})$ can for example be the variance or expected shortfall of the distorted random variable $X_{\sD_{\tau}}$. See  \cite{EmbrechtsMaoWangWang2021} for details and a discussion 
connecting risk measures of the form in \eqref{eq:argmin:Exp} and \eqref{eq:min:Exp}.

While we present some general theory for risk functionals applied to distorted random variables $X_{\sD_{\tau}}$ in Section~\ref{sec:risks of distorted rvs}, our study mainly focusses on the specific case of expected loss minimizers as in \eqref{eq:argmin:Exp},  due to their ubiquitous nature in the literature.

The paper is organized as follows. In Section \ref{sec:argmin:Exp} we introduce the generalized extremile functionals, establish some basic properties, and explain the link with distortion risks. Section \ref{sec:estimation} deals with estimation of the broad class of generalized extremiles, and proves consistency  and asymptotic normality results. These are then applied to various settings, revealing the breath of the  studied functionals. Numerical studies and real data analysis are given in Sections \ref{sec: simulations} and \ref{sec: realdata}. The proofs of the main theoretical results are provided in the Appendices. Some further discussions are given in a final section.
Some additional information and explanations are in the Supplementary Material.

\section{Generalized extremiles and their properties }\label{sec:argmin:Exp}
In this section we introduce the concept of generalized extremile, study its basic probabilistic properties, highlight some issues toward solving the associated optimization problem, and review and discuss some loss functions. In a final subsection we discuss how generalized extremiles link up to distortion risk measures. In Section \ref{sec:estimation} we then turn to statistical estimation of generalized extremiles.

\subsection{Generalized extremiles and distorted random variables}
The class of functionals that we study is based on a pair of key ingredients,
denoted by $(D_{\btau}, \ell_{\delta})$, consisting of a distribution function $D_{\btau}$ and a loss function $ \ell_{\delta}$. Their formal definitions are provided below. 
\begin{definition}\label{def: LossFunction}
Consider a set $\Delta \subset \mathbb{R}$. For every $\delta \in \Delta$ we denote by $\ell_{\delta}$ the loss function 
$$\ell_{\delta}(x,c) : \mathbb{R} \times \mathbb{R} \to \mathbb{R}: (x,c)\mapsto \ell_{\delta}(x,c).$$
\end{definition}
\noindent
Where appropriate we use $\ell_{\delta}^{\prime}(x,c):=\frac{\mathrm{d}}{\mathrm{d} c}\ell_{\delta}(x,c)$ to denote the partial derivative of $\ell_{\delta}$ with respect to $c$.

\begin{definition}\label{def: DistribFunction}
    Let $\Tau$ be some subset  of $\mathbb{R}^k$ where $k$ is some positive integer. For each $\btau \in \Tau$ denote by $D_{\btau}$ an absolutely continuous cumulative distribution function supported on $[0,1]$, with density $d_{\btau}$. We denote by $\mathbb{D}$ for the set of all such $D_{\btau}$.
\end{definition}

\begin{definition}[Generalized extremile] \label{def: generalized extremiles}
Consider a cumulative distribution function $D_{\btau} \in \mathbb{D}$, and a loss function $\ell_{\delta}$. For a random variable $X\sim F_X$ we define the generalized extremile  of index $(\btau,\delta)$, denoted $e_{\btau,\delta}(X;D_{\btau},\ell_{\delta})$, as 
\begin{equation}
e_{\btau,\delta}(X;D_{\btau},\ell_{\delta}) \in \argmin_{c \in \mathbb{R}} \mathbb{E}[d_{\btau}(F_X(X))(\ell_{\delta}(X,c)-\ell_{\delta}(X,0))].
\label{GenExtremiles}  
\end{equation}
\end{definition}

Taking $D_{\tau}$, with $\btau=\tau 
\in (0,1)$,  equal to
\begin{equation}
K_{\tau}(u)=\begin{cases}
        1-(1-u)^{s(\tau)} &\text{ if } 0<\tau <  \frac{1}{2}\\
        u^{r(\tau)} &\text{ if } \frac{1}{2}\leq \tau<1 , 
    \end{cases} \label{DtauExtremiles}  
\end{equation}
with $r(\tau)=s(1-\tau)=\ln(1/2)/ \ln(\tau)$, in minimization problem  \eqref{GenExtremiles},  leads on the one hand to the concept of $\tau$th quantiles when taking as loss function $\ell_{\delta}(x,c)=|x-c|$ (see \eqref{Quantile2}), and on the other hand to the concept of $\tau$th extremile with squared loss  
$\ell_{\delta}(x,c)=(x-c)^2$ (see  \eqref{Extremile}). We  refer to the class of functionals in Definition \ref{def: generalized extremiles} as \textit{generalized extremiles}, since it steps away from the specific choice of 
the pair $(D_{\btau}, \ell_{\delta})$ in quantiles and extremiles, but having its inspiration strongly influenced by these two concepts. 

Note that the choice of the distribution function $D_{\tau}$ in \eqref{DtauExtremiles} already reveals why we allow  this distribution function to depend on a parameter (vector) $\btau$. From the viewpoint on quantiles expressed in \eqref{Quantile1} it is also clear why we allow the loss function to depend on a parameter $\delta$. Figure \ref{fig: KandJFunctionPlots} depicts the function $K_{\tau}$ and its derivative $d_{\tau}$ for various values of $\tau$.
\begin{figure}[htbp]
    \centering
    \includegraphics[width=\textwidth]{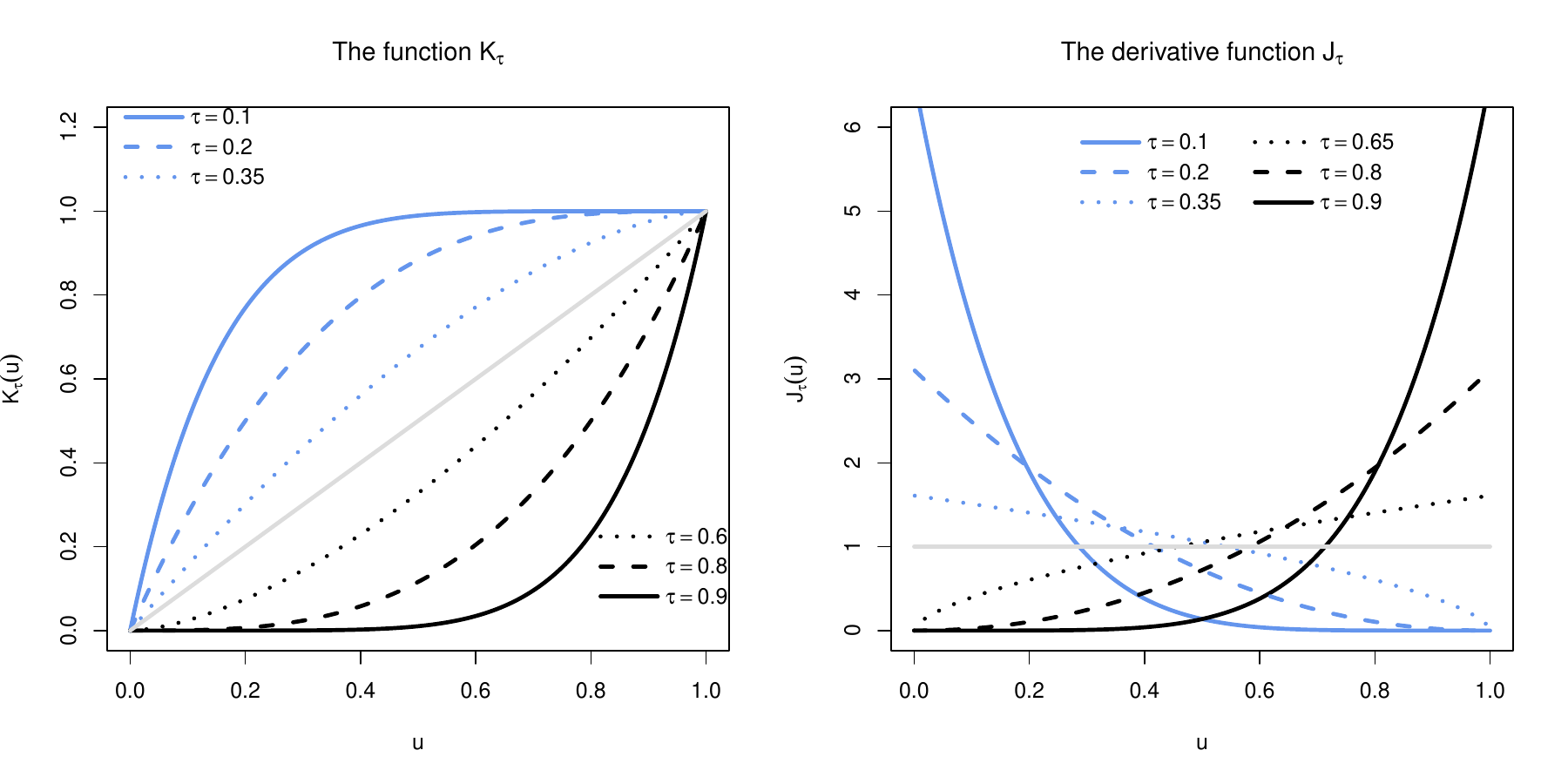}
    \vspace*{-1.2 cm}
    
    \noindent
    \caption{The functions $K_{\tau}$ in \eqref{DtauExtremiles} and its derivative $d_{\tau}$ for various values of $\tau$.}
    \label{fig: KandJFunctionPlots}
\end{figure}

\begin{table}[htb]
\caption{Commonly used nonnegative loss functions and their basic properties.}
\label{LossFunctions1A}
\vspace*{0.2 cm}

\noindent
\begin{tabular}{l|l|l|l|l|}
\hline
 & loss function  &  sign symmetric & shift invariant & positive homogeneous  \\ 
 & &  & & of degree $k$  \\[1.2 ex]
 & $\ell_{\delta}(x, c)$ &  $\ell_{\delta}(-x, c)= \ell_{\delta}(x, -c)$
 & $\ell_{\delta}(x+b, c+b)$
 & $\ell_{\delta}(ax, ac)= a^k\ell_{\delta}(x, c)$
 \\
 & & & $\quad = \ell_{\delta}(x, c)$ & \\
 \hline
 & & & & \\
 1 & $|x-c|$ & & &   \\
  & absolute value loss & yes & yes & yes, of degree 1   \\[1.6 ex] 
 2 & $|x-c|^{p}$ with $p>0$ & & &   \\
  & power $p$ absolute  loss & yes & yes & yes, of degree $p$   \\[1.6 ex]   
  3 & $ \left|\delta-1_{\{x\le c\}}\right| |x-c| $ & & &  \\
  & quantile  loss & no & yes & yes, of degree 1   \\[1.6 ex]
  4 & $ \left|\delta-1_{\{x\le c\}}\right| |x-c|^2  $ & & &   \\
  & expectile  loss & no & yes &  yes, of degree 2 \\[1.6 ex]
   5 & $\ell_{\delta}(x,c)$  &  yes & yes  &  no \\
 &  $=\left \{
  \begin{array}{lll}
  0.5 (x-c)^2 & \mbox{if} & |x-c|\le \delta \\
  \delta ( |x-c| - 0.5 \delta)& \mbox{if} & |x-c|> \delta 
  \end{array}
  \right . $ & & &\\
  & Huber loss function & &  &  \\[1.6 ex]
  6 & $\ell_\delta(x,c)=(c-x)^2 e^{\delta x}$ & no & no & no \\
  & Esscher loss (see \cite{VANHEERWAARDEN1989261})& & &  \\[1.6 ex]
 \hline 
\end{tabular}
\end{table}

\begin{table}[h!]
\caption{Commonly used real-valued loss functions and their basic properties.}
\label{LossFunctions2A}
\vspace*{0.2 cm}

\noindent
\begin{tabular}{l|l|l|l|l|}
\hline
 & loss function  &  sign symmetric & shift invariant & positive homogeneous  \\ 
 & &  & & of degree $k$  \\[1.2 ex]
 & $\ell_{\delta}(x, c)$ &  $\ell_{\delta}(-x, c)= \ell_{\delta}(x, -c)$
 & $\ell_{\delta}(x+b, c+b)$
 & $\ell_{\delta}(ax, ac)= a^k\ell_{\delta}(x, c)$
 \\
 & & & $\quad = \ell_{\delta}(x, c)$ & \\
 \hline
 & & & & \\
 G1 & $\ell_{\delta}(x,c)=-c (x^2-x)+c^2/2$ & no & no & no   \\[1.6 ex] 
 G2 & $\ell_{\delta}(x,c)=-c \lvert x - b \rvert^{\delta} +c^2/2$,  & & &   \\
 & (see e.g.~\cite{Hurlimann2006}) & & &\\[1.2 ex]
 & \hspace*{2.2 cm}  with  $b \in \mathbb{R}$ & no & no & no \\[1.6 ex] 
  G3 & $\ell_{\delta}(x,c)=1_{\{c<x\}}(x-c)+cx$ & no & no & no  \\[1.6 ex]
  G4 & $\ell_{\delta}(x,c)=\frac{c^2}{2}-(1+\delta) cx $   & yes & no & yes, of degree 2  \\
  & (see \cite{Dickson2005}) & & & \\[1.6 ex]
 \hline 
\end{tabular}
\end{table}

Note that in Definition \ref{def: LossFunction} we allow a loss function to take values in  $\mathbb{R}$. Several commonly used  loss functions are restricted to take  values in  $[0,+\infty)$. Table \ref{LossFunctions1A} lists some of these.  Table \ref{LossFunctions2A} provides some examples of real-valued loss functions. Throughout the paper we allow for a general loss function, and state explicitly when we restrict it to be nonnegative. 

When defining generalized extremiles in Definition \ref{def: generalized extremiles} we require the distribution function $D_{\btau}$ to be absolutely continuous admitting a density function $d_{\btau}$. The latter can be relaxed, as we illustrate next.

Recall that $X_{\sD_{\btau}}$ denotes the random variable with cumulative distribution function $D_{\btau} \circ F_X$, i.e., $X_{\sD_{\btau}} \sim D_{\btau} \circ F_X$. 
\begin{proposition}\label{prop:equivalentDefinition}
Consider  $X_{\sD_{\btau}}  \sim D_{\btau} \circ F_X$. We then have  
$$\argmin_{c\in \mathbb{R}} \mathbb{E}[d_{\btau}(F_X(X))(\ell_{\delta}(X,c)-\ell_{\delta}(X,0))]=\argmin_{c\in \mathbb{R}}\ \mathbb{E}[\ell_{\delta}(X_{\sD_{\btau}},c)-\ell_{\delta} (X_{\sD_{\btau}},0)],$$
where the expectation on the right-hand side is taken with respect to the random variable $X_{\sD_{\btau}}$. 
\end{proposition}

\begin{proof}
The proof is straightforward by noting that 
\begin{align*}
    \mathbb{E}[d_{\btau}(F_X(X))(\ell_{\delta}(X,c)-\ell_{\delta}(X,0)]&=\int_{\mathbb{R}} d_{\btau}(F_X(x))(\ell_{\delta}(x,c)-\ell_{\delta}(x,0))\mathrm{d} F_X(x)\\
    &=\int_{\mathbb{R}} (\ell_{\delta}(x,c)-\ell_{\delta}(x,0))\mathrm{d} D_{\btau}(F_X(x))\\
    &=\mathbb{E}[\ell_{\delta}(X_{\sD_{\btau}},c)-\ell_{\delta}(X_{\sD_{\btau}},0)].
\end{align*}
\end{proof}
An alternative would be to define generalized extremiles as stipulated in 
Proposition \ref{prop:equivalentDefinition}, via 
\[
 \argmin_{c\in \mathbb{R}}\ \mathbb{E}[\ell_{\delta}(X_{\sD_{\btau}},c)-\ell_{\delta} (X_{\sD_{\btau}},0)] ,
\]
as such not requiring the existing of the density $d_{\btau}$. 
Although this would be a slightly broader definition, it would make it more tedious to establish (asymptotic) properties when dealing with statistical inference for the studied functionals. The latter is the main goal of the paper. 

In Tables \ref{LossFunctions1B} and \ref{LossFunctions2B} (see later) we indicate  for each of the  loss functions in Tables \ref{LossFunctions1A} and \ref{LossFunctions2A}  this alternative view, with focus on the loss function applied to the random variable $X_{\sD_{\btau}}$. 
Table \ref{DistrFunctions} lists a selection of distribution (distortion) functions (with their associated densities).

\begin{table}[h!]
\caption{Commonly used distribution functions/distortion functions and their basic properties.}
\label{DistrFunctions}
\vspace*{0.2 cm}

\noindent
\hspace*{-1.6 cm}
\begin{tabular}{l|l|l|l|}
\hline
 & distrib. function  or distortion function &  density & bounded density  \\
 & $D_{\btau}(u)\quad $ or $\quad g_{\btau}(u)=1-D_{\btau}(1-u)$ & $d_{\btau}(u)$ &  \\
 & \hspace*{2.2 cm} and $D_{\btau}(u)=1-g_{\tau}(1-u)$  & &  \\[1.2 ex]
 \hline
 & & &  \\
1 & $D_{\btau}(u)=u$ & $d_{\tau}(u)=1$ & yes   \\
& uniform distr. & &  \\[1.6 ex]
2 & $D_{\btau}(u)=K_{\tau}(u)$ as in \eqref{DtauExtremiles} & $d_{\tau}(u)=J_{\tau}(u)$ as in \eqref{densitytauExtremiles}   & yes   
\\[2 ex]
3 & $D_{\btau}(u) =[B(a,b)]^{-1} \int_{0}^{u} t^{a-1} (1-t)^{b-1} \mathrm{d}t  $ & $d_{\tau}(u) = [B(a,b)]^{-1} u^{a-1} (1-u)^{b-1}$ & possibly unbounded\\[1.2 ex]
&\hspace*{4.2 cm}  $ a, b >0 $ & & \hspace*{2 mm} in 0 and/or 1 \\
& Beta distr. & &  bounded if $a, b >1$ \\[2 ex]
4 & $D_{\btau}(u) =1-(1-u^a)^b \quad a, b >0 $ & $d_{\btau}(u)=ab u^{a-1}(1-u^a)^{b-1}$&  possibly unbounded  \\
& Kumaraswamy distr. & &  \hspace*{2 mm} in 0 and/or 1  \\
& & & bounded if $a, b >1$  \\[2 ex]
5& $g_{\tau}(u)= u (1-\tau)^{-1} 1_{\{ u \le 1 - \tau \}} + 1_{\{ u >1 - \tau \}} \quad \tau \in [0,1] $ & $d_{\tau}(u) = (1-\tau)^{-1}1_{\{ u > \tau \}}$ & yes, for $\tau \in [0, 1)$  \\
& or $D_{\tau} (u)=(1-\tau)^{-1}(u-\tau)1{\{u \ge \tau\}}$  & & \\
& \hspace*{0.4 cm} or shortly $D_{\tau} \sim \mbox{U}[\tau,1]$ & & \\
& expected shortfall & & \\[2 ex]
6 & $g_{\btau}(u)= \Phi \left (\Phi^{-1} (u) + \tau \right) \quad \tau \in \mathbb{R} $ & $d_{\btau}(u)= \ds \frac{\phi \left ( \Phi^{-1} (1-u)+ \tau \right )}{\phi \left ( \Phi^{-1} (1-u)\right )}$ & for $\tau=0$ on $[0, 1]$\\
& & & unbounded in 1 for $\tau>0$ \\
& Wang Transform risk measure  & & unbounded in 0 for $\tau<0$  \\[2 ex]
7 & $g_{\btau}(u)= u^{1 / \tau} \quad \tau \ge 1  $ & $d_{\tau}(u) = \tau^{-1} (1-u)^{\{\tau^{-1} -1\}}$&  \\
& Proportional Hazard transform  & & unbounded in 1 \\[2 ex]
8 & $D_{\tau}(u) = u^{\tau +1}$, with $\tau \in \mathbb{R}_+$    & $d_{\btau}(u)=  (\tau +1) u^{\tau} $ & bounded \\
& {\sc minvar} (see \cite{ChernyAndMadan2009})   & &   \\[2 ex]
9 & $D_{\tau}(u) = 1 -(1-u)^{1/(\tau +1)}$, with $\tau \in \mathbb{R}_+$ &
$d_{\btau}(u)=  (\tau +1)^{-1} (1-u)^{- \tau/(\tau +1)}$  & unbounded in 1,  \\
& {\sc maxvar}  (see \cite{ChernyAndMadan2009}) & & unless $\tau=0$   \\
& & & \hspace*{2 mm} (but then $\mbox{U}[0,1]$)  \\[2 ex]
10 &  $D_{\tau}(u) = \left (1 -(1-u)^{1/(\tau +1)}\right )^{\tau+1}$, with $\tau \in \mathbb{R}_+$ & $d_{\btau}(u) $   & unbounded in 1  \\
& {\sc minmaxvar} (see \cite{ChernyAndMadan2009})  & \hspace*{2 mm} $= \left (1 -(1-u)^{1/(\tau +1)}\right )^{\tau}
(1-u)^{-\tau/(\tau+1)}$  &   \\[2 ex]
11 &$D_{\tau}(u) =1 -\left (1-u^{\tau +1}\right )^{1/(\tau +1)}$, with $\tau \in \mathbb{R}_+$& $d_{\tau}(u) =\left (1-u^{\tau +1}\right )^{-\tau/(\tau +1)}u^{\tau}$ & unbounded in 1   \\
& {\sc maxminvar} (see \cite{ChernyAndMadan2009})  & &  \\[2 ex]
12 & $ g_{\btau}(u)= G\left ( G^{-1}(u) + \tau \right ) \quad \tau \ge 0   $ &
$d_{\tau}(u) =\ds \frac{ g\left ( G^{-1} (1-u) + \tau \right )}{ g\left ( G^{-1} (1-u) \right )}$  & yes, on $(0,1)$   \\
& $G$ cumul. distr. function, with $G'=g$ log-concave & & possibly unbounded \\
& Junike's distortion function (see \cite{Junike2019})   & &  in 0 or 1, depending on $g$ \\[2 ex]
13& $g_{\tau}(u)= \tau^{-1} C(u, \tau)$ $ \quad \tau \in (0, 1]$ & 
$d_{\tau}(u) = \tau^{-1} \frac{\partial C}{\partial u} (1-u; \tau) $ & yes  \\
& $C(\cdot, \cdot)$ bivariate copula  & &  \\
& (see \cite{YinAndZhu2018}) & & \\[2 ex]
 \hline 
\end{tabular}
\end{table}

\subsection{Generalized extremiles: properties}\label{sec: GenExtremilesProperties}
Before stating some probabilistic properties of generalized extremiles, we first define some properties of loss functions, as well as the notion of a dual of a cumulative distribution function $D_{\btau}$. 

\begin{definition}\label{def: propertiesLossFunction}
   Consider a loss function $\ell_{\delta}: \mathbb{R}\times \mathbb{R}  \to \mathbb{R}$. It is called 
   \begin{enumerate}
       \item sign symmetric if $\ell_{\delta}(-x,c)=\ell_{\delta}(x,-c) \text{ for all } x,c \in \mathbb{R}$;
        \item shift invariant if $\ell_{\delta}(x+b,c+b)=\ell_{\delta}(x,c)$ for all $x,c,b \in \mathbb{R}$;
        \item positive homogeneous of degree $k>0$ if
        $\ell_{\delta}(ax,ac)=a^k \ell_{\delta}(x,c)$ for all $a>0$ and $x,c \in \mathbb{R}$.
   \end{enumerate}
\end{definition}

In Tables \ref{LossFunctions1A} and \ref{LossFunctions2A} we indicate for each of the listed loss functions whether they possess the property of sign symmetry, shift invariance  and positive homogeneity.\\

For a given cumulative distribution function $D_{\btau}$ in the set $\mathbb{D}$ it is of interest to consider its dual. 
\begin{definition}\label{def: DualDistribFunction} 
    Given $D_{\btau} \in \mathbb{D}$ we define its dual $\widetilde{D}_{\btau}$ as $\widetilde{D}_{\btau}(t)=1-D_{\btau}(1-t)$.
\end{definition}
It is clear that the dual function is again part of the set $\mathbb{D}$. Furthermore, the dual of $\widetilde{D}_{\btau}$ is equal to $D_{\btau}$. \\

The proofs of all propositions in this section can be found in 
 Appendix \ref{App: ProofsPropositions}.
Proposition \ref{prop: finiteness} provides sufficient conditions under which the expectation in Definition \ref{def: generalized extremiles} is finite. The proof is straightforward and omitted.  
\begin{proposition}\label{prop: finiteness}
    If the cumulative distribution function $D_{\btau}$ is differentiable with bounded density $d_{\btau}$ and $\mathbb{E}[\lvert \ell_{\delta}(X,c)-\ell_{\delta}(X,0)) \rvert ]<\infty$ for all $c\in \mathbb{R}$ then 
    $$\lvert \mathbb{E}[d_{\btau}(F_X(X))(\ell_{\delta}(X,c)-\ell_{\delta}(X,0))] \rvert < \infty $$
    for all $c \in \mathbb{R}$.
\end{proposition}

\begin{remark}\label{Remark1}  The condition on the loss function (and $X$) in Proposition \ref{prop: finiteness} is very mild, and fulfilled for several commonly used loss functions. In case of the absolute value loss $\ell(x, c)=\lvert x-c\rvert $ the reverse triangle inequality yields $\lvert \lvert x-c \rvert-\lvert x\lvert \lvert  \leq \lvert c\rvert $ and hence $\mathbb{E}[\lvert \ell(X, c)-\ell(X, 0)\rvert ]<\infty$ for all random variables $X$. In case of the square loss $\ell(x, c)=(x-c)^2$ we have $\lvert(x-c)^2-x^2\rvert =\lvert c^2-2 c x\rvert \leq c^2+2 \lvert  c \rvert \cdot \lvert  x \rvert $ and hence $\mathbb{E}[\lvert \ell(X, c)-\ell(X, 0) \rvert ]<\infty$ for all random variables $X$ with $\mathbb{E}[\lvert X \rvert]< \infty$.
\end{remark}

From Remark \ref{Remark1} it also becomes clear why the term $\ell_{\delta}(X, 0)$ is included in \eqref{GenExtremiles}, whereas it does not depend on the argument $c$. Due to the inclusion of that term we have, for example, that quantiles obtained by taking $D_{\tau}=K_{\tau}$ and $\ell_{\delta}(x,c)=|x-c|$, exist for all random variables. Similarly, expectiles of $X$ only require the finiteness of the first absolute moment of $X$.\\

Propositions \ref{prop:signSymmetry}, \ref{prop:transformations} and \ref{prop:kTransformation} establish the behaviour of generalized extremiles when the underlying random variable is transformated either by taking the negative of it (Proposition \ref{prop:signSymmetry}), or by an affine transformation  (Proposition \ref{prop:transformations}), or by applying a strictly monotone function to it (Proposition  \ref{prop:kTransformation}). From these propositions we can conclude that generalized extremiles have similar properties as quantiles and expectiles, among others.

\begin{proposition}[Sign symmetry]\label{prop:signSymmetry}
Denote by $\ell_{\delta}$ a sign symmetric loss function. If the distribution function $F_X$ of $X$ is continuous then we have 
\[
e_{\btau,\delta}\left(X ; D_{\btau}, \ell_{\delta} \right)=-e_{\btau,\delta}(-X ; \widetilde{D}_{\btau}, \ell_{\delta} ) \qquad \text{ and hence also }
\qquad 
    e_{\btau,\delta}(X ; \widetilde{D}_{\btau}, \ell_{\delta} )=-e_{\btau,\delta}(-X ; D_{\btau}, \ell_{\delta}).
    \]
\end{proposition}

\begin{proposition}\label{prop:transformations}
Consider scalars $a>0$ and $b \in \mathbb{R}$, and a random variable $X \sim F_X$ such that $\lvert e_{\btau,\delta}(X;D_{\btau},\ell_{\delta}) \rvert < \infty $.   If the loss function $\ell_{\delta}$ is shift invariant, positive homogeneous of degree $k>0$ and if
\begin{equation}\label{eq:vglInTranslationInvarianceProp}
    \lvert \mathbb{E}[ d_{\btau}(F_X(X))(\ell_{\delta}(aX,0)-\ell_{\delta}(aX+b,0))  ]\rvert< \infty
\end{equation}
then
$$e_{\btau,\delta}(aX+b;D_{\btau},\ell_{\delta})=a e_{\btau,\delta}(X;D_{\btau},\ell_{\delta})+b.$$
If in addition the loss function $\ell_{\delta}$ is sign symmetric,   $F_X$ is continuous and  (\ref{eq:vglInTranslationInvarianceProp}) holds for $a<0$ and $b \in \mathbb{R}$ then
\[
  e_{\btau,\delta}(aX+b;D_{\btau},\ell_{\delta})=a e_{\btau,\delta}(X;\widetilde{D}_{\btau},\ell_{\delta})+b \qquad \text{ and hence also }\qquad  e_{\btau,\delta}(aX+b;\widetilde{D}_{\btau},\ell_{\delta})=a e_{\btau,\delta}(X;D_{\btau},\ell_{\delta})+b.
  \]  
\end{proposition}

\begin{proposition}[Monotone transformations]\label{prop:kTransformation}
    Suppose $\alpha:\mathbb{R}\to \mathbb{R}$ is a strictly increasing function and hence measurable function, then
    $$e_{\btau,\delta}(\alpha(X);D_{\btau},\ell_{\delta})=e_{\btau,\delta}(X;D_{\btau},\widetilde{\ell}_{\delta}),$$
    where $\widetilde{\ell}_{\delta}(x,c)=\ell_{\delta}(\alpha(x),c)$.
    If $\alpha:\mathbb{R}\to \mathbb{R}$ is strictly decreasing and $X$ is a random variable with continuous $F_X$  then
    $$e_{\btau,\delta}(\alpha(X);D_{\btau},\ell_{\delta})=e_{\btau,\delta}(X;\widetilde{D}_{\btau},\widetilde{\ell}_{\delta}).$$
\end{proposition}

When a random variable $X$ is symmetric and has finite expectation it is possible to express this mean in terms of differences or sums of generalized extremiles, as stated in Proposition \ref{prop:MeanSymmetry}. 

\begin{proposition}[Mean symmetry] \label{prop:MeanSymmetry} Denote by $\ell_{\delta}$ a shift invariant loss function. If $\mathbb{E}[X]$ is finite and $X$ is symmetric, i.e., $X \stackrel{d}{=}-X$, then
$$
\frac{1}{2}(e_{\btau,\delta}(X ; D_{\btau}, \ell_{\delta})-e_{\btau,\delta}(-X ; D_{\btau}, \ell_{\delta}))=\mathbb{E}[X]
$$
If in addition $\ell_{\delta}$ is a sign symmetric loss function and the distribution function $F_X$ of $X$ is continuous then
$$\frac{1}{2}\left( e_{\btau,\delta}(X;D_{\btau},\ell_{\delta}) + e_{\btau,\delta}(X;\widetilde{D}_{\btau},\ell_{\delta}) \right)=\mathbb{E}[X].$$
\end{proposition}
\medskip

For two random variables $X$ and $Y$, it is interesting to know how the generalized extremile of their sum, i.e., $e_{\btau,\delta}(X+Y ; D_{\btau}, \ell_{\delta})$, relates to the sum of the generalized extremiles of each of the random variables. Proposition \ref{prop:ComonotonicAdditivity} below provides some answer to this question, establishing the property of comonotonic additivity for the absolute value and square loss functions.

\begin{proposition}[Comonotonic additivity]\label{prop:ComonotonicAdditivity}
Denote by $X \sim F_X$ and $Y \sim F_Y$ two random variables that have a comonotonic dependence structure, i.e., $F_{X, Y}(x, y)=\min \{F_X(x), F_Y(y)\}$. If $\ell_{\delta}$ is either the absolute value loss or the square loss we have (assuming the considered generalized extremiles are finite) that 
$$
e_{\btau,\delta}(X+Y ; D_{\btau}, \ell_{\delta})=e_{\btau,\delta}(X ; D_{\btau}, \ell_{\delta})+e_{\btau,\delta}(Y ; D_{\btau}, \ell_{\delta}).
$$
\end{proposition}

A final result  states the equivalence of an injective transformation of a generalized extremile and a generalized extremile of the same random variable with a transformed loss function. 
It is similar to  results  in \cite[see Theorem 4]{Gneiting2011} and in \cite[see p. 9]{Osband1985} known as revelation principle. 
\begin{proposition}[Revelation principle]\label{prop:injectiveTransformation}
    Let $\beta:\mathbb{R}\to \mathbb{R}$ be an injective function. Define $\widetilde{\ell}_{\delta}(x,c)=\ell_{\delta}(x,\beta^{-1}(c))$, then
    $$\beta(e_{\btau,\delta}(X;D_{\btau},\ell_{\delta}))=e_{\btau,\delta}(X;D_{\btau},\widetilde{\ell}_{\delta}).$$
\end{proposition}
\medskip 

The result in Proposition \ref{prop:injectiveTransformation} has some interesting consequences. Consider the square loss function $\ell_{\delta}(x,c)=(x-c)^2$. Then we know that 
\[
e_{\btau,\delta}(X;D_{\btau},(x-c)^2)= \mathbb{E}[X_{\sD_{\btau}}].
\]
Applying Proposition \ref{prop:injectiveTransformation}  with injective function $\beta:\mathbb{R} \to \mathbb{R}: x \mapsto -x$, gives that 
\[ 
e_{\btau,\delta}(X;D_{\btau},(x+c)^2)= -\mathbb{E}[X_{\sD_{\btau}}].
\] 
\smallskip

\subsection{Towards solving the minimization problem}\label{sec: MinimizationProblem}
When it comes to finding a generalized extremile, defined in \eqref{GenExtremiles}, Proposition \ref{prop:equivalentDefinition} is a good starting point, showing that we need to find a minimizer of 
$\mathbb{E}[\ell_{\delta}(X_{\sD_{\btau}},c)-\ell_{\delta} (X_{\sD_{\btau}},0)],$ with respect to $c \in \mathbb{R}$. 
Proposition \ref{prop: convexLossDeriviativeZero} below  provides a way to find such a minimizer, and states sufficient conditions on the pair $(D_{\btau}, \ell_{\delta})$. 

\begin{proposition}\label{prop: convexLossDeriviativeZero}
Let $\ell_\delta$ be a function such that $c \mapsto \ell_\delta(x, c)$ is (strictly) convex for every $x$ and $c$. Additionally, assume that both $D_{\btau}$ and $F_X$ are differentiable functions. Under these conditions, the function $c \mapsto \mathbb{E}\left[\ell_\delta\left(X_{D_{\btau}}, c\right)-\ell_\delta\left(X_{D_{\btau}}, 0\right)\right]$ is (strictly) convex, and the generalized extremile is equal to any value $c$ for which
\begin{equation}\label{eq:uniqueness}
\frac{\mathrm{d}}{\mathrm{d} c}\left(\int_{-\infty}^{+\infty} \ell_\delta(x, c) \, \mathrm{d}  D_{\btau}\left(F_X(x)\right)\right)=0,
\end{equation}
provided the derivative exists. When the mapping $c\mapsto \ell_\delta(x,c)$ is strictly convex, there exists at most one $c$ value fulfilling (\ref{eq:uniqueness}). The existence of such a $c$ value is ensured when $c\mapsto \ell_\delta(x,c)$ is convex and $c \mapsto \mathbb{E}\left[\ell_\delta\left(X_{D_{\btau}}, c\right)-\ell_\delta\left(X_{D_{\btau}}, 0\right)\right]$ is coercive.
\end{proposition}
\begin{proof} 
To solve the minimization problem we want to find the value $c$ that minimizes
$$
I(c)=\mathbb{E}\left[\ell_\delta\left(X_{D_\tau}, c\right)\right]=\int_{-\infty}^{+\infty} \ell_\delta(x, c) d_{\boldsymbol{\tau}}\left(F_X(x)\right) f_X(x) \mathrm{d} x .
$$
By \cite{BoydAndVandenberghe2004}, p. 79, we know that $I(c)$ is convex in $c$ if $c \mapsto \ell_\delta(x, c)$ is convex for every $x$.
It is clear that strict convexity of $c\mapsto \ell_\delta(x,c)$ implies that $I(c)$ is also strict convex and hence the minimum is unique (see Theorem 3.4.4 on p. 114 in \cite{niculescuEtal2006}).\\
The existence of a value $c$ follows by Theorem 4.4 on p. 638 in \cite{HerrmannEtal2018}.
\end{proof}

In Proposition~\ref{prop: convexLossDeriviativeZero} uniqueness of  $c$ solving \eqref{eq:uniqueness} is guaranteed when $c \mapsto \mathbb{E}\left[\ell_\delta\left(X_{D_{\btau}}, c\right)-\ell_\delta\left(X_{D_{\btau}}, 0\right)\right]$ is a coercive function. Some sufficient conditions for this to hold are discussed in the following remark.

\begin{remark}\label{remark:SufficientConditionsCoercive}
Consider a random variable $Y$ (with associated probability measure $\mathbb{P}_{Y}$) such that $E_0=\Exp[\ell_{\delta}(Y,0)]$ is finite.
Additional to convexity, assume that $c\mapsto\ell_{\delta}(y,c)$ is (i) coercive for $\mathbb{P}_{Y}$ almost all $y\in\mathbb{R}$, i.e.,
$ \lim_{c_n \to \pm \infty }  \ell_{\delta}(y,c_n) = \infty$, and (ii) bounded from below, i.e.,  $\exists \, C>0$, such that $\ell_{\delta}(y,c) \geq C$ $\mathbb{P}_Y$-a.s.
Then also $c\mapsto \Exp[\ell_{\delta}(Y,c)-\ell_{\delta}(Y,0)]$ is coercive. 
To see this consider a sequence $(c_n)$ such that $\lvert c_n \rvert \to\infty$ and define $Y_n:=\ell_{\delta}(Y,c_n)-C$ to have $Y_n\geq 0$ $\mathbb{P}_Y$-a.s. and $Y_n\to\infty$ due to the coercivity of $\ell_{\delta}$.
A slight modification of the monotone convergence theorem, see \cite{Feinstein2007}, then implies $\lim_{n\to\infty} \Exp[\ell_{\delta}(Y,c_n)-\ell_{\delta}(Y,0)] = \lim_{n\to\infty} \Exp[\ell_{\delta}(Y,c_n)-C] + C-E_0 = \Exp[\lim_{n\to\infty}\ell_{\delta}(Y,c_n)-C] + C-E_0= \infty$, showing that $c\mapsto \Exp[\ell_{\delta}(Y,c)-\ell_{\delta}(Y,0)]$ is coercive.
The result is now applicable in the context of Proposition~\ref{prop: convexLossDeriviativeZero} if the necessary conditions apply for the choice of $Y=X_{D_{\tau}}$.
\end{remark}
\smallskip

As is clear from the above result, convexity of the function $c\mapsto \ell_{\delta}(x,c)$ will be an important requirement when it comes to calculating generalized extremiles. In addition it also makes it transparant that interchanging integral and differentiation operators will be helpful. For any random variable $X$ (not necessarily absolutely continuous) and any distribution function $D_{\btau}$ (not necessarily admitting a density) the next theorem (stating a type of Leibniz rule) is helpful.  

\begin{theorem}[see \cite{Klenke2013}, Theorem 6.28, p. 142]
    Denote $\mu$ the measure induced by $D_{\btau} \circ F_X$. Suppose that a loss function $\ell_{\delta}(x,c)$ satisfies the following properties.
    \begin{itemize}
        \item For every $c$, the map $x\mapsto \ell_{\delta}(x,c)$ is integrable (with respect to the measure $\mu$);
        \item For $\mu$-almost all $x$, the map $c\mapsto \ell_{\delta}(x,c)$ is differentiable with derivative $\ell_{\delta}'(x,c)$;
        \item There is an integrable function $h: \mathbb{R} \to (0, \infty)$ (with respect to the measure $\mu$) such that $\lvert \ell_{\delta}'(x,c) \rvert \leq h(x)$ for all $c$ and $\mu$-almost all $x$.
    \end{itemize}
    Then for all $c \in \mathbb{R}$ we have
    $$\frac{\mathrm{d}}{\mathrm{d}c}\Big( \int_{-\infty}^{+\infty} \ell_{\delta}(x,c)\mathrm{d} D_{\btau}(F_X(x)) \Big)= \int_{-\infty}^{+\infty} \ell_{\delta}'(x,c) \mathrm{d}D_{\btau}(F_X(x)).$$
\end{theorem}
\medskip 

For absolutely continuous $X$ with density function $f_X$, and $D_{\btau}$ admitting a density $d_{\btau}$, the adaptation of this theorem is as follows. 
\begin{theorem}\label{theorem:leibnizAbsolutelyContinuous}
    Consider an absolutely continuous random variable $X$, and a cumulative distribution function $D_{\btau} \in \mathbb{D}$ with density $d_{\btau}$. Suppose that a loss function $\ell_{\delta}(x,c)$ satisfies the following properties.
    \begin{itemize}
        \item For every $c$, the map $x\mapsto \ell_{\delta}(x,c) d_{\btau}(F_X(x))f_X(x)$ is integrable;
        \item For almost all $x$, the map $c\mapsto \ell_{\delta}(x,c)$ is differentiable with derivative $\ell_{\delta}'(x,c)$;
        \item There is an integrable function $h: \mathbb{R} \to (0, \infty)$ (with respect to the measure $\mu$) such that \newline  $\lvert \ell_{\delta}'(x,c)d_{\btau}(F_X(x))f_X(x) \rvert \leq h(x)$ for all $c$ and almost all $x$.
    \end{itemize}
Then for all $c \in \mathbb{R}$ we have
    $$\frac{\mathrm{d}}{\mathrm{d}c}\Big( \int_{-\infty}^{+\infty} \ell_{\delta}(x,c)d_{\btau}(F_X(x))f_X(x)\mathrm{d} x \Big)= \int_{-\infty}^{+\infty} \ell_{\delta}'(x,c) d_{\btau}(F_X(x))f_X(x)\mathrm{d}x.$$
\end{theorem}

In Tables \ref{LossFunctions1B} and \ref{LossFunctions2B} we return to the loss functions listed in respectively Tables \ref{LossFunctions1A} and \ref{LossFunctions2A}, and now focus on the properties of convexity and differentiability of these loss functions. 

\begin{table}[h]
\caption{Commonly used nonnegative loss functions and basic properties (convexity, differentiability, ... ).}
\label{LossFunctions1B}
\vspace*{0.2 cm}

\noindent
\hspace*{-2 cm}
\begin{tabular}{l|l|l|l|l|l|}
\hline
 & loss function  &  alternative view  & convex & differentiability & derivative \\   
 & & from $X_{\sD_{\btau}}$ & in $c$ & in $c$ &  with respect to $c$\\
 \hline
 & & & & &\\
 1 & $|x-c|$ & $\mbox{Median}(X_{\sD_{\btau}})$& yes & yes, for $c\neq x$ &  $\ell'_{\delta}(x,c)=1_{\{x\leq c\}}-1_{\{x> c\}}$\\
  & absolute value loss & & & &   \\[1.6 ex] 
 2 & $|x-c|^{p}$ with $p>0$ & $\mathbb{E} [X_{\sD_{\btau}}]$ for $p=2$  & yes,  & depends on the &\\
  & power $p$ absolute  loss & &if $p \ge 1$ &  value of $p$&  \\[1.6 ex]   
  3 & $ \left|\delta-1_{\{x\le c\}}\right| |x-c| $ & $\delta$th quantile of $X_{\sD_{\btau}}$ & yes & yes, for $c\neq x$ & $\ell'_{\delta}(x,c)=1_{\{x \le  c\}} -\delta$ \\
  & quantile  loss &  & & &    \\[1.6 ex]
  4 & $ \left|\delta-1_{\{x\le c\}}\right| |x-c|^2  $ & $\delta$th expectile of $X_{\sD_{\btau}}$ & yes & yes, $\forall c$ &  $\ell'_{\delta}(x,c)$ \\
  & expectile  loss &  & & &$\quad =-2 \delta \lvert x-c\rvert -2 (x-c) 1_{\{ x\leq c\}}$  \\[1.6 ex]
  5 & $\ell_{\delta}(x,c)$  & ${}^{(\star 1)}$ & yes & yes, for $c$  &  $\ell'_{\delta}(x,c)=-(x-c)1_{\{|x-c| \le \delta\}}$ \\
 &  $=\left \{
  \begin{array}{ll}
  0.5 (x-c)^2 & \mbox{if }  |x-c|\le \delta \\
  \delta ( |x-c| - 0.5 \delta)& \mbox{if }  |x-c|> \delta 
  \end{array}
  \right . $ & & &\hspace*{1.2 mm} such that  &$ \hspace*{1.6 cm}  -\mbox{sign}(x-c) \delta \, 1_{\{|x-c| > \delta\}}$ \\
  & Huber loss function & &  & \hspace*{1.2 mm}  $|x-c|\neq \delta$   & \\[1.6 ex]
  6 & $\ell_\delta(x,c)=(c-x)^2 e^{\delta x}$ & $\ds \frac{\mathbb{E}[X_{\sD_{\btau}} e^{\delta X_{\sD_{\btau}}}]}{\mathbb{E}[e^{\delta X_{\sD_{\btau}}}]}$ & yes & yes, $\forall c$ & $\ell_\delta'(x,c)=2(c-x)e^{\delta x}$\\
  & Esscher loss & & & & \\[1.6 ex]
 \hline 
\end{tabular}
\vspace*{0.12 cm}

\noindent ${}^{(\star 1)}$: Value c for which
$c=\frac{1}{\mathbb{E}[1_{\{\lvert X_{\sD_{\btau}}-c \rvert \leq \delta\}}]}\Big( \mathbb{E}[X_{\sD_{\btau}} 1_{\{\lvert X_{\sD_{\btau}}-c \rvert \leq \delta \}}] - \delta \mathbb{E}[1_{\{X_{\sD_{\btau}} -c\leq -\delta \}}]+\delta \mathbb{E}[1_{\{X_{\sD_{\btau}}-c \geq \delta  \}}] \Big)$).
\end{table}
\begin{table}[h!]
\caption{Commonly used real-valued loss functions and basic properties (convexity, differentiability, ....).}
\label{LossFunctions2B}
\vspace*{0.2 cm}

\noindent
\begin{tabular}{l|l|l|l|l|l|}
\hline
 & loss function  &  alternative view  & convex & differentiability & derivative \\   
 & &  from $X_{\sD_{\btau}}$ & in $c$ & in $c$ & with respect to $c$ \\
 \hline
 & & & & &\\
  G1 & $\ell_{\delta}(x,c)=-c (x^2-x)+c^2/2$ & $\mathbb{E}[(X_{\sD_{\btau}})^2-X_{\sD_{\btau}}]$ & yes & yes, for all $c\neq x$&  $\ell'_\delta(x,c)=-(x^2-x)+c$   \\[2 ex] 
 G2 & $\ell_{\delta}(x,c)=-c \lvert x - b \rvert^{\delta} +c^2/2$,  & $\mathbb{E}[\lvert X_{\sD_{\btau}} - b\rvert^\delta]$ & yes & yes, $\forall c$ & $\ell_\delta'(x,c)=-\lvert x - b \rvert^\delta + c$ \\[1.2 ex]
 & \hspace*{2.2 cm}  with  $b \in \mathbb{R}$ & & & & \\[2 ex] 
  G3 & $\ell_{\delta}(x,c)=1_{\{c<x\}}(x-c)+cx$ & value $c$ for which & yes & yes, for all $c\neq x$ & $\ell'_\delta(x,c)=-1_{\{c<x\}}+x$\\
  && $P(X_{\sD_{\btau}} \geq c)=\mathbb{E}[X_{\sD_{\btau}}]$ & & & \\[2 ex]
  G4 & $\ell_{\delta}(x,c)=\frac{c^2}{2}-(1+\delta) cx $ & $(1+\delta)\mathbb{E}[X_{\sD_{\btau}}]$ & yes & yes, $\forall c$ &  $\ell_\delta'(x,c)=c-(1+\delta)x$ \\
 & & expected value & & & \\
 & & premium principle  & & & \\[2 ex]
 \hline 
  \hline 
\end{tabular}
\end{table}

\subsection{Examples}\label{sec: ExamplesLossFunctions}
Complementary to the loss functions listed in Tables \ref{LossFunctions1A} and \ref{LossFunctions2A}
we review here some loss functions that have been proposed in recent literature, illustrating as such the broad framework we consider. 
Some more loss functions are given in Section S1 in the Supplementary Material
(non exhaustive list). 
 
\subsubsection{Nonnegative valued loss functions}
\textbf{Adaptive quantile loss function for censoring}\\
In the context of quantile estimation under right random censoring,
\cite{DeBackerEtAL2019} established an approach based on considering a loss function that adapts to the cumulative  distribution function $F_C$ of the censoring variable $C$: 
\begin{equation}\label{eq: QuantileLossCensoring}
\ell_{\delta}(x,c) = 
\left|\delta-1_{\{x\le c\}}\right| |x-c|  - (1 - \delta) \ds \int_{0}^c F_C(s)\mathrm{d} s . 
\end{equation}
Note that the mapping  $c \mapsto\ell_{\delta}(x,c)$ is in general not convex, nor is it monotone. In case of no censoring, i.e. when $P\{ C < \infty\}=0$, the adapted loss function reduces to the usual quantile loss function. 
\medskip

\noindent
\textbf{M-quantile loss}\\
For a  function $\psi:\mathbb{R} \to [0,+\infty)$, consider the loss function 
\[
\ell_{\delta}(x,c)=\delta \psi(x-c)1_{\{x>c\}}+(1-\delta)\psi(x-c)1_{\{x\leq c\}}. 
\]
This is the M-quantile loss function. Taking $\psi(u)=|u|^p$ leads to the  power $p$ absolute  loss function, also called $L_p-$quantiles as defined in \cite{Chen1996}. See also \cite{JiangEtAl2021}. Obviously, the qualitative properties of $\ell_{\delta}(x,c)$ depend on the qualitative properties of the function $\psi$. 
\medskip

\noindent
\noindent
\textbf{Generalized quantile loss}\\
An extension of the M-quantiles are obtained by involving two different functions $\psi_1$ and $\psi_2$. Namely consider $\psi_1,\psi_2: [0,+\infty)\to [0,+\infty)$, convex and  strictly increasing functions satisfying  
$
\psi_j(0)=0 \text{ and } \psi_j(1)=1 \text{ for } j=1,2.
$
The generalized quantile loss function is then
\[
\ell_{\delta}(x,c)=\delta \psi_1(x-c)1_{\{x>c\}}+(1-\delta)\psi_2(x-c)1_{\{x\leq c\}},
\] 
leading to the $\delta$th generalized quantile of a random variable. 
See \cite{BelliniEtAl2014}.   
The qualitative properties of the generalized quantile loss depend on these of the functions $\psi_1$ and $\psi_2$.

    In \cite{MaoCai2018}, the authors further generalized this class by using rank-dependent expected utility (RDEU) theory. This corresponding class of generalized quantiles based on RDEU theory is given by 
    \begin{equation}\label{eq:generalizedQuantilesRDUE}
       \argmin_{c \in \mathbb{R}} \Big\{ \delta \int_c^{+\infty} \phi_1(x-c)dh_1(F_X(x))+(1-\delta)\int_{-\infty}^c \phi_2(c-x)dh_2(F_X(x))   \Big\}. 
    \end{equation}
    The functions $h_1,h_2$ are two distortion functions on $[0,1]$, which means they are right-continuous and increasing with $h_j(0)=0$ and $h_j(1)=1$ for $j=1,2$, with no jumps at $0$ and $1$. Furthermore, $\phi_1$ and $\phi_2$ are nondegenerate increasing convex functions on $[0,\infty)$. When $h_1=h_2:=D_{\btau}$ we can describe (part of) this class using the generalized extremile framework. Indeed, to do so, we have to take
    $$D_{\btau}=h_1=h_2 \text{ and } \ell_{\delta}(x,c)=\delta \phi_1((x-c)_+)+(1-\delta)\phi_2((c-x)_-).$$
    Here we used the notations:  $x_+=\max\{x,0\}$ and $x_-=-\min\{x,0\}$.
\medskip

\noindent
\textbf{Generalized shortfall loss}\\
    In \cite{MaoCai2018} and \cite{MaoEtAl2023}, the authors study generalized shortfall risk measures. These are obtained by replacing $\phi_1,\phi_2$ with their derivatives $\phi_1',\phi_2'$ in (\ref{eq:generalizedQuantilesRDUE}). Just as before, part of these generalized shortfall risk measures can be described in our generalized extremile framework when $h_1=h_2:=D_{\btau}$. 
\medskip

\noindent
\textbf{Loss function for trimmed mean}\\
The Hampel family of functions (see \cite{Hampel1986}) is defined as 
$$\ell(x,c)=\begin{cases}
        \frac{1}{2}(x-c)^2 &\text{ if } \lvert x-c\rvert <  a\\
        a(x-c-a/2)  &\text{ if } a \le \lvert x-c \rvert < b \\
    \ds\frac{a(x-c-d)^2}{2(b-d)} + a \frac{(b+d-a)}{2}      &\text{ if } b \le \lvert x-c \rvert < d \\
    \frac{1}{2} a(b+d-a) &\text{ if } \lvert x-c \rvert \ge d , \\
    \end{cases} $$
    where the function $\ell$ depends on parameters $a,b$ and $d$, with $0<a< b<d$. This is a non-convex function.
    
A limiting case of this family is obtained by letting $a,b$ and $d$ taking the limiting value $\delta >0$, leading to the loss function (see p. 79 in \cite{Huber1964})
    $$\ell_\delta(x,c)=\begin{cases}
        (x-c)^2 &\text{ if } \lvert x-c\rvert <  \delta\\
        \delta^2 &\text{ if } \lvert x-c \rvert  \ge  \delta, 
    \end{cases} $$
appearing in the context  of trimmed means. 
Note that the  function $c\mapsto \ell_\delta(x,c)$
is quadratic in $c$ around $x$, but attributes always the same constant value ($\delta^2$)  when $c$ is such that $\lvert x -c \rvert \ge \delta$. Consequently the function $c \mapsto \ell_{\delta}(x,c)$ is not convex.
\\

\subsubsection{Real-valued loss functions}
\textbf{Various functionals involving moments}\\
Let $I\subset \mathbb{R}$ be an interval. Consider $r:I \to \mathbb{R}$ and $s:I \to (0,+\infty)$ measurable functions. Let $\phi$ be a strictly convex function with subgradient $\phi'$. Let $X_{\sD_{\btau}}$ be a random variable supported on $I$ such that $\mathbb{E}[r(X_{\sD_{\btau}})]$, $\mathbb{E}[s(X_{\sD_{\btau}})]$, $\mathbb{E}[r(X_{\sD_{\btau}})\phi'(X_{\sD_{\btau}})]$,  $\mathbb{E}[s(X_{\sD_{\btau}})\phi(X_{\sD_{\btau}})]$ and  $\mathbb{E}[X_{\sD_{\btau}} s(X_{\sD_{\btau}})\phi'(X_{\sD_{\btau}})]$ exist and are finite. 
Consider the loss function 
$$\ell(x,c)=s(x)(\phi(x)-\phi(c))-\phi'(c)(r(x)-cs(x))+\phi'(x)(r(x)-xs(x)).$$
The corresponding generalized extremile equals 
$$\frac{\mathbb{E}[r(X_{\sD_{\btau}})]}{\mathbb{E}[s(X_{\sD_{\btau}})]}.$$
We refer to \citeauthor{Gneiting2011} (see e.g. Theorem 8 on p. 754 in \cite{Gneiting2011}) for an in-depth discussion.\\

Taking $r(x)=x^2$, $s(x)=x$ and $\phi(x)=x^2$ leads  to the special case of 
\[
\frac{\mathbb{E}[\left ( X_{\sD_{\btau}}\right )^2]}{\mathbb{E}[X_{\sD_{\btau}}]}=\frac{\text{var}(X_{\sD_{\btau}})}{\mathbb{E}[X_{\sD_{\btau}}]}+\mathbb{E}[X_{\sD_{\btau}}].
\] 
In this special case the loss function is $\ell(x,c)=x(c-x)^2$ (see also p. 83 in \cite{Heilmann1989}). This is a convex function in $c$ for $x>0$, and the derivative of $c\mapsto \ell(x,c)$ exists for every $x$ and is given by
$\ell_\delta'(x,c)=-2x(x-c)$.\\

Consider $w:[0,+\infty)\to [0,+\infty)$; and  take $r(x)=xw(x)$ and $s(x)=w(x)$. The interpretation of the resulting generalized extremile 
\[
\frac{\mathbb{E}[X_{\sD_{\btau}} w(X_{\sD_{\btau}})]}{\mathbb{E}[w(X_{\sD_{\btau}})]}, 
\]
 for different weighting functions $w$ is discussed in \cite{Furman2009}.\\
 \medskip

\subsection{Square loss and distortion risk measures}\label{sec: SquareLossDistortion}
A special subclass of functionals of generalized extremiles is obtained when focussing on the square loss $\ell_{\delta}(x,c)=(x-c)^2$. Indeed, as already pointed out in Section \ref{sec: Introduction}, we know that $e_{\btau, \delta}(X; D_{\btau}, \ell_{\delta})=\mathbb{E}[X_{\sD_{\btau}}]$ where $X_{\sD_{\btau}}$ is distributed with cumulative distribution function $D_{\btau} \circ F_X$.
This implies that
\begin{equation}\label{eq: DistortionRisk}
e_{\btau}(X; D_{\btau}):= e_{\btau}(X,D_{\btau},(x-c)^2)=\int_{-\infty}^0 [g(1-F_X(y))-1]\mathrm{d}y+\int_0^{+\infty}g(1-F_X(y))\mathrm{d}y , 
\end{equation} 
where $g(y)=1-D_{\btau}(1-y)=\widetilde{D}_{\btau}(y)$.  
An expression as in \eqref{eq: DistortionRisk}, with $g: [0,1] \to [0,1]$, a non-decreasing function and $g(0)=0$ and $g(1)=1$, is termed a distortion risk of $X$. For this distortion risk to be a coherent risk measure (see for example \cite{Wang1996} and \cite{ArtznerEtAl1999}) one needs to require in addition that $g$ is a concave function. See, for example, \cite{WangEtAl1997} and 
\cite{CaiWangMao2017}. 
Note that generalized extremiles, in case of  square loss function, can be seen as a distortion risk measure with distortion function $1-D_{\btau}(1-y)$. This implies that $$D_{\btau}(y)=1-g(1-y).$$
Note that concavity of $g$ is equivalent with convexity of $D_{\btau}$. Therefore not every 
coherent distortion risk measure is necessarily part of our class. If the function $g$ is differentiable then
$d_{\tau}(u) = g'(1-u)$.  In summary, the class of generalized extremiles contains part of the distortion risk measures, but it also contains  non-distortion risk measures, for example, expectiles.
Table \ref{DistrFunctions}  lists several distortion risk measures that can be seen as part of our general class of functionals.

\section{Estimation}\label{sec:estimation}

Based on an i.i.d. sample $X_1, \dots, X_n$ from $X$, with unknown distribution function $F_X$, the aim is to estimate the generalized extremile
$e_{\btau, \delta}(X; D_{\btau}, \ell_{\delta})$ defined in Definition \ref{def: generalized extremiles}, and this for a general pair $(D_{\btau}, \ell_{\delta})$. 
We  firstly focus on the setting of a square loss function, but with general $D_{\btau}$; and secondly treat the setting of a general pair $(D_{\btau}, \ell_{\delta})$. The reason is that in the former setting estimation  of the generalized extremile is simpler, and subsequently establishing (asymptotic) properties for it is somewhat less involved. 

\subsection{Square loss and general distribution function $D_{\btau}$}\label{SquareLossGeneralDtau}
As discussed in Section \ref{sec: SquareLossDistortion} in case of square loss the generalized extremile reduces to
\begin{equation}\label{eq: SquareLossGenExtrelime} 
e_{\btau}(X;D_{\btau})=\mathbb{E}[X_{{\sD}_{\btau}}]=\int_{0}^{1}F_X^{-1}(u)\mathrm{d} D_{\btau}(u). 
\end{equation}
An estimator of $e_{\btau}(X;D_{\btau})$ (and hence of the corresponding distortion risk) can be easily derived from \eqref{eq: SquareLossGenExtrelime}. 
 Denote by $X_{k: n}$ the $k$th order statistic in the  random sample $X_1, \dots, X_n$ of size $n$. Let $F_n(x)= (n+1)^{-1} \sum_{i=1}^n 1_{\{ X_i \le x\}}$ be the empirical cumulative distribution function  (with factor $1/(n+1)$ instead of $1/n$), and denote by 
  $F_n^{-1}(p)$ the corresponding empirical quantile function, i.e. for $p\in (0,1)$,   $F_n^{-1}(p)=X_{k: n}$, where $k$ is chosen such that $(k-1) / (n+1) <p \leq k / (n+1)$. A plug-in type of estimator for $e_{\btau}(X;D_{\btau})$ is then
\[
\widehat{e}_{\btau,n}^{\sL} =\sum_{k=1}^n\left(D_{\btau}(k / (n+1))-D_{\btau}((k-1) / (n +1))\right) X_{k: n}
:= \frac{1}{n+1} \sum_{k=1}^n w_{k, n} X_{k: n},
\]
where it is thus readily seen that $\widehat{e}_{\btau,n}^{\sL}$ is proportional (with factor $n/(n+1)$) to an  $\mathrm{L}$-statistic with weights $w_{k, n}=(n+1)\left (D_{\btau}(k / (n+1))-D_{\btau}((k-1) / (n+1))\right)$.

A second estimator, denoted $\widehat{e}_{\btau,n}^{\sLM}$, relies on the fact that in case of square loss we have
\[
 e_{\btau}(X; D_{\btau})= \mathbb{E}[X d_{\btau} (F_X(X))], 
 \]
leading to the  estimator
 $$\widehat{e}^{\sLM}_{\btau,n}=\frac{1}{n} \sum_{i=1}^n d_{\btau}\left (\frac{i}{n+1}\right )X_{i:n}. $$

A third estimator is obtained by 
 solving the empirical minimization problem.  That is by solving 
$$\argmin_{c\in \mathbb{R}} \frac{1}{n} \sum_{i=1}^n d_{\btau}\left (\frac{i}{n+1}\right )\Big((X_{i:n}-c)^2-X_{i:n}^2  \Big) = \argmin_{c \in \mathbb{R}} \frac{1}{n} \sum_{i=1}^n d_{\btau}\left (\frac{i}{n+1}\right )\Big(-2 c X_{i:n}+c^2  \Big). $$
Differentiating the objective function with respect to $c$ and equating this to zero gives
\[
\frac{1}{n}\sum_{i=1}^n d_{\btau}\left ( \frac{i}{n+1}\right) (-2X_{i:n}+2c)=0 
\quad  \Longleftrightarrow  \quad c=\ds \frac{\ds\sum_{i=1}^n d_{\btau}\left (\frac{i}{n+1}\right )X_{i:n}}{\ds\sum_{i=1}^n d_{\btau}\left (\frac{i}{n+1}\right )}. 
\]
We thus obtain a new M-type estimator  
\[
  \widehat{e}^{\sM}_{\btau,n}=\frac{\ds \sum_{i=1}^n d_{\btau}\left (\frac{i}{n+1}\right)X_{i:n}}{\ds \sum_{i=1}^n d_{\btau}\left (\ds \frac{i}{n+1}\right )}=\frac{\widehat{e}^{\sLM}_{\btau,n}}{\ds \frac{1}{n} \ds \sum_{i=1}^n d_{\btau}\left (\frac{i}{n+1}\right )}.
  \]  

Note that, since we work with the empirical cumulative  distribution function rescaled with the factor $1/(n+1)$, which takes values between $1/(n+1)$ and $n/(n+1)$, when evaluated in  the observations $X_1, \dots, X_n$, we avoid evaluation of the density $d_{\btau}$ in the endpoints 0 and 1. 

All the above three  estimators are consistent and have the same (first order) asymptotic distribution. The proofs of the results stated in this section are given in Appendix \ref{App: ProofsConsistencyANresultsSquareLoss}.

\begin{theorem}\label{TheorySquareLossEstimation}
    Given is a fixed index $\btau$. Assume that $D_{\btau}$ is differentiable with density $d_{\btau}$. Let $\widehat{e}_{\btau,n}$ be any of the two estimators  $\widehat{e}_{\btau,n}^{\sLM}$ or $\widehat{e}_{\btau,n}^{\sM}$.
    \begin{itemize}
    \item[(i)] \textbf{[Consistency]}  Assume $\mathbb{E}[\lvert 
     X \rvert^\kappa]<\infty$ for some $\kappa>1$. Let $d_{\btau}$ be continuous almost everywhere. If for some $0<M<\infty$
    $$\lvert d_{\btau} (t) \rvert \leq M [t(1-t)]^{-1+1/\kappa+\gamma} \quad \text{ for } t \in (0,1)$$
    for some $\gamma>0$, then $\widehat{e}_{\btau,n}\xrightarrow{\text{a.s.}} e_{\btau}(X; D_{\btau})$ as $n\to \infty$.\\
    If in addition $d_{\btau}$ is Lipschitz continuous or bounded (uniformly) on $(0,1)$ then $\widehat{e}_{\btau,n}^{\sL} \xrightarrow{\text{a.s.}} e_{\btau}(X; D_{\btau})$ as $n\to \infty$
    \item[(ii)] \textbf{[Asymptotic normality result]} Assume  $\mathbb{E}[\lvert 
     X \rvert^\kappa]<\infty$ for some $\kappa>2$. Suppose $d_{\btau}'$ exists and is continuous on $(0,1)$. Furthermore, suppose there exists some $\gamma>0$ such that
     $$\lvert d_{\btau} (t) \rvert \leq M [t(1-t)]^{-1/2+1/\kappa+\gamma} \quad \text{ for } t \in (0,1) $$
     and
      $$\lvert d_{\btau}'(t) \rvert \leq M [t(1-t)]^{-3/2+1/\kappa+\gamma} \quad \text{ for } t \in (0,1).$$
      Then  $\sqrt{n}(\widehat{e}_{\btau,n}-e_{\btau}(X; D_{\btau}) )\xrightarrow{d} \mathcal{N}(0; \sigma_{\btau}^2)$ as $n\to \infty$ where
\begin{equation}\label{eq:asymptoticVarianceSquaredLoss}
  \sigma_{\btau}^2=\int_{-\infty}^\infty \int_{-\infty}^\infty(F_X(\min (s, t))-F_X(s) F_X(t)) d_{\btau}(F_X(s)) d_{\btau}(F_X(t)) \mathrm{d}s \mathrm{d}t.  
\end{equation}
If in addition $d_{\btau}$ is Lipschitz continuous  or bounded (uniformly) then $\sqrt{n}(\widehat{e}_{\btau,n}^{\sL}-e_{\btau}(X; D_{\btau}) )\xrightarrow{d} \mathcal{N}(0; \sigma_{\btau}^2)$ as $n\to \infty$.
    \end{itemize}
\end{theorem}

When $d_{\btau}$ is bounded, we obtain the almost sure convergence for any random variable $X$ with a finite absolute $\kappa$th moment (for some $\kappa>1$). In particular, when $D_{\btau}(u)=K_{\btau}(u)$, that is the case of extremiles, we obtain the same result as in Theorem 1(i) on p. 1370 of \cite{DaouiaAndGijbels2019}. In the case of extremiles, also the function $\lvert d_{\btau}'(t) \rvert t(1-t) $ is bounded. Hence, also the asymptotic normality result in Theorem 1(ii) on p. 1370 of \cite{DaouiaAndGijbels2019} is included as a special case of Theorem \ref{TheorySquareLossEstimation} in our more general framework.

\subsection{General loss functions and distribution functions}\label{sec:estimation--GeneralCase}

We now turn to the general setting of any pair $(D_{\btau}, \ell_{\delta})$.
From now on we assume that  $X$ is an absolutely continuous random variable, i.e., $F_X$ has an associated density function $f_X$. 
We only consider loss functions $\ell_{\delta}(x,c)$ for which 
$
c\mapsto \ell_{\delta}(x,c) \, \mbox{ is convex for every $x$}. 
$
By Proposition \ref{prop: convexLossDeriviativeZero} we know that solving the minimization problem is equivalent to finding a value $c$ for which
$$\frac{\mathrm{d} }{\mathrm{d} c}\Big( \int_{-\infty}^{+\infty} \ell_{\delta}(x,c\, )\mathrm{d} D_{\btau} (F_X(x))\Big)=0.$$
Since  $X$ is an absolutely continuous random variable,  Theorem \ref{theorem:leibnizAbsolutelyContinuous} then ensures the right-hand side equals 
$$\frac{\mathrm{d}}{\mathrm{d}c}\Big( \int_{-\infty}^{+\infty} \ell_{\delta}(x,c)d_{\btau}(F_X(x))f_X(x)\mathrm{d}x \Big)= \int_{-\infty}^{+\infty} \ell_{\delta}'(x,c) d_{\btau}(F_X(x))f_X(x)\mathrm{d}x.$$
We are now in the setting of an $M$-functional problem: finding a value $c$ such that
$$\lambda_F(c):=\int_{-\infty}^{+\infty} d_{\btau}(F_X(x))\ell_{\delta}'(x,c) f_X(x)\mathrm{d}x=0.$$

Given a random  sample $X_1,\dots, X_n$ from $X$, the $M$-estimator is defined as the solution $T_n$ to 
$$\lambda_{F,n}(c):=\frac{1}{n} \sum_{i=1}^n  d_{\btau}(F_X(X_i))\ell_{\delta}'(X_i,c) =0.$$
However, the appearance of $F_X$ leads to an extra complication. 
Since $F_X$ is unknown we replace it by a consistent estimator $\widehat{F}_n$ for it, for example,  by its empirical version $F_n(x)= (n+1)^{-1} \sum_{i=1}^n 1_{\{ X_i \le x\}}$. The estimator $T_{n}^*$ is then defined as the value $c$ for which 
$$\lambda_{\widehat{F}_n}^*(c):=\frac{1}{n}  \sum_{i=1}^n d_{\btau}(\widehat{F}_n(X_i))\ell_{\delta}'(X_i,c)=0.$$
If there is no such value $c$, then we define $T_{n}^*$ as the value minimizing $\lvert \lambda_{\widehat{F}_n}^*(c) \rvert$.
If, in practice, there are multiple values of $c$ for which $\lambda_{\widehat{F}_n}^*(c)=0$, we take the smallest of them.

In what follows we establish the consistency and the asymptotic normality of the estimator $T_{n}^*$. Before stating the results, we introduce some assumptions, needed in various parts of the theoretical results. 
\begin{itemize}
    \item[\textbf{(A1)}] 
    \hspace*{-0.4 cm}\begin{minipage}[t]{10 cm}
    \vspace*{-0.50 cm}
    
    \noindent
    \begin{enumerate}
        \item $X$ is an absolutely continuous random variable;  \item $\lambda_F(c)$ is finite for every $c$ and has a unique root $t_0$;
        \item $c\mapsto \ell_{\delta}(x,c)$ is convex for every $x$;
        \item $x\mapsto\ell_\delta'(x,c)$ is measurable for every $c$.
    \end{enumerate}
    \end{minipage}
    \item[\textbf{(A2)}] $\mathbb{E}[\lvert \ell_{\delta}'(X,c)\rvert]<\infty$.
    \item[\textbf{(A3)}] There exist functions $h_1,h_2$ with 
    $x\mapsto h_j(x,c)$ non-decreasing and left continuous for $j=1,2$ such that $\ell_\delta'(x,c)=h_1(x,c)-h_2(x,c)$ for every $c$ and $x$.
    \item[\textbf{(A4)}] There exists an $r>0$, such that $\mathbb{E}[\lvert h_j(X,c)\rvert^r]<\infty$ for $j=1,2$ with $h_j$ from \textbf{(A3)}.
    \item[\textbf{(A5)}] There exist constants $M,\gamma >0$ such that 
    $\lvert d_{\btau}(t) \rvert  \leq M [t(1-t)]^{-1+1/r+\gamma}$ with $r$ from \textbf{(A4)}.
    \item[\textbf{(A6)}] $d_{\btau}$ has a continuous derivative on $(0,1)$. Furthermore, with $r$ from \textbf{(A4)} there exist $M,\gamma >0$ such that  
    $\lvert d_{\btau}(t) \rvert  \leq M [t(1-t)]^{-1/2+1/r+\gamma}$ and $\lvert d_{\btau}'(t) \rvert  \leq M [t(1-t)]^{-3/2+1/r+\gamma}$.
    \item[\textbf{(A7)}] $d_{\btau}$ is continuous except at possibly a finite number of points.
    \item[\textbf{(A8)}] 
    \begin{enumerate}
        \item $d_{\btau}$ is bounded on $[0,1]$;
         \item $d_{\btau}$ is locally uniformly continuous except at a finite number of points.
    \end{enumerate}
    \item[\textbf{(A9)}] $\lambda_F(c)$ is differentiable at $t_0$ with derivative $\lambda_F'(t_0)\neq 0$.
    \item[\textbf{(A10)}]  For any sequence $\omega_n \rightarrow 0$ as $n \to \infty$ $$\sup_{ \lvert c-t_0 \rvert \leq \omega_n} \frac{\sqrt{n} \lvert \lambda_{\widehat{F}_n}^*(c)-\lambda_{\widehat{F}_n}^*(t_0) - \lambda_F(c) \rvert}{1+\sqrt{n} \lvert c-t_0 \rvert }  \xrightarrow{\text{P}} 0.$$ 
\end{itemize}

Throughout this section we assume the existence (i.e. finiteness), as well as uniqueness, of the target quantity of interest $e_{\btau, \delta}(X; D_{\btau}, \ell_{\delta})$ (also shortly denoted  as $t_0$). \\

Appendix \ref{App: ProofsConsistencyANresultsGenLoss} contains the proofs of the results stated in the various lemmas in this section. 
Lemma \ref{lemma:consistency} expresses the consistency of  $T_{n}^*$ when using a general estimator $\widehat{F}_n$ of $F_X$, satisfying some condition.
Note that this lemma requires the Lipschitz continuity of $d_{\btau}$ which might be a quite restrictive condition. For this reason, we also prove the consistency of  $T_{n}^*$, when using as an estimator for $F_X$ the empirical distribution function. This then allows to show consistency under weaker conditions on $d_{\btau}$. See Lemma \ref{lemma:almostSureConvergence}.

\begin{lemma}\label{lemma:consistency}
Suppose Assumption \textbf{(A1)} and Assumption \textbf{(A2)} hold. Furthermore,  assume that $d_{\btau}$ is Lipschitz continuous. Then for any estimator $\widehat{F}_n$ of $F_X$ for which $\sup_{x\in \mathbb{R}} \lvert \widehat{F}_n(x)-F_X(x) \rvert \xrightarrow{\text{a.s.}} 0$ as $n\to \infty$, we have that $T_n^* \xrightarrow{\text{a.s.}} e_{\btau,\delta}(X;D_{\btau},\ell_{\delta})$ as $n \to \infty$. 
\end{lemma}

\begin{remark}\label{remark:existenceOfRootLambda}
\textbf{(i)} 
    Suppose in addition to the assumptions of Lemma \ref{lemma:consistency} that $\ell_{\delta}'(x,c)$ is continuous in $c$ in a neighborhood of $e_{\btau,\delta}(X;D_{\btau},\ell_{\delta})$, for every $x$. Then $\lambda_{\widehat{F}_n}^*$ has a root and this root is equal to $T_n^*$.\\
\textbf{(ii)} The requirement on the nonparametric estimator $\widehat{F}_n$  for $F_X$ is satisfied by the usual empirical  distribution function (or its rescaled version $F_n$). Indeed, this is expressed  by the Glivenko-Cantelli theorem (see p. 62 in \cite{Serfling1980}). 
\end{remark}

In Lemma \ref{lemma:almostSureConvergence}, and in all subsequent results,  we work with the estimator 
$\widehat{F}_n(x):=F_n(x)=(n+1)^{-1} \sum_{i=1}^n 1_{\{X_i\leq x\}}$.
This allows us to weaken the assumption on the density $d_{\btau}$. 

\begin{lemma}\label{lemma:almostSureConvergence}
Suppose Assumptions \textbf{(A1)}, \textbf{(A3)},  \textbf{(A4)},  \textbf{(A5)}  and \textbf{(A7)} hold, then $T_n^*\xrightarrow{\text{a.s.}} e_{\btau, \delta}(X; D_{\btau}, \ell_{\delta})$ as $n\to \infty$.
\end{lemma}
\medskip 

We next turn to establishing an asymptotic normality result for the estimator $T_n^*$. We therefore  first prove the asymptotic normality of $\lambda_{F_n}^*(t_0)$. In Lemma \ref{lemma:asympLambda} this is done under, among others, the assumption that the density $d_{\btau}$ has a continuous derivative. In Lemma \ref{lemma:asympLambdaSecondVersion} this assumption is relaxed.
To ease calculations in practice, we provide two equivalent expressions for the asymptotic variance in Lemma \ref{lemma:asympLambda} (and subsequently Lemma \ref{lemma:asympLambdaSecondVersion}). 
\begin{lemma}\label{lemma:asympLambda}
  Suppose Assumptions \textbf{(A1)}, \textbf{(A3)}, \textbf{(A4)} and \textbf{(A6)} are satisfied. Then $\sqrt{n} \lambda_{F_n}^*(t_0) \xrightarrow{\text{d}} \mathcal{N}(0;\sigma^2_{t_0})$ as $n\to \infty$, where
\begin{align}
   \sigma^2_{t_0}&= \int_{-\infty}^{+\infty} \int_{-\infty}^{+\infty} [F_X(\min\{ x,y\}) - F_X(x)F_X(y)] \cdot d_{\btau}(F_X(x))d_{\btau}(F_X(y)) \mathrm{d}\ell_{\delta}'(x,t_0) \mathrm{d}\ell_{\delta}'(y,t_0)\label{eq: sigmat01}\\
   &=\sum_{i,j\in \{1,2\}} (-1)^{i+j} \int_{-\infty}^{+\infty} \int_{-\infty}^{+\infty} [\min\{ F_{h_i(X,t_0)}(x),F_{h_j(X,t_0)}(y)\} - F_{h_i(X,t_0)}(x)F_{h_j(X,t_0)}(y)]\notag\\ 
   &\hspace{5cm}\cdot d_{\btau}(F_{h_i(X,t_0)}(x))d_{\btau}(F_{h_j(X,t_0)}(y)) \mathrm{d}x \mathrm{d}y.\label{eq: sigmat02}
\end{align}
\end{lemma}

As Lemma \ref{lemma:asympLambda} does not allow for a discontinuous density $d_{\btau}$, it does not cover the case for $D_{\btau}$ corresponding to expected shortfall (entry number 5 in Table 3). The following result covers this specific case.

\begin{lemma}\label{lemma:asympLambdaSecondVersion}
Suppose Assumptions \textbf{(A1)}, \textbf{(A3)}, \textbf{(A4)} with $r>2$ and \textbf{(A8)} are satisfied. Then $\sqrt{n} \lambda_{F_n}^*(t_0) \xrightarrow{\text{d}} \mathcal{N}(0;\sigma^2_{t_0})$ as $n\to \infty$.
\end{lemma}

We can now  state (and prove) the asymptotic normality result for $T_n^*$.

\begin{theorem}\label{theorem: asymptoticNormalityEmpirical}
Suppose that in addition to Assumptions \textbf{(A1)}, \textbf{(A3)}, \textbf{(A9)} and \textbf{(A10)} one of the following sets of assumptions is satisfied
\begin{itemize}
    \item  \textbf{(A4)} and \textbf{(A6)}
    \vspace*{-0.44 cm}
    
    \noindent
    \item[] \hspace*{-0.8 cm} or 
        \vspace*{-0.44 cm}
    
    \noindent
    \item \textbf{(A4)} with $r>2$ and \textbf{(A8)}.
\end{itemize}
Then $\sqrt{n}(T_n^* - t_0)  \xrightarrow{\text{d}} \mathcal{N}\left ( 0; \frac{\sigma^2_{t_0}}{[\lambda_F'(t_0)]^2}\right )$ as $n\to \infty$, with $\sigma^2_{t_0}$ as in \eqref{eq: sigmat01} (or \eqref{eq: sigmat02}) in  Lemma \ref{lemma:asympLambda}.
\end{theorem}
\begin{proof}
    This follows by Theorem 7.2 on p. 2186 in \cite{NeweyAndMcFadden1986}. See also Theorem 3.3 on p. 1040 in \cite{PakesAndPollard1989}. Both references state the assumption (using our notations)
$$(\lambda_{F_n}^*(T_n^*))^2 \leq \inf_{c \in \mathbb{R}} \Big(\lambda_{F_n}^*(c)\Big)^2 + o_P(n^{-1}).$$ By definition of $T_n^*$ this is always satisfied.
\end{proof}

\begin{remark}\label{remark:equicontcontinuity condition}
\textbf{(i)} 
Assumption \textbf{(A10)} is a kind of  equicontinuity assumption. 
In Section S2 of the Supplementary Material it is shown that this assumption is satisfied for several loss functions, among which the square loss, the Esscher loss (entry number 6 in Table \ref{LossFunctions1A}), and the loss function with entry G2 in Table \ref{LossFunctions2A}. \\
\noindent
\textbf{(ii)} In Section \ref{SquareLossGeneralDtau} we established an asymptotic normality result for the considered estimators, in the specific setting of square loss ($\ell_{\delta}(x,c)=(x-c)^2$). See Theorem \ref{TheorySquareLossEstimation}. Theorem  \ref{theorem: asymptoticNormalityEmpirical}
on the other hand is for general loss functions. 
It is easily seen that in case of square loss the expression for the asymptotic variance in \eqref{eq: sigmat02} reduces to that in \eqref{eq:asymptoticVarianceSquaredLoss}. Indeed, note that  for the square loss function: $\ell_{\delta}'(x,c)=-2(x-c)$ and  $\lambda_F'(c)=-2 $. A simple substitution into \eqref{eq: sigmat02} leads to \eqref{eq:asymptoticVarianceSquaredLoss}. A setting with square loss was also considered in \cite{JonesAndZitikis2003}.
Although the authors consider another estimator, the asymptotic variances coincide. Since \cite{JonesAndZitikis2003} (see pp. 49--50) also discuss a consistent estimator for the asymptotic variance, this estimator can also be used in our context.
\end{remark} 

\subsection{Applications of the asymptotic normality result}\label{sec:estimation--ApplicationsANResult}

The asymptotic normality results established in Theorem \ref{TheorySquareLossEstimation}(ii) and Theorem \ref{theorem: asymptoticNormalityEmpirical} allow, for example, to create asymptotic (pointwise) confidence intervals for the generalized extremile $e_{\btau, \delta}(X; D_{\btau}, \ell_{\delta})$. 

In this section, we apply Theorem \ref{theorem: asymptoticNormalityEmpirical} for various loss functions, establishing as such results on the corresponding estimator $T_n^*$  for the generalized extremile $e_{\btau, \delta}(X; D_{\btau}, \ell_{\delta})$ under study. 
In each of the illustrative examples, we focus on the asymptotic normality result and in particular the expression for the asymptotic variance. 

We start by looking into the two loss functions encountered when estimating quantiles. Recall from Section \ref{sec: Introduction} that quantiles  of $X$ can be obtained using at least two different pairs $(D_{\btau}, \ell_{\delta})$.  
\begin{eqnarray}
\mbox{approach 1}: \qquad & & D_{\btau}(u)=u \qquad \mbox{and} \qquad \ell_{\delta}(x,c)=\left|\delta-1_{\{x\le c\}}\right| |x-c| \label{QuantileApproach1}\\
\mbox{approach 2}: \qquad  & & D_{\btau}(u)=K_{\tau}(u) \quad \mbox{in \eqref{DtauExtremiles}}  \qquad \mbox{and} \qquad \ell_{\delta}(x,c)= |x-c| , 
\label{QuantileApproach2}
\end{eqnarray}
where in approach 2, one takes $\tau=\delta$ to obtain the $\delta$th quantile of $X$. 

In Corollary \ref{cor:absoluteValueLossAsymptoticVariance} we consider an absolute value loss function, whereas in Corollary \ref{cor:QuantileLossAsymptoticVariance} we apply Theorem \ref{theorem: asymptoticNormalityEmpirical} with the loss function 
$\ell_{\delta}(x,c)=\left|\delta-1_{\{x\le c\}}\right| |x-c|$. 
Before stating the first corollary we make an important remark.
\begin{remark}
    Consider the absolute value loss $\ell_{\delta}(x,c)=\lvert x-c\rvert.$ The derivative with respect to $c$ of this function is not defined in the point $c=x$. In order for $x \mapsto \ell_{\delta}'(x,c)$ to be left continuous (see Assumption \textbf{(A3)}), we will choose 
    $\ell_{\delta}'$ to be the right-hand side  derivative. In the specific case of the absolute value loss, this means $\ell_{\delta}'(x,c)=1_{\{x\leq c\}}-1_{\{x> c\}}$. This choice for the right-hand side derivative is justified since we are integrating with respect to $x$. Meaning that the integrand is uniquely defined almost everywhere.
\end{remark} 

\begin{corollary}\label{cor:absoluteValueLossAsymptoticVariance}
    Consider the absolute value loss function $\ell_{\delta}(x,c)=\lvert x-c\rvert.$  Consider a random variable $X$ and a cumulative distribution function $D_{\btau}$ such that the assumptions of Theorem \ref{theorem: asymptoticNormalityEmpirical} are satisfied. Then the asymptotic variance reduces to
\begin{equation}\label{eq: ANVarianceQuantile1}
\frac{\text{Median}(D_{\btau})(1-\text{Median}(D_{\btau}))}{\left [ f_X\left (  \mbox{Median}(X_{\sD_{\btau}})\right ) \right ]^2}, 
\end{equation} 
where $\mbox{Median}(D_{\tau})$ denotes the median of a random variable with cumulative distribution function $D_{\btau}$. 
\end{corollary}
\begin{proof}
We have $\ell'_{\delta}(x,c)=1_{\{x\leq c\}}-1_{\{x> c\}}$, with thus $x\mapsto \ell'_{\delta}(x,c) $ non-increasing. This means we can take $h_1(x,c)=0$ and $h_2(x,c)=-\ell'_{\delta}(x,c)$ for Assumption \textbf{(A3)}. 
    When integrating a function, say $k(s)$, with respect to this measure, we get
    $$\int_{\mathbb{R}} \,   k(s) \mathrm{d}\ell'_{\delta}(s,t_0)=-2k(t_0).$$
    This implies $\sigma_{t_0}^2=4[F_X(t_0)-F_X^2(t_0)][d_{\btau}(F_X(t_0))]^2$.
    Further, we have that
    \[
    \lambda_F(c)=\int_{-\infty}^{c}d_{\btau}(F_X(x))f_X(x)\mathrm{d}x-\int_{c}^{+\infty} d_{\btau}(F_X(x))f_X(x)\mathrm{d}x 
    = P \left \{ X_{\sD_{\btau}} \le c\right \} -  P \left \{ X_{\sD_{\btau}} > c \right \}= 2 D_{\btau}\left (F_X(c) \right ) -1 , 
    \]
    from which it follows that $t_0 = \mbox{Median}\left (X_{\sD_{\btau}} \right )$, and also $F_X(t_0)=\mbox{Median}( D_{\btau})$.
    
    Furthermore $\lambda_F'(c)=2 f_X(c) d_{\btau}(F_X(c))$. In conclusion, 
 $\sigma_{t_0}^2/ \left [ \lambda_F'(t_0) \right ]^2$ leads to expression \eqref{eq: ANVarianceQuantile1}.
\end{proof}

\begin{remark}\label{AVarQuantileEstimators} 
When taking $D_{\btau}(u)=K_{\tau}(u)$, expression \eqref{eq: ANVarianceQuantile1} leads to the asymptotic variance for the quantile estimator, obtained via approach 2 (see \eqref{QuantileApproach2}). It is easily seen from \eqref{DtauExtremiles}
that $\mbox{Median}(K_{\tau})=\tau$  and $\mbox{Median}\left ( X_{\sD_{\tau}} \right ) = q_{\tau}(X)$ the $\tau$th order quantile of $X$. Hence for the setting of approach 2 in \eqref{QuantileApproach2}, expression \eqref{eq: ANVarianceQuantile1} reduces to 
\begin{equation}
\label{ANQuantile}
\frac{\tau(1-\tau)}{\left [ f_X \left (q_{\tau}(X)\right ) \right ]^2}, 
\end{equation}
which is the well-known asymptotic variance for the empirical quantile estimator (see p. 77 in \cite{Serfling1980}).
\end{remark}

We next apply Theorem \ref{theorem: asymptoticNormalityEmpirical} for the loss function $\ell_{\delta}(x,c)=\left|\delta-1_{\{x\le c\}}\right| |x-c|$, leading to the statement in Corollary \ref{cor:QuantileLossAsymptoticVariance}. 

\begin{corollary}\label{cor:QuantileLossAsymptoticVariance}
    Consider the quantile loss function $\ell_{\delta}(x,c)=\left|\delta-1_{\{x\le c\}}\right| |x-c|$ (with $\delta \in (0, 1)$). 
    Consider a random variable $X$ and a cumulative distribution function $D_{\btau}$ such that the assumptions of Theorem \ref{theorem: asymptoticNormalityEmpirical} are satisfied. Then the asymptotic variance reduces to
\begin{equation}\label{eq: ANVarianceQuantile2}
\frac{F_X(q_{\delta}(X_{\sD_{\btau}}))\left [ 1 -F_X(q_{\delta}(X_{\sD_{\btau}}))\right ]}{\left [ f_X(q_{\delta}(X_{\sD_{\btau}}))\right ]^2}.
\end{equation}
\end{corollary}
\begin{proof}
  We have $\ell'_{\delta}(x,c)=1_{\{x \leq c\}} -\delta$ (see also Table \ref{LossFunctions1B}), with $x\mapsto\ell'_{\delta}(x,c)$ non-increasing. This means we can take $h_1(x,c)=0$ and $h_2(x,c)=-\ell'_{\delta}(x,c)$ for Assumption \textbf{(A3)}. 
    If we integrate with respect to this measure then
    $$\int k(s)\mathrm{d}\ell'_{\delta}(s,t_0)=-k(t_0).$$
    For this reason $\sigma_{t_0}^2=[F_X(t_0)-F^2_X(t_0)][d_{\btau}( F_X( t_0))]^2$.
    Furthermore,
    $$\lambda_F(c)=\int_{-\infty}^c d_{\btau}(F_X(x))f_X(x)\mathrm{d}x-\delta
    = P \left \{ X_{\sD_{\btau}} \le c \right \} - \delta ,$$
    leading to $t_0= q_{\delta}(X_{\sD_{\btau}})$, the $\delta$th quantile of the random variable $X_{\sD_{\btau}}$. 
    Moreover, $\lambda_F'(t_0)=d_{\btau}(F_X(t_0))f_X(t_0)$. This concludes the proof.
\end{proof}
\begin{remark}\label{AVarQuantileEstimatorsBis} 
When estimating a quantile, via approach 1 (see \eqref{QuantileApproach1}), we take $D_{\btau}(u)=u$, and hence $X_{\sD_{\btau}}=X$, and expression \eqref{eq: ANVarianceQuantile2} reduces to 
\[
\ds \frac{\delta(1-\delta)}{\left [ f_X(q_{\delta}(X)) \right ]^2}, 
\]
corresponding to \eqref{ANQuantile} (where herein $\tau=\delta$). 
\end{remark}

The asymptotic variance stated in Corollary \ref{cor:QuantileLossAsymptoticVariance}
can be estimated by 
\begin{equation}\label{eq:quantileLossEstimatorAsymptoticVariance}
    \frac{[F_n(T_n^*)-(F_n(T_n^*))^2]}{\left [ \widehat{f}_X(T_n^*)\right ]^2},
\end{equation}
where $F_n$ is the empirical cumulative distribution and $\widehat{f}_X$ a kernel density estimator.\\

Corollary \ref{cor:expectilesAsymptoticVariance}
results from applying Theorem \ref{theorem: asymptoticNormalityEmpirical} 
with the loss function $\ell_{\delta}(x,c)=\left|\delta-1_{\{x\le c\}}\right| |x-c|^2$, i.e.  the expectile loss function (entry number 4 in Table \ref{LossFunctions1A}). 

\begin{corollary}\label{cor:expectilesAsymptoticVariance}
    Consider the loss function $\ell_{\delta}(x,c)= \left|\delta-1_{\{x\le c\}}\right| |x-c|^2$ (with $\delta \in (0, 1)$), corresponding to expectiles. Consider a random variable $X$ and a cumulative distribution function $D_{\btau}$ such that the assumptions of Theorem \ref{theorem: asymptoticNormalityEmpirical} are satisfied. Then the asymptotic variance reduces to
$\sigma_{\sE}^2= \mbox{N}_{\sE}/ \mbox{D}_{\sE}$, where
\begin{eqnarray*}
  \mbox{N}_{\sE} &= & \int_{-\infty}^{+\infty}(-2\delta 1_{\{ y>t_0\}} + (2\delta -2) 1_{\{ y<t_0 \}}) \int_{-\infty}^{+\infty}(-2\delta 1_{\{ x>t_0\}} + (2\delta -2) 1_{\{ x<t_0 \}})  \\ 
    & & \hspace{1cm}\times [F_X(\min \{ x,y \})-F_X(x)F_X(y)]d_{\btau}(F_X(x))d_{\btau}(F_X(y))\mathrm{d}x \mathrm{d}y\\
    & =& 4\delta^2 \int_{t_0}^{+\infty} \int_{t_0}^{+\infty}
    [F_X(\min \{ x,y \})-F_X(x)F_X(y)]d_{\btau}(F_X(x))d_{\btau}(F_X(y))\mathrm{d}x \mathrm{d}y\\ 
    & & \quad- 8\delta(\delta-1)  \int_{t_0}^{+\infty} \int_{-\infty}^{t_0} [F_X(x)-F_X(x)F_X(y)]d_{\btau}(F_X(x))d_{\btau}(F_X(y))\mathrm{d}x \mathrm{d}y \\ 
    & & \quad + 4(\delta-1)^2
    \int_{-\infty}^{t_0} \int_{-\infty}^{t_0}[F_X(\min \{ x,y \})-F_X(x)F_X(y)]d_{\btau}(F_X(x))d_{\btau}(F_X(y))\mathrm{d}x \mathrm{d}y,
\\
 \mbox{D}_{\sE}&= & (-4\delta D_{\btau}(F_X(t_0))+2\delta+2D_{\btau}(F_X(t_0)))^2  .
\end{eqnarray*}
\end{corollary}

\begin{proof}
We have $\ell'_{\delta}(x,c)=-2 \delta \lvert x-c\rvert -2 (x-c) 1_{\{ x\leq c\}}$, for which $x\mapsto\ell'_{\delta}(x,c)$ is non-increasing. This means we can take $h_1(x,c)=0$ and $h_2(x,c)=-\ell'_{\delta}(x,c)$ for Assumption \textbf{(A3)}. If we integrate with respect to this measure, then we obtain
    $$\int_{a}^b g(x)\mathrm{d}\ell'_{\delta}(x,c) = \int_a^b (-2\delta 1_{\{ x>c\}} + (2\delta -2) 1_{\{x<c\}})g(x)\mathrm{d}x.$$
    This explains the factors in the numerator. By expanding the factors we obtain the second expression for the numerator $\mbox{N}_{\sE}$. Here we take into account that
    by changing the order of integration 
    \begin{align*}
         &4\delta(\delta-1) \int_{t_0}^{+\infty} \int_{-\infty}^{t_0} [F_X(\min \{ x,y \})-F_X(x)F_X(y)]d_{\btau}(F_X(x))d_{\btau}(F_X(y))\mathrm{d}x \mathrm{d}y\\
         &= 4\delta(\delta-1) \int_{-\infty}^{t_0} \int_{t_0}^{+\infty} [F_X(\min \{ x,y \})-F_X(x)F_X(y)]d_{\btau}(F_X(x))d_{\btau}(F_X(y))\mathrm{d}y \mathrm{d}x.
    \end{align*}
    This explains the term with factor $8\delta(\delta-1)$.
\end{proof}
\begin{remark}\label{remark:expectileVariance}
When using the expectile loss function, the corresponding generalized extremile equals the $\delta$th expectile of $X_{\sD_{\btau}}$. Hence, if we consider $D_{\btau}$ the uniform cumulative distribution function, we can use our estimator to estimate the $\delta$th expectile of $X$. Doing so, our estimator coincides with the estimator studied in \cite{HolzmannAndKlar2016}.
The latter authors also established the asymptotic normality of their estimator (see p. 2359 in \cite{HolzmannAndKlar2016}), with asymptotic variance given by
\begin{equation}\label{eq:expectilesRemarkVariance}
    \frac{\mathbb{E}\left [\left ( I_{\delta}(t_0,X)\right )^2\right ]}{\left [ \delta(1-F_X(t_0))+(1-\delta)(F_X(t_0))\right ]^2} , 
\end{equation}
where $$I_{\delta}(t_0,X)=\delta(X-t_0) 1_{\{X \geq t_0\}} -(1-\delta)(t_0-X)1_{\{X <t_0\}}.$$ In Section S3 of the Supplementary Material we prove that the asymptotic variance as stated in Corollary \ref{cor:expectilesAsymptoticVariance}, equals that in  \eqref{eq:expectilesRemarkVariance}, for the specific case that $D_{\btau}(u)=u$. Hence the asymptotic normality result for expectiles is obtained as special case of our Theorem \ref{theorem: asymptoticNormalityEmpirical} and Corollary \ref{cor:expectilesAsymptoticVariance}. 
\end{remark}

Applying Theorem \ref{theorem: asymptoticNormalityEmpirical}  to the loss function
$\ell_{\delta}(x,c)=-c \lvert x - b \rvert^\delta +c^2/2$ establishes the asymptotic normality results for the proposed estimator for $\mathbb{E}[\lvert X_{\sD_{\btau}} - b \rvert^\delta ]$. 
\begin{corollary}\label{cor: generalizedDistrortionAsymptoticNormality}
    Consider the loss function $\ell_{\delta}(x,c)=-c \lvert x - b \rvert^\delta +c^2/2$, for $\delta>0$. This corresponds to $t_0=\mathbb{E}[\lvert X_{\sD_{\btau}} - b \rvert^\delta ]$. Consider a random variable $X$ and a cumulative distribution function $D_{\btau}$ such that the assumptions of Theorem \ref{theorem: asymptoticNormalityEmpirical} are satisfied.
    Then the asymptotic variance equals 
    \begin{align*}
         &\hspace*{-0.2 cm}\int_{-\infty}^{t_0} \int_{-\infty}^{t_0} [F_{h_1(X,t_0)}(\min\{x,y\})-F_{h_1(X,t_0)}(x)F_{h_1(X,t_0)}(y)]d_{\btau}(F_{h_1(X,t_0)}(x))d_{\btau}(F_{h_1(X,t_0)}(y))\mathrm{d}x\mathrm{d}y\\
         &\hspace*{-0.2 cm}+\int_{0}^{+\infty} \int_{0}^{+\infty} [F_{h_2(X,t_0)}(\min\{x,y\})-F_{h_2(X,t_0)}(x)F_{h_2(X,t_0)}(y)]d_{\btau}(F_{h_2(X,t_0)}(x))d_{\btau}(F_{h_2(X,t_0)}(y))\mathrm{d}x\mathrm{d}y\\
         &\hspace*{-0.2 cm}-2\int_{0}^{+\infty} \int_{-\infty}^{t_0} [\min\{F_{h_1(X,t_0)}(x),F_{h_2(X,t_0)}(y) \}-F_{h_1(X,t_0)}(x)F_{h_2(X,t_0)}(y)]d_{\btau}(F_{h_1(X,t_0)}(x))d_{\btau}(F_{h_2(X,t_0)}(y))\mathrm{d}x\mathrm{d}y,
    \end{align*}
    where $$F_{h_1(X,t_0)}(z)=\begin{cases}
        1 &\text { if } z \geq t_0\\
        F_X(b-(t_0-z)^{1/\delta}) &\text{ if } z<t_0
    \end{cases}$$
    and
    $$F_{h_2(X,t_0)}(z)=\begin{cases}
        0 &\text { if } z < 0 \\
        F_X(z^{1/\delta}+b) & \text{ if } z\ge 0.
    \end{cases}$$
\end{corollary}
\begin{proof}
We calculate $\sigma_{t_0}^2$ using (\ref{eq: sigmat02}), where
we use $\ell_\delta'(x,c)=h_1(x,c)-h_2(x,c)$ with 
$$h_1(x,c)=-(-x+b)^\delta 1_{\{x\leq b\}}+c \hspace{0.5cm} \text{ and } \hspace{0.5cm} h_2(x,c)=(x-b)^\delta 1_{\{x>b\}},$$
non-decreasing and left-continuous in $x$.
First, we find an expression for $F_{h_1(X,t_0)}(z)$.  This is done by observing that $h_1(x,t_0)\leq t_0$, consequently $F_{h_1(X,t_0)}(z)=1$ for $z\geq t_0$. For $z< t_0$ the result follows by noting that 
\[
    F_{h_1(X,t_0)}(z)=P((-X+b)^\delta 1_{\{X\leq b\}} \geq -z+t_0)
    =P(-X+b\geq (-z+t_0)^{1/\delta})
    =F_X(b-(t_0-z)^{1/\delta}).
\]
Similarly, we obtain the stated expression for $F_{h_2(X,t_0)}(z)$. 
Furthermore, it is clear that 
$\lambda_{F}'(c)=1$, and application of Theorem \ref{theorem: asymptoticNormalityEmpirical} concludes the proof.
\end{proof}

\section{Simulation study}\label{sec: simulations} 
\subsection{Aims of the simulation study, and simulation models} 

In this section we investigate the finite-sample behaviour of the estimator $T_n^*$ of the generalized extremile $e_{\btau, \delta}(X; D_{\btau}, \ell_{\delta})$. We conduct five small simulation studies considering different pairs for $(D_{\btau}, \ell_{\delta})$, and hence looking into estimating different quantities. Tables \ref{Table: Studies1--3En5} and \ref{Table: Study4}  summarize the settings for the various parts (simulation studies (St) 1--5). The various scenarios  have the following aims: investigating
\begin{itemize}
\item[\textbf{St1}]  the performance of $T_n^*$ for different choices of $F_X$; 
\item[\textbf{St2}] the performance of $T_n^*$ for different choices of parameters $(\tau, \delta)$; 
\item[\textbf{St3}]  estimation of quantiles via two different pairs $(D_{\btau}, \ell_{\delta})$; 
\item[\textbf{St4}] the asymptotic normality result for quantile estimation under censoring; 
\item[\textbf{St5}] the asymptotic normality result when using as loss function 
$-c \lvert x \rvert +c^2/2$ (see Corollary \ref{cor: generalizedDistrortionAsymptoticNormality}). 
\end{itemize}
In each simulation study we draw 500 samples of size $n$ from the simulation model. We adapt various ways to present the results. We evaluate the finite-sample performance of $T_n^*$ by reporting on the bias, variance and Mean Squared Error (MSE) across the Monte Carlo runs. For illustration of the sampling distribution of $T_n^*$ we present a kernel density estimate based on the 500 values for $\sqrt{n} (T_n^* - t_0)$.  

\begin{table}[h]
\centering
\caption{Model elements in Simulation studies  1--3 and 5.}
\label{Table: Studies1--3En5}
\vspace*{0.2 cm} 

\noindent
{\small
\hspace*{-1 cm}
\begin{tabular}{|l|l|l|ll|l|}
\hline
& \multicolumn{5}{c|}{Simulation study (St)}\\
\cline{2-6}
& St1 & St2 & \multicolumn{2}{c}{St3} & St5 \\
\hline
$\ell_{\delta}(x,c)$ & expectile  & expectile  & absolute  value &  quantile  & \\
&loss & loss & loss & loss  & loss\\[1.4 ex]
&  $ \left|\delta-1_{\{x\le c\}}\right| |x-c|^2$ & same as in St1 & $|x-c|$  & $\left|\delta-1_{\{x\le c\}}\right| |x-c|$ & $ -c|x|+c^2/2$ \\
& $\delta =0.9$ & $\delta=0.1, 0.9, 0.95$ & &  $\delta=0.01, 0.05, 0.1, $  & \\
& & & & $\, 0.5, 0.9, 0.95, 0.99 $   & \\
& & & & &  \\[2 ex]
$D_{\btau}(u)$ & $K_{\tau}(u)$ in \eqref{DtauExtremiles}& $D_{\tau}(u)$ & $K_{\tau}(u)$ in \eqref{DtauExtremiles} & $D_{\btau}(u)=u$     & same as in St2\\
& $\tau=0.1, 0.9, 0.95 $ & $ \,\,  =(1-\tau)^{-1} (u-\tau) \, 1_{\{u \ge \tau \}}$    &  $\tau=0.01, 0.05, 0.1, $ & &  \\
& & $\tau=0.1, 0.9, 0.95 $ &  $\, 0.5, 0.9, 0.95, 0.99 $ &  &  \\[2 ex]
$F_X$ & $X\sim \mathcal{N}(0; 1)$  &$X\sim \mathcal{N}(0; 1)$   & \multicolumn{2}{c|}{$X \sim \mbox{Expo}(1)$}  &  same as in St1\\
&  or $X \sim \mbox{Expo}(1)$ & &  &  & \\[2 ex] 
$n$ & $n=50 \, \mbox{or} \, 800$ & same as in St1  & \multicolumn{2}{c|}{$n=50 \, \mbox{or} \, 400$}  &  same as in St1  \\
\hline 
\end{tabular}}
\end{table}

\begin{table}[htb]
\centering
\caption{Model elements in Simulation study 4 (St4).}
\label{Table: Study4} 
\vspace*{0.2 cm} 

\noindent
\begin{tabular}{|l|l|l|l|l|l|}
\hline
$\ell_{\delta}(x,c)$  & $D_{\btau}(u)$ &  $F_X$ & $F_C$  & $Y$ & $n$ \\[1.2 ex]
\hline 
adapted quantile & $D_{\btau}(u)=u$ & $X\sim \mbox{Expo}(1)$ & $C\sim \mbox{Expo}(\frac{p_C}{1-p_C})$  & $ Y= \min(X,C)$ & $n=100 \, \mbox{or} \, 400$ \\[1.2 ex]
loss; see \eqref{eq: QuantileLossCensoring} & & &  $p_C=0.1, 0.3, 0.5$  & & \\
$\delta=0.1, 0.5$ & & & &  &  \\ 
\hline 
\end{tabular}
\end{table}

\subsection{Simulation results} 
\subsubsection*{Simulation study 1}
 Figure \ref{fig:densityEstimates} shows a kernel density estimate of  $\sqrt{n}(T_n^*-t_0)$ for $n=50,800$. The left (respectively right) panels are for samples drawn from $X\sim \mathcal{N}(0,1)$ (respectively $X\sim \mbox{Expo}(1)$). Presented are the results for sample size $n=50$ (blue lines) and $n=800$ (red lines), together with the normal density with variance equal to the asymptotic variance of Theorem \ref{theorem: asymptoticNormalityEmpirical}. Note that with increasing $n$ the kernel density estimates approach the asymptotically normal density. To be noted is also that for values of $\tau$ closer to 1, the variance of the sampling distribution of $T_n^*$ as well as the asymptotic variance is largest (see the scale of the horizontale axis). Table \ref{table:MSEtableNormalAndExp} in the Supplementary Material provides the estimated bias, variance and MSE  of the estimator $T_n^*$. 
\begin{figure}[htb]
    \centering
    \includegraphics[width=0.44\textwidth]{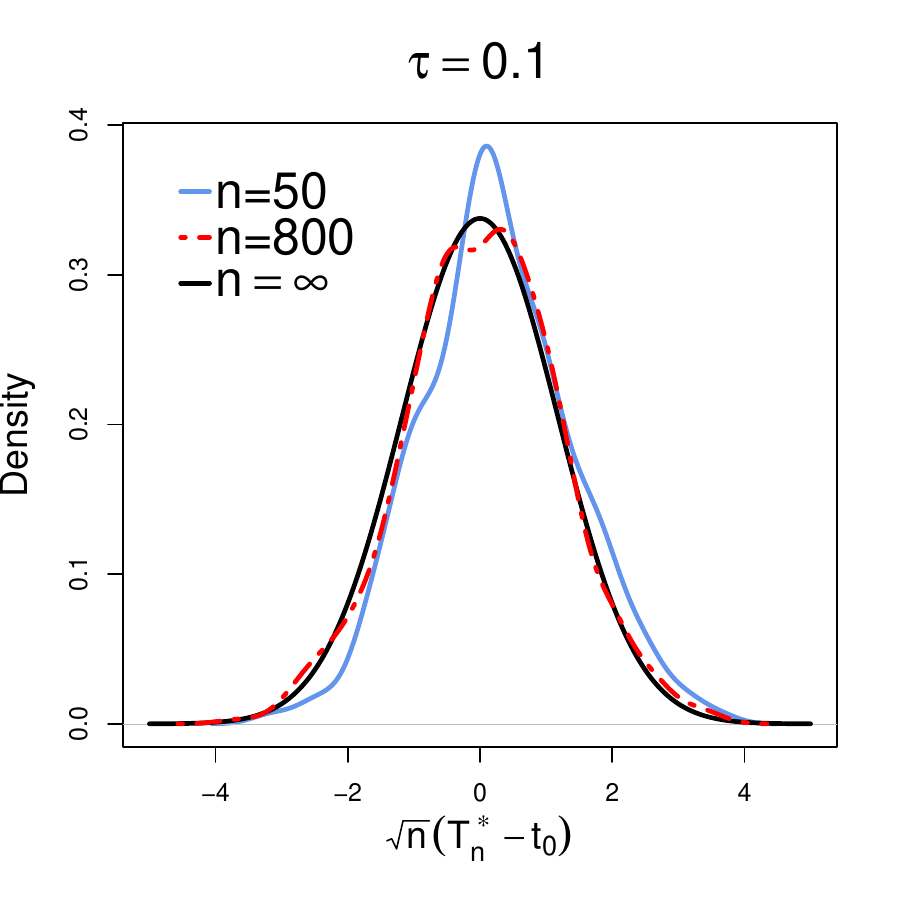}
    \includegraphics[width=0.44\textwidth]{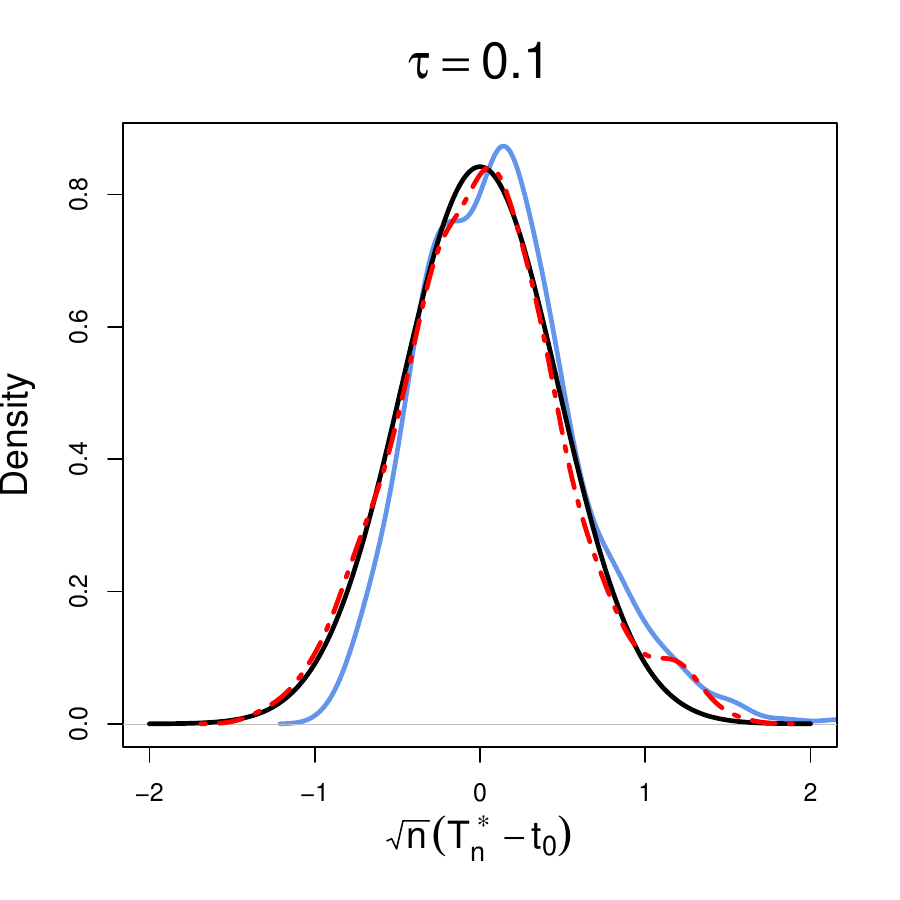}\\
\includegraphics[width=0.44\textwidth]{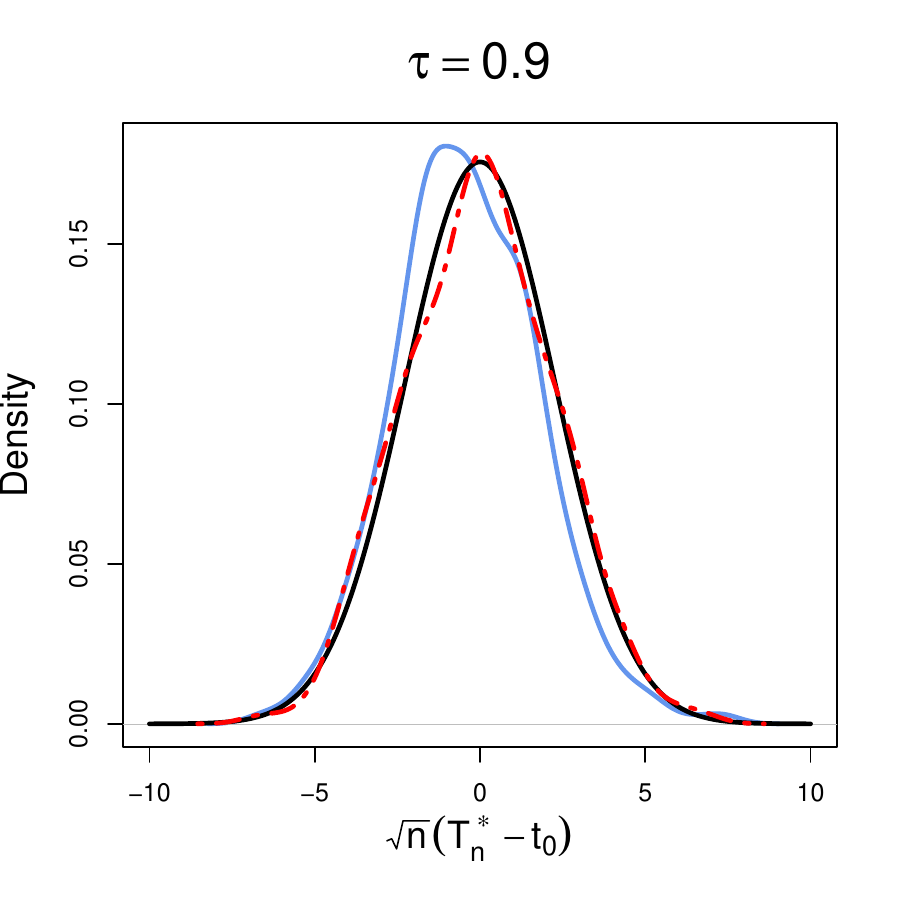}
    \includegraphics[width=0.44\textwidth]{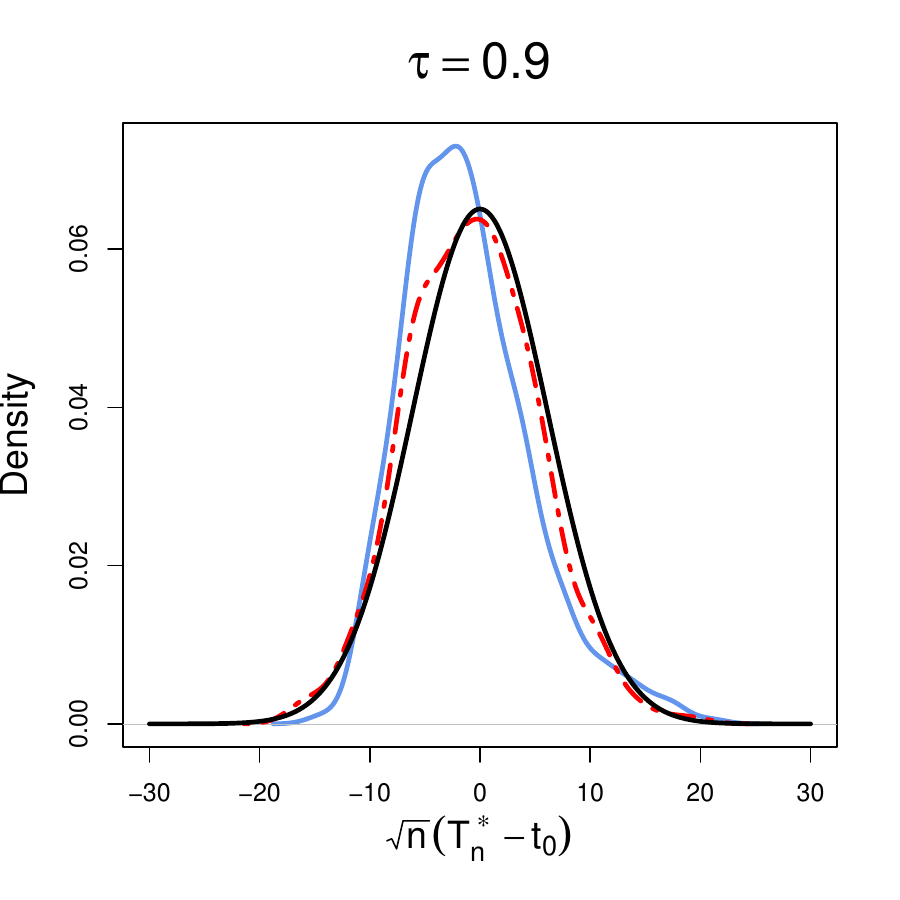}  \\
    \includegraphics[width=0.44\textwidth]{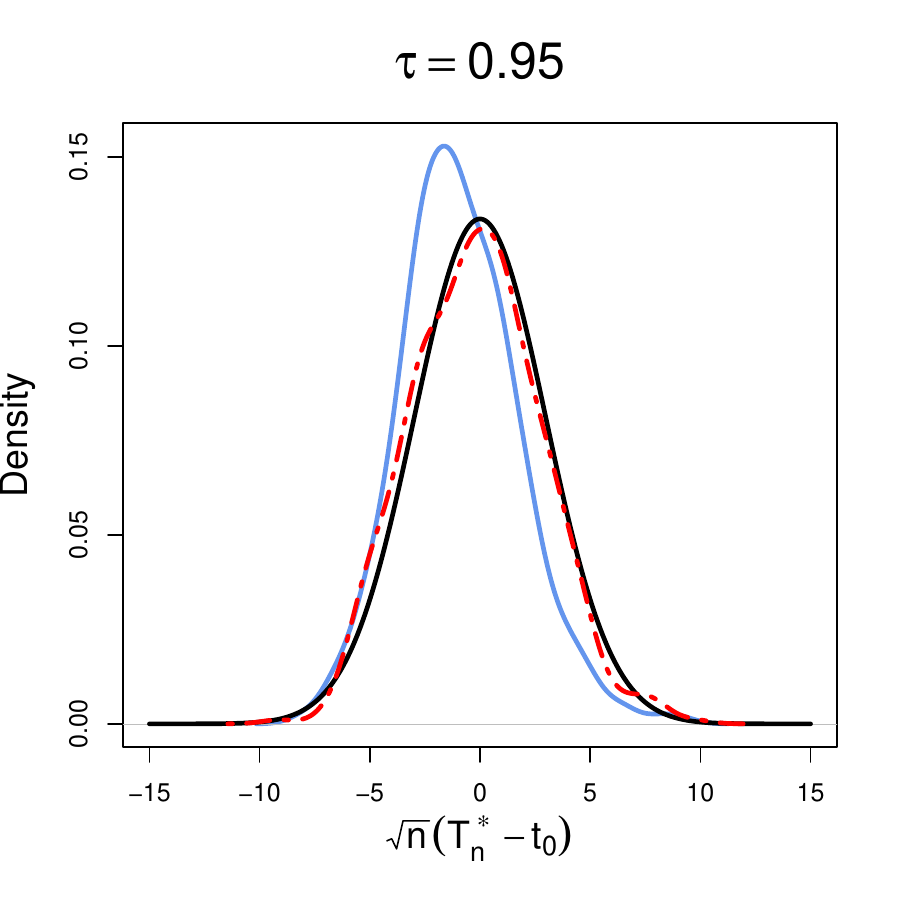}
    \includegraphics[width=0.44\textwidth]{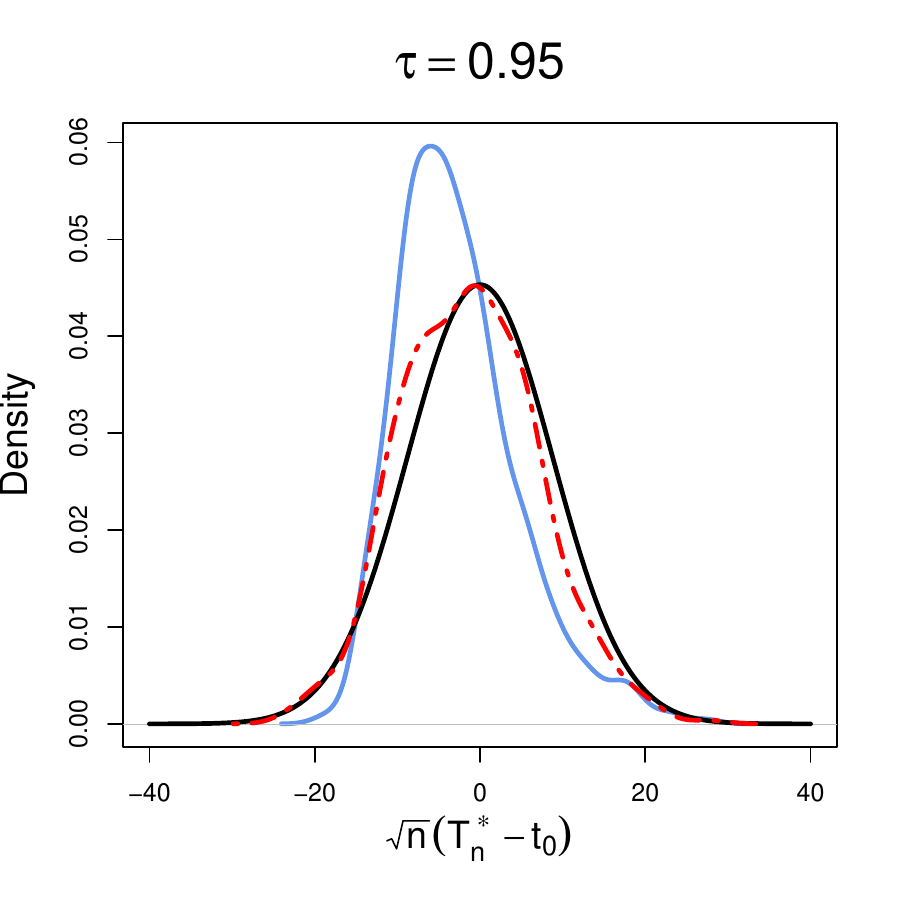}     
    \vspace*{-0.64 cm}
    
    \noindent
    \caption{Simulation study 1. 
    Density estimates based on 500 values of $\sqrt{n}(T_n^*-t_0)$ for $n=50$ and $n=800$. Samples drawn from $X\sim \mathcal{N}(0;1)$ (left panels); and from  $X \sim \mbox{Expo}(1)$ (right panels).  Black lines: the normal density with variance equal to the asymptotic variance of Theorem \ref{theorem: asymptoticNormalityEmpirical}. }
    \label{fig:densityEstimates}
\end{figure}

\subsubsection*{Simulation study 2}
In this part, the function $D_{\tau}$ corresponds to the distortion function of expected shortfall (see entry 5 in Table \ref{DistrFunctions}).
Several values of the parameters $(\tau, \delta)$ of respectively the distribution and the loss function are considered. The simulation results are summarized in 
Figure \ref{fig:boxplotExpectilesAndESn50And800FromNormal} as boxplots of  $T_n^*-t_0$ (with $t_0$ the true generalized extremile) over all Monte Carlo runs. The value of $t_0$ is indicated on the right-hand side of each plot. Note that when passing from sample size $n=50$ to  $n=800$ all boxplots become narrower, and more positioned around zero. For $n=50$, the median is close to zero for the cases where $\delta=0.1$ and/or $\tau=0.1$. In all cases the boxplots are relatively symmetric.  As already seen from Figure \ref{fig:densityEstimates}, for fixed $\delta$, our estimates are more variable for larger values of $\tau$. Similarly, for fixed $\tau$, the estimates are more variable  for larger  $\delta$.
\begin{figure}[htb]
    \centering
    \includegraphics[scale=0.44]{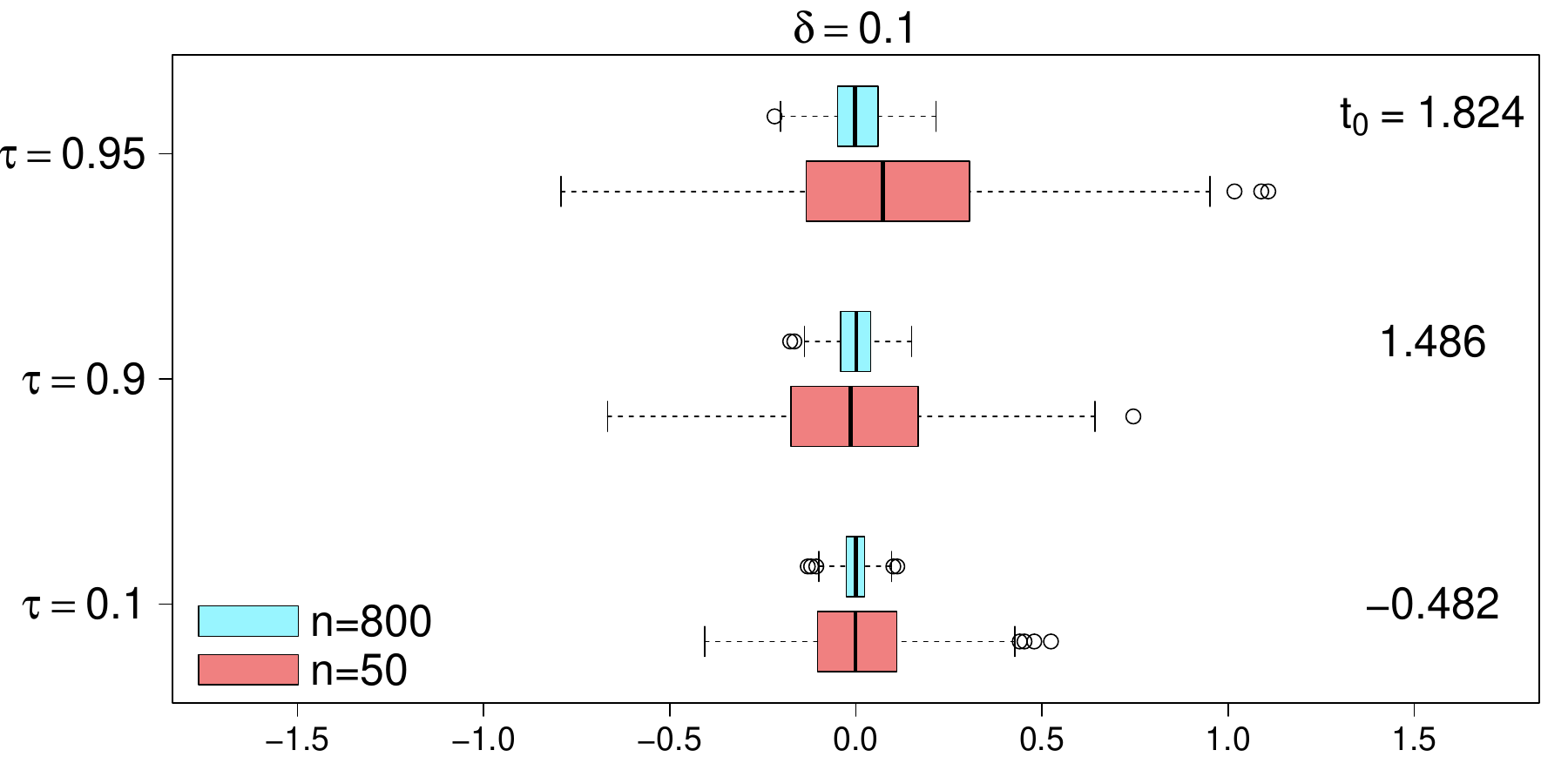}\\
  \includegraphics[scale=0.44]{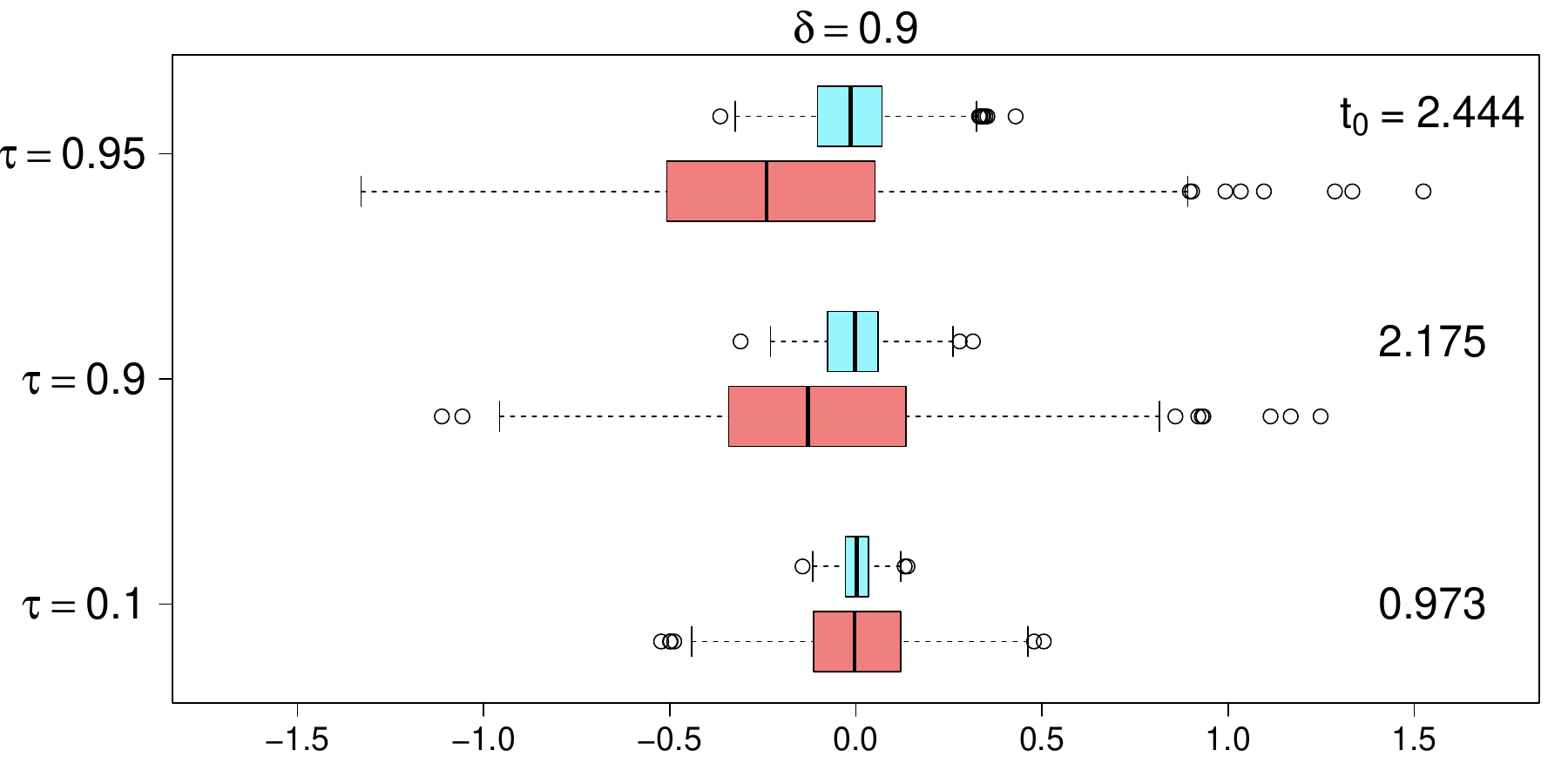}\\
    \includegraphics[scale=0.44]{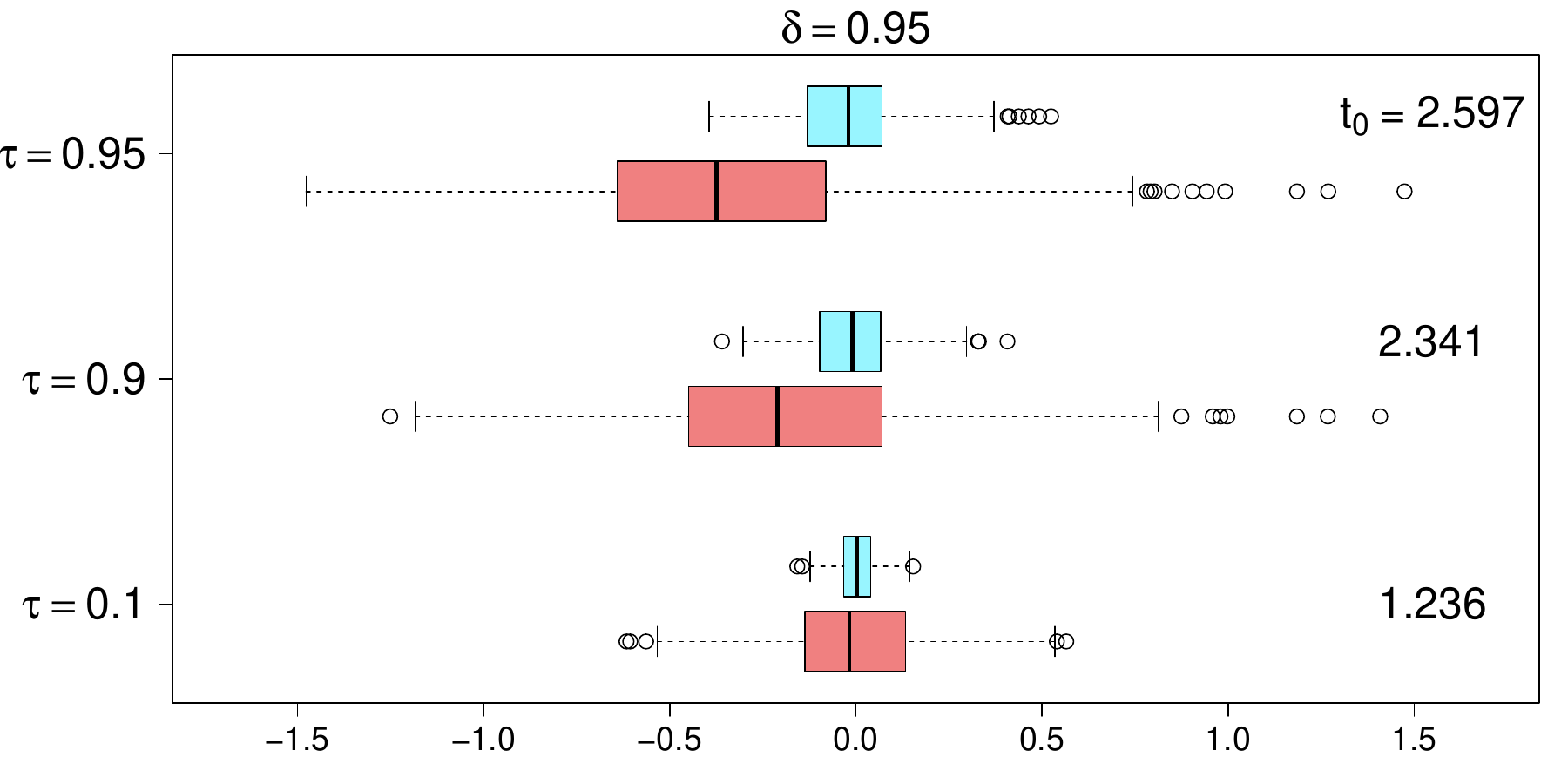} 
    \vspace*{-0.44 cm}
    
    \noindent
    \caption{Simulation study 2. Finite-sample performance of the estimator $T_n^*$. Boxplots of $T_n^*-t_0$ based on 500 samples of size $n=50,800$ drawn from $X\sim \mathcal{N}(0;1)$  
    for different combination of parameters $(\delta,\btau)$. The true value of $t_0$ is indicated next to each boxplot.}
    \label{fig:boxplotExpectilesAndESn50And800FromNormal}
\end{figure}
\FloatBarrier
\subsubsection*{Simulation study 3}
As already mentioned at the start of Section \ref{sec:estimation--ApplicationsANResult} quantiles can be estimated from two viewpoints, i.e. using two different pairs $(D_{\btau}, \ell_{\delta})$. We present simulation results of the two resulting estimators:
\begin{itemize}
\item[] \textbf{Estimator 1}: based on using $D_{\btau}(u)=u$ and the quantile loss function $\ell_{\delta}(x,c) = | \delta- 1_{\{ x \le c\}} \, |x-c| $;
\item[] \textbf{Estimator 2}: based on using $D_{\btau}=K_{\tau}$ and absolute value loss function $\ell_{\delta}(x,c) =|x-c|$.
\end{itemize}
 By Corollaries \ref{cor:absoluteValueLossAsymptoticVariance} and  \ref{cor:QuantileLossAsymptoticVariance}, and Remarks \ref{AVarQuantileEstimators}  and  \ref{AVarQuantileEstimatorsBis}  we know that both estimators are asymptotically equivalent to the empirical quantile estimator.
\FloatBarrier

Table \ref{table:MSEtableExpComparingQuantiles} lists the  mean squared error (MSE) for the two estimators, together with their Monte Carlo approximated variance (between brackets). The results for both estimators for sample size $n=50$ 
are identical, except for quantile level $0.01$. When the sample size increases to $n=400$, the estimators are indistinguishable when only considering the mean squared error. The difference between the two estimators for quantile level $0.01$ is due to a larger bias of estimator 1.
\begin{table}[htb]
\centering
\caption{Simulation study 3. Mean squared error (MSE) of two different quantile estimators ($\tau$th = $\delta$th quantile). The approximate variance of the estimators is given between brackets.}
\label{table:MSEtableExpComparingQuantiles}
\vspace*{0.2 cm}

\noindent
\begin{tabular}{l|cc}
               & \multicolumn{2}{c}{Estimator 2}  \\
$\btau$ & $n=50$     & $n=400$                          \\ \hline
$0.01$       & $1.01 \cdot 10^{-3}$  $(5.65 \cdot 10^{-4})$   & $3.68\cdot  10^{-5} (3.08 \cdot 10^{-5})$          \\
$0.05$       & $1.99 \cdot 10^{-3}$ $(1.49 \cdot 10^{-3})$    & $1.33 \cdot 10^{-4} $ $(1.28 \cdot 10^{-4})$                                           \\
$0.1$        & $ 2.53 \cdot 10^{-3}$ $(2.35 \cdot 10^{-3})$    & $2.69 \cdot 10^{-4}$ $(2.69 \cdot 10^{-4})$                                           \\
$0.5$        & $1.98 \cdot 10^{-2}$  $(1.96 \cdot 10^{-2})$    & $2.62 \cdot 10^{-3}$ $(2.61 \cdot 10^{-3})$                                             \\
$0.9$        & $1.66 \cdot 10^{-1}$  $(1.66 \cdot 10^{-1})$     & $2.27 \cdot 10^{-2} $    $(2.25 \cdot 10^{-2})$                                              \\
$0.95$       & $3.08 \cdot 10^{-1}$ $(2.69 \cdot 10^{-1})$    & $4.70 \cdot 10^{-2}$  $(4.67 \cdot 10^{-2})$                                              \\
$0.99$       & 1.26  $(9.29 \cdot 10^{-1})$      & $1.97 \cdot 10^{-1}$   $(1.97 \cdot 10^{-1})$  \\                          
\hline \hline 
\multicolumn{2}{c}{}\\[0.02cm]
 & \multicolumn{2}{c}{Estimator 1} \\   
$\btau$ & $n=50$     & $n=400$ \\ \cline{1-3} \\
0.01 & $9.54 \cdot 10^{-1} (6.75 \cdot 10^{-4})$ & $3.68\cdot 10^{-5} (3.08 \cdot 10^{-5}) $\\
Other & \multicolumn{2}{c}{Identical to results for Estimator 2}
\end{tabular}%
\end{table}

\subsubsection*{Simulation study 4}
The adaptive loss function for quantile estimation in case of censoring was briefly discussed in Section \ref{sec: ExamplesLossFunctions}. This loss function is not convex, and hence does not satisfy the conditions stated in Theorem \ref{theorem: asymptoticNormalityEmpirical}. Since these are however sufficient conditions, we might wonder whether the finite-sample distribution can nevertheless be well approximated with the asymptotic normal distribution as stated in the theorem. 

In the right-random censoring setting, the variable of interest $X$ is possibly censored by a  variable $C$. We are interested in a $\delta$th quantile of $X$, but only have data of the form $(Y,\Delta_C)$, where $Y=\min\{X,C\}$ and $\Delta_C=1_{\{X \leq C\}}$. Given a sample from $(Y,\Delta_C)$, we can use our estimator to estimate such a $\delta$th quantile of $X$.  Table \ref{Table: Study4} summarizes the various elements of the simulation model. 

When calculating the estimator we have to deal with the non-convexity of $\lambda_F(c)$ and $\lambda_{\widehat{F}_n}^*(c)$. 
Globally minimizing $\lambda_{\widehat{F}_n}^*(c)$ is not feasible, and instead  we implemented  the following procedure. Based on a sample from $Y$, we construct a sequence $\boldsymbol{c}=(c_1,\cdots,c_k)$ of values ranging from $
\min(Y_i)$ to $\max(Y_i)$. More precisely we take a equispaced grid of points, with an increment of 0.01, i.e. for $m, m^* \in \mathbb{N}$, with $m^* > m$, we have $ c_{m^*}-c_m=(m^*-m) 0.01$. We then  search for the smallest value in this sequence, say $c_m$, which satisfies the following inequalities 
$$\frac{\lambda_{\widehat{F}_n}^*(c_{m+j})-\lambda_{\widehat{F}_n}^*(c_{m})}{j\cdot 0.01}>0.01 \quad \text{ for } j=1,\cdots,20.$$
This requires difference quotients (discrete slopes) to be strictly positive in an interval past a given point. In this manner, we search for a (local) minimum. By considering multiple subsequent  difference quotients, we avoid selecting a value due to a small spike in the function $\lambda_{\widehat{F}_n}^*(c)$.\\

Recall that $f_X, F_X$ denote respectively  the density and the cumulative distribution function of $X$. Similarly, we use the notations $f_C,F_C$ for these quantities for the censoring variable $C$.
The asymptotic variance stated in Theorem \ref{theorem: asymptoticNormalityEmpirical} equals in this case
\begin{equation}\label{eq:asymptoticVarianceMyCensoring}
    \frac{F_X(t_0)(1-F_X(t_0))}{(f_X(t_0)-(1-\delta)f_C(t_0))^2}.
\end{equation}
When there is no censoring (meaning that $F_C(t)=0$ for all $t$), this expression reduces to the asymptotic variance of the empirical quantile estimator.\\

\cite{SanderJM1975} estimated a quantile for censored data relying on the Kaplan-Meier estimator for a distribution function. We refer to this quantile estimator as the Kaplan Meier-based (KM-based) estimator. The author established asymptotic normality for the estimator (see Corollary 1 on p. 5 in \cite{SanderJM1975}). The asymptotic variance of this estimator equals
\begin{equation}\label{eq:asymptoticVarianceOtherCensoring}
  \frac{(1-\delta)^2}{\left( f_X(t_0)\right )^2}\int_0^{t_0} \frac{f_X(x)}{(1-F_X(x))^2 (1-F_C(x))}\mathrm{d} x.  
\end{equation}
Again, in case of no censoring, this asymptotic variance reduces to the one of the empirical quantile estimator.

\FloatBarrier
\begin{table}[h!]
\centering
\caption{Simulation study 4. Mean squared error (MSE) for the loss-based estimator and the KM-based estimator for different censoring proportions $p_C$, different quantile levels $\delta$ and different sample sizes $n$. Variances are given  between brackets, and should be multiplied with the same factor $10^{-a}$ as appearing in the corresponding MSE-value. }
\label{table:MSESimulation4}
\vspace*{0.2 cm}

\noindent
\begin{tabular}{l|l|cc|cc|cc|cc}
                           &              & \multicolumn{4}{c|}{MSE Loss-based estimator} & \multicolumn{4}{c}{MSE KM-based estimator} \\
 $p_C$ & $\delta$   & \multicolumn{2}{c}{$n=100$}          & \multicolumn{2}{c|}{$n=400$}          & \multicolumn{2}{c}{$n=100$}        & \multicolumn{2}{c}{ $n=400$}         \\ \hline
\multirow{4}{*}{$0.1$} 
& & & & & \\
& $0.1$ &  $1.43 \cdot 10^{-3}$ & ($1.39$)            & $3.53 \cdot 10^{-4}$ & ($3.52$)         &  $1.14 \cdot 10^{-3}$  & ($1.13$)       & $3.17 \cdot 10^{-4}$ & ($3.17$)        \\
                & $0.5$ & $1.20 \cdot 10^{-2}$ & ($1.19$)            & 
$2.90 \cdot 10^{-3}$  & ($2.90$)          
& $9.90 \cdot 10^{-3}$  & ($9.85$)       &$2.51 \cdot 10^{-3}$            & ($2.51$) \\[1.2 ex]
\cline{1-2}                          
\multirow{4}{*}{$0.3$}& & & & & & & \\
& $0.1$ &  $2.09 \cdot 10^{-3}$   & ($2.01 $)           & 
$4.78 \cdot 10^{-4}$  & ($4.76 $)          & $1.23 \cdot 10^{-3}$  
& ($1.23 $)       & $3.13 \cdot 10^{-4}$ & ($3.12$)         \\
                           & $0.5$ &  $1.74 \cdot 10^{-2}$  & ($1.73$)           & 
$4.24 \cdot 10^{-3}$  & ($4.24$)         & $1.15 \cdot 10^{-2}$    & ($1.15 $)                       & $2.76 \cdot 10^{-3}$     & ($2.76$)     \\[1.2 ex]
\cline{1-2}                   
\multirow{4}{*}{$0.5$} & & & & & & & \\
& $0.1$ &   $4.23 \cdot 10^{-3}$  & ($3.98$)            & 
$6.57 \cdot 10^{-4}$    & ($6.40$)        &  $1.41 \cdot 10^{-3}$ & ($1.37$)                         & $3.40 \cdot 10^{-4}$   & ($3.38$)      \\
& $0.5$ &  $3.34 \cdot 10^{-2} $  & ($3.28$)          & 
  $8.37 \cdot 10^{-3}$  & ($8.32$)           & $1.69 \cdot 10^{-2}$  & ($1.65 $)                        &  $3.64 \cdot 10^{-3}$ & ($3.64$)    \\                    
\end{tabular}
\end{table}
\begin{figure}[h!]
    \centering
    \includegraphics[scale=0.56]{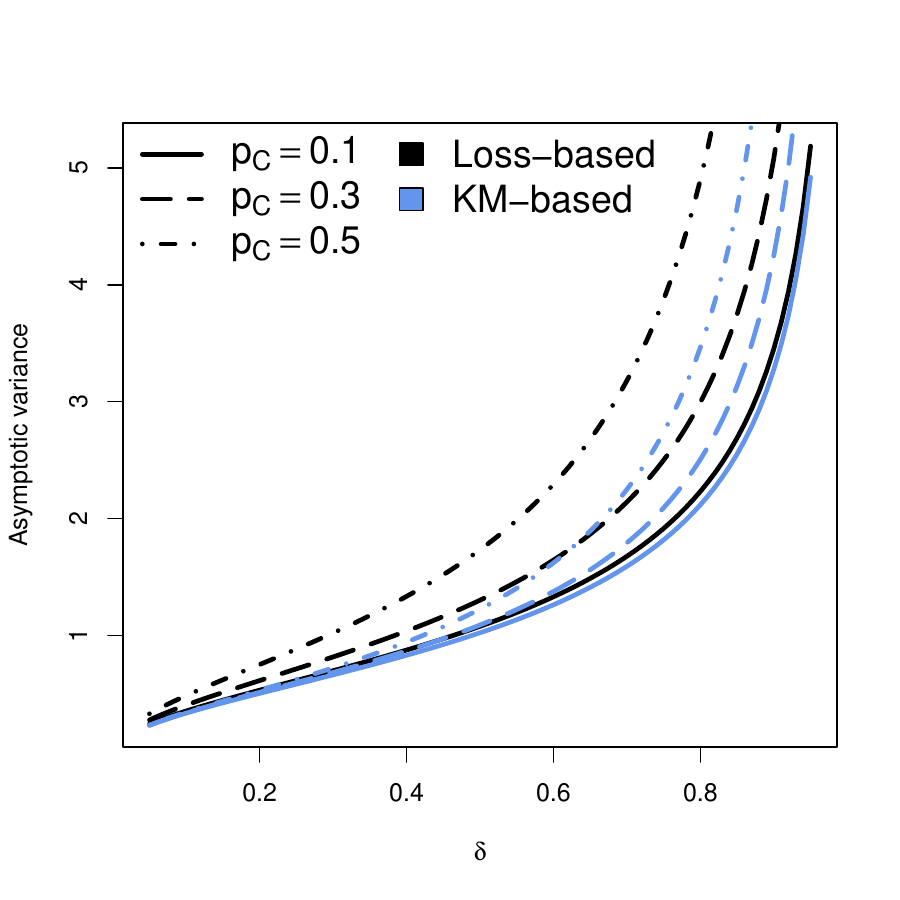}
    \vspace*{-0.94 cm}
    
    \noindent
    \caption{Simulation 4. Asymptotic variance expressions, i.e. (\ref{eq:asymptoticVarianceMyCensoring}) and (\ref{eq:asymptoticVarianceOtherCensoring}) for different values of $\delta$ and different censoring proportions $p_C$. Distributions of $X$ and $C$ as in Table \ref{Table: Study4}.}
    \label{fig:comparingAsymptoticVarianceCensoring}
\end{figure}

Figure \ref{fig:adaptiveQuantileLossDensityPlots} in the Supplementary Material depicts the density estimates of the loss-based estimator $\sqrt{n}(T_n^*-t_0)$, as well of the KM-based estimator. The figure also superimposes the asymptotic normal distributions with asymptotic variances as in \eqref{eq:asymptoticVarianceMyCensoring} and  \eqref{eq:asymptoticVarianceOtherCensoring} respectively. These results confirm the appropriateness of the asymptotic variance \eqref{eq:asymptoticVarianceMyCensoring}, and seem to indicate that the asymptotic normality result of Theorem \ref{theorem: asymptoticNormalityEmpirical} holds, even in this case where the loss function does not satisfy the stipulated conditions.  

Some findings from the simulation results are as follows.  
\begin{itemize}
\item[$\bullet$]
The  KM-based estimator shows a smaller (finite-sample) variance compared to the loss-based estimator.
\vspace*{-0.54 cm}

\noindent
\item[$\bullet$] 
 For the different censoring proportions, the asymptotic variance of both estimators is larger when $\delta$ is $0.5$.  
 \vspace*{-0.34 cm}

\noindent
\item[$\bullet$]
The finite-sample variances of both estimators increase with increasing censoring proportion.
\end{itemize}

Table \ref{table:MSESimulation4} provides a summary of bias, variance and MSE for both estimators, and confirms the visual findings from Figure \ref{fig:adaptiveQuantileLossDensityPlots}.

Figure \ref{fig:comparingAsymptoticVarianceCensoring} presents the asymptotic variances in (\ref{eq:asymptoticVarianceMyCensoring}) and (\ref{eq:asymptoticVarianceOtherCensoring}) for different quantile levels $\delta$ and different censoring proportions $p_C$.  As to be expected the (asymptotic) variances of both estimators increase as the censoring proportion increases. For small censoring proportion $p_C=0.1$ the asymptotic variances are almost identical. As the censoring proportion increases, so does the difference between the two asymptotic variances.  For larger censoring proportions $p_C$ it becomes more advantageous to use the KM-based estimator. We end this simulation regarding censored quantiles by noting that in \cite{DeBackerEtAL2019}, the authors established, in a linear regression setting, asymptotic normality of an estimator using the adapted quantile loss. They coped with non-convexity by using a majorize-minimize (MM) algorithm.

\subsubsection*{Simulation study 5}
With this simulation we aim to illustrate the result in Corollary \ref{cor: generalizedDistrortionAsymptoticNormality}. 
Figure \ref{fig:denistyEstimatesSimulation5} shows density estimates (based on 500 values) of $\sqrt{n}(T_n^*-t_0)$ for $n=50,800$. The left panels are for samples drawn from $X\sim \mathcal{N}(0;1)$. The right panels are for samples drawn from $X\sim \mbox{Expo}(1)$. The normal density with variance equal to the asymptotic variance is presented as the black lines. 
\begin{figure}[h!]
    \centering
    \includegraphics[width=0.44\textwidth]{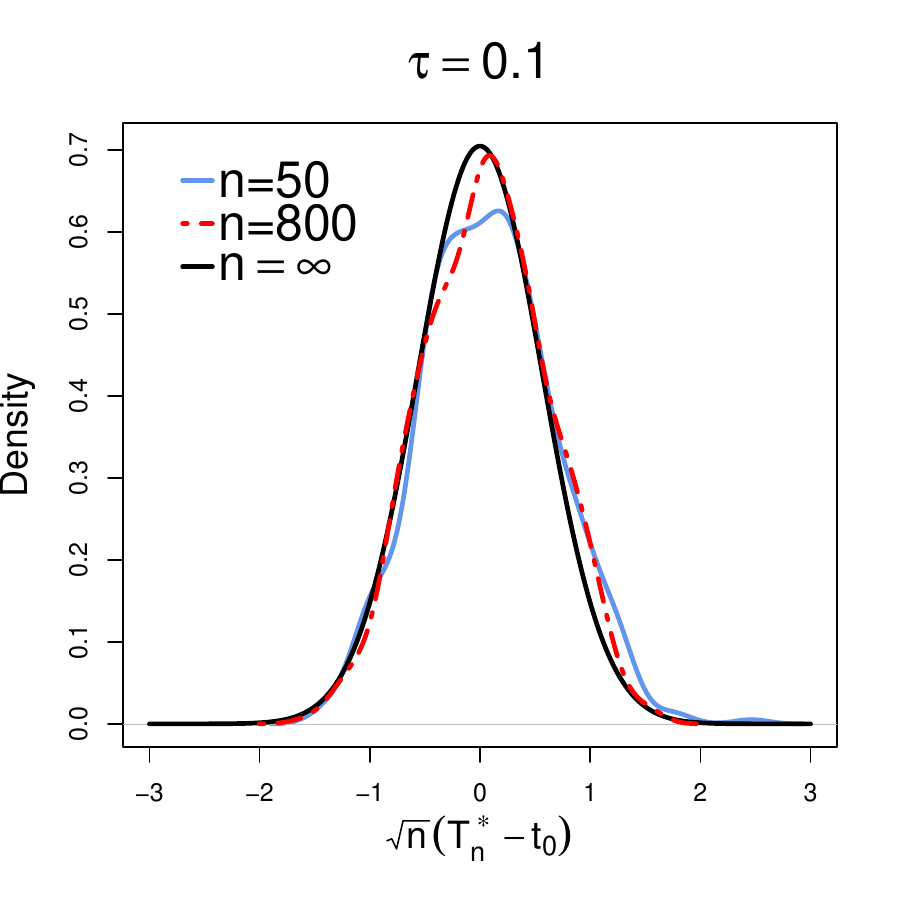}
    \includegraphics[width=0.44\textwidth]{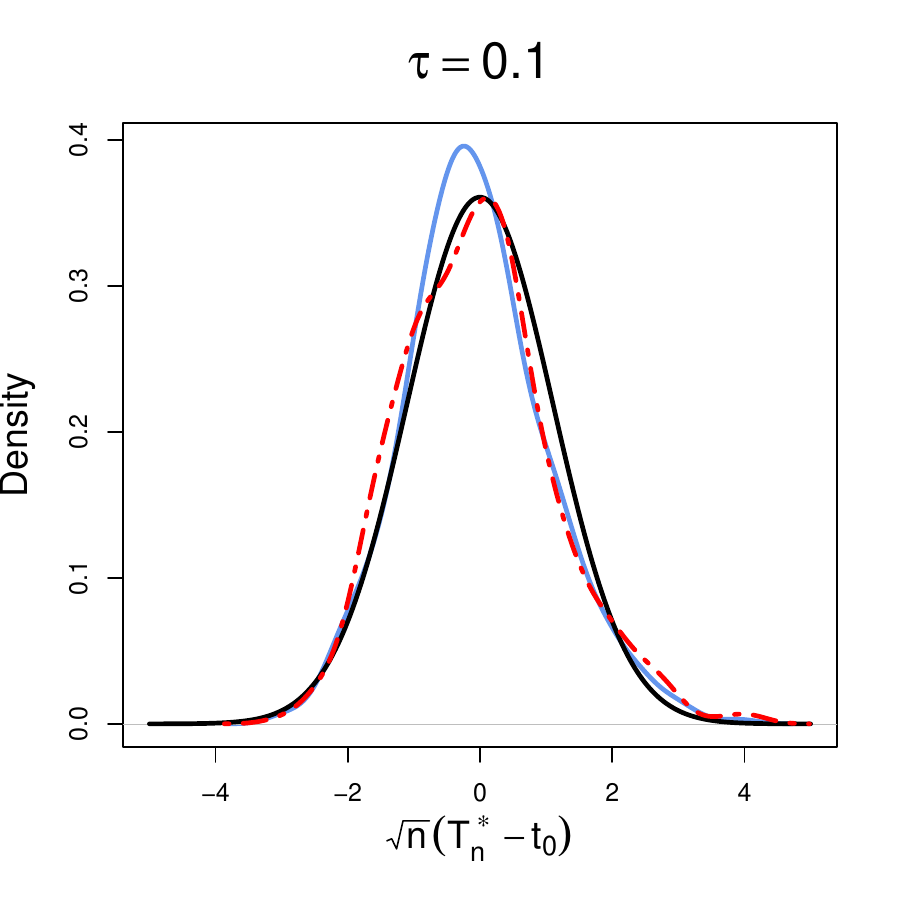}\\
\includegraphics[width=0.44\textwidth]{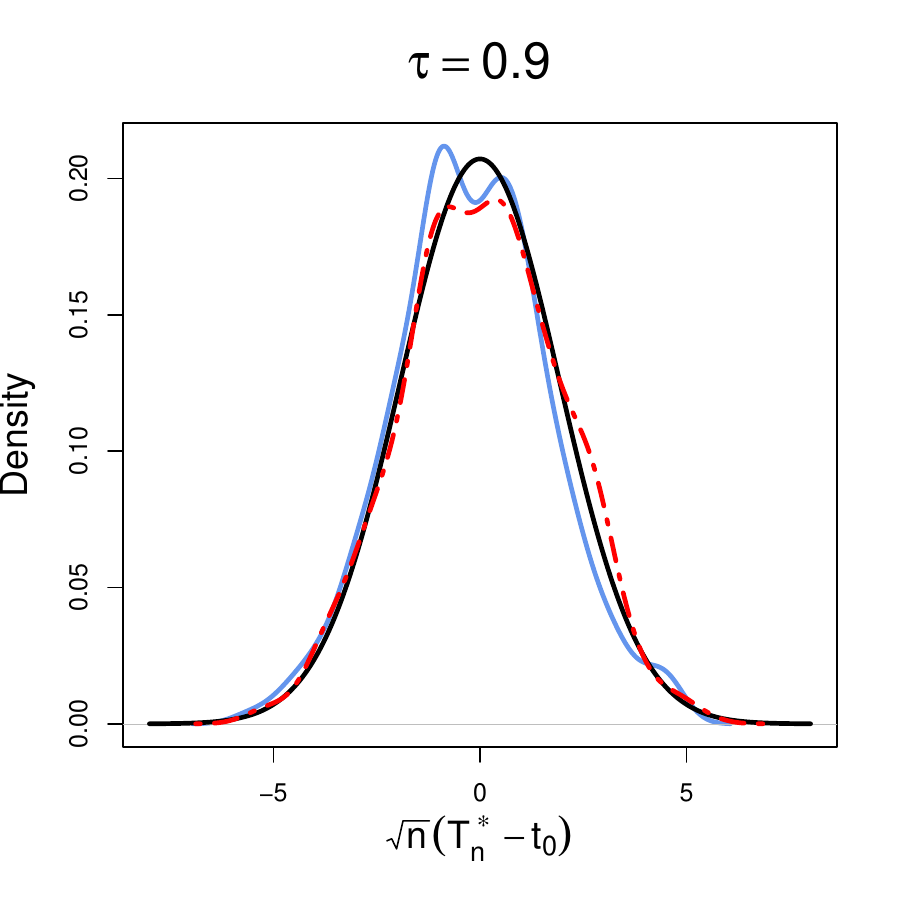}
    \includegraphics[width=0.44\textwidth]{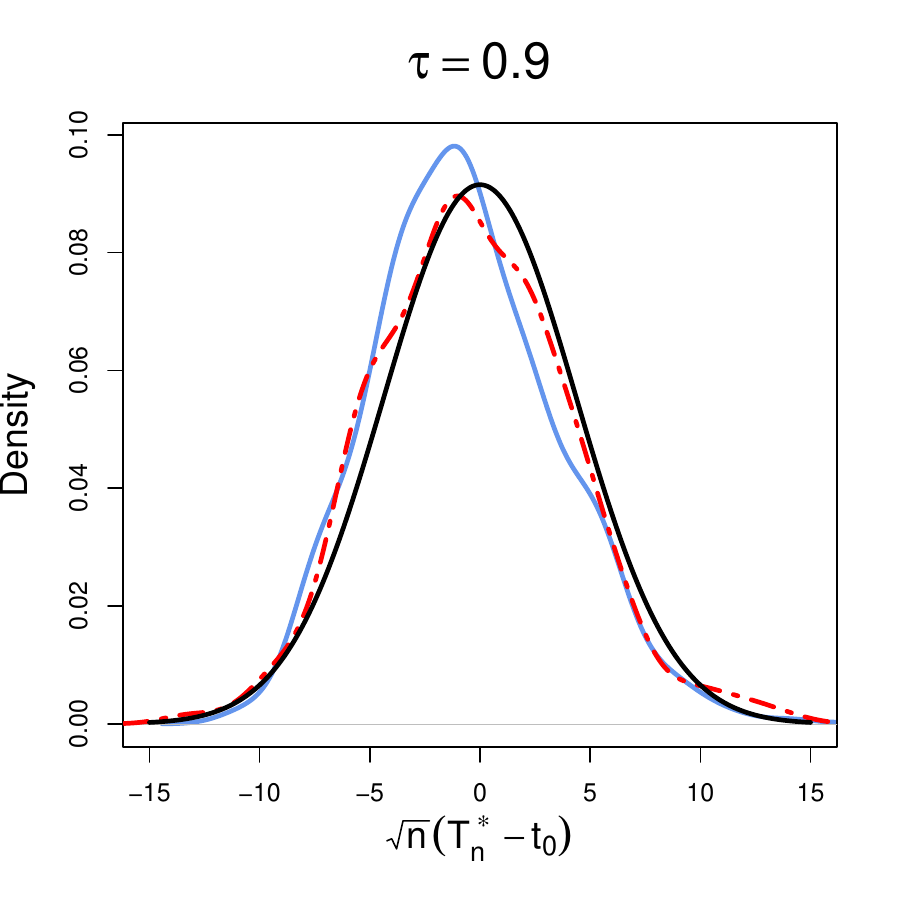}  \\
    \includegraphics[width=0.44\textwidth]{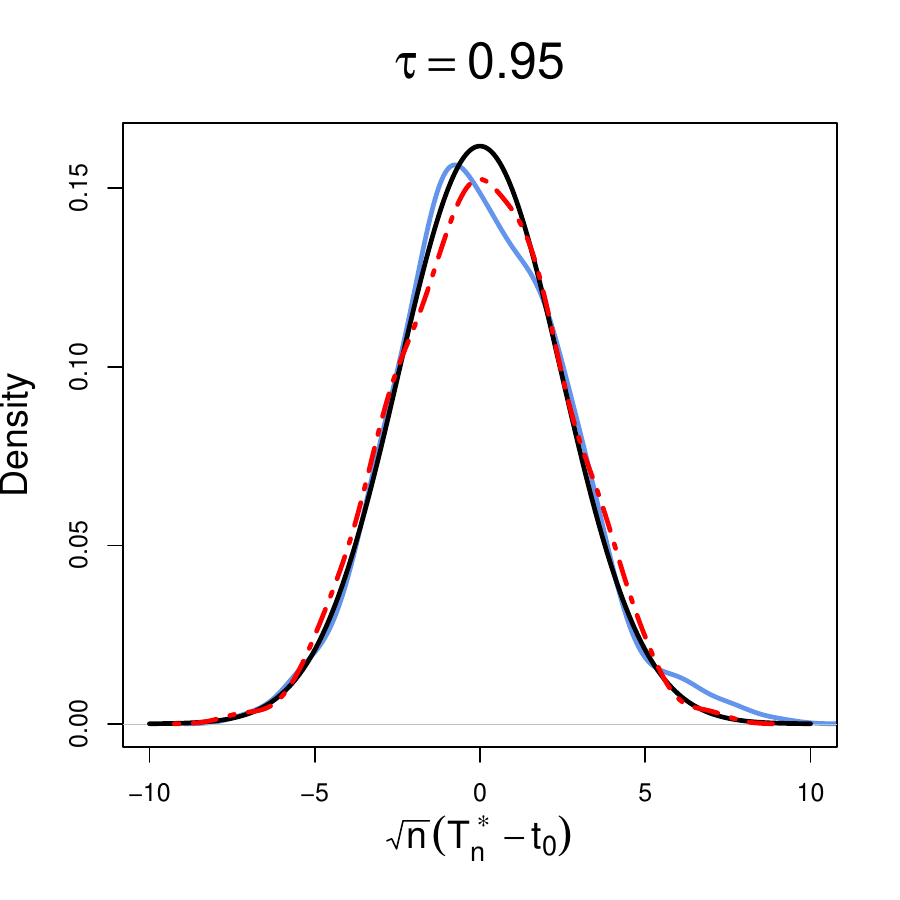}
    \includegraphics[width=0.44\textwidth]{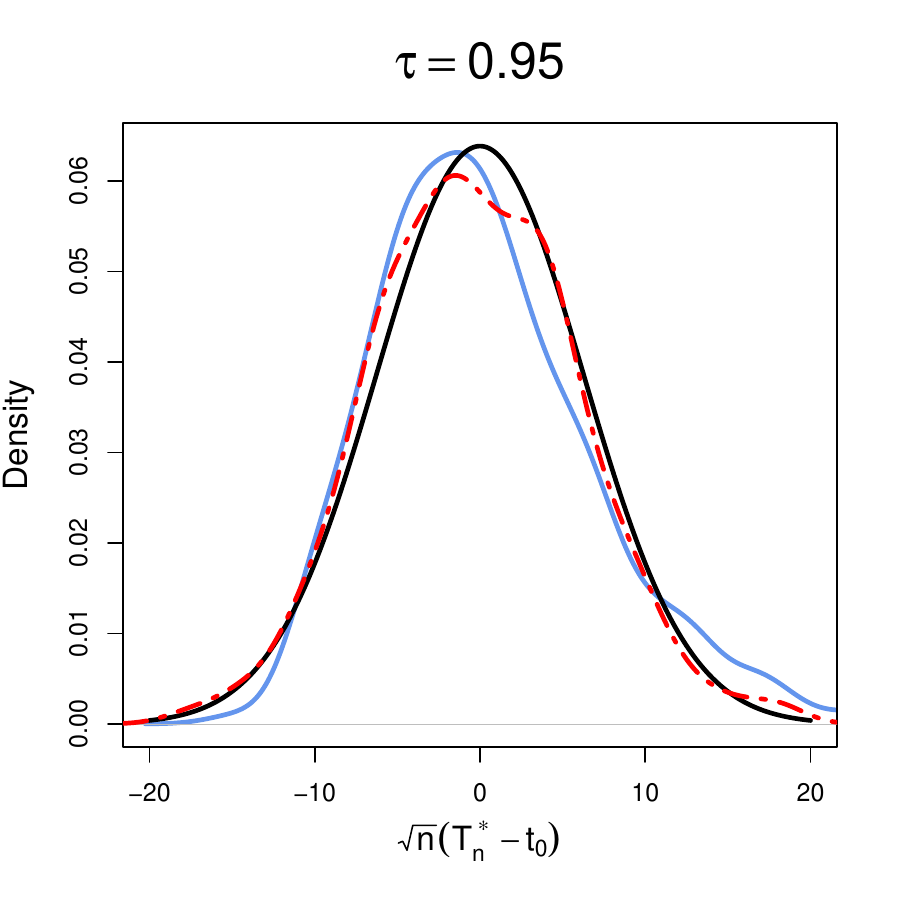}     
    \vspace*{-0.64 cm}
    
    \noindent
    \vspace*{-0.44 cm}
    
    \noindent
    \caption{Simulation study 5. Finite-sample distribution of the estimator $T_n^*$.
    Density estimates of 500 values of $\sqrt{n}(T_n^*-t_0)$ for $n=50$ and $n=800$. Samples drawn from $X\sim \mathcal{N}(0,1)$ (left panels); and from  $X \sim \mbox{Expo}(1)$ (right panels).  In black, the normal density with variance equal to the asymptotic variance of Theorem \ref{theorem: asymptoticNormalityEmpirical}. }
    \label{fig:denistyEstimatesSimulation5}
\end{figure}

Note that even for small sample  size $n=50$, there is a good correspondence between the finite-sample distribution and the asymptotic disribution. Recall from Section S2.2 that the equicontinuity assumption of Theorem \ref{theorem: asymptoticNormalityEmpirical} is satisfied under the sufficient condition that $d_{\tau}$ is Lipschitz continuous. Despite the fact that this not holds for $d_{\tau}$ in this simulation model, the asymptotic normality result still seems to hold. 

\FloatBarrier

\section{General considerations for  risks of distorted random variables}\label{sec:risks of distorted rvs}
In this section we establish some general properties about risk measures $\rho(X_{\sD_{\btau}})$. 
In the real data example in Section \ref{sec: realdata} we refer to these results. 
For notational simplicity we simply write $X_{\sD}$ instead of $X_{\sD_{\btau}}$ in this section. 
Consider a random variable $X$ defined on the probability space $(\Omega, \mathcal{F}, P)$,  with $\mathcal{F}$ the sigma-algebra, and $P$ the probability measure. Denote by $\mathcal{X}$ the linear space of random variables defined on $(\Omega, \mathcal{F}, P)$, and for a given $D \in \mathbb{D}$ consider a risk measure 
\[
\begin{array}{ll}
\rho:  & \mathcal{X} \to \mathbb{R} \\
& X \to \rho(X_{\sD})  .
\end{array} 
\]

Before stating the main result of this section, we recall some basic properties of a risk measure $\rho$. 
\begin{itemize}
\item[$\star$] A risk measure $\rho$ is called \textit{monotone} if: $ \forall X,Y \in \mathcal{X} :  X \leqas Y \Longrightarrow \rho(X)\le \rho(Y)$.
\item[$\star$] A risk measure $\rho$ is called \textit{translation invariant} if:  $\forall X \in \mathcal{X}$, $c \in \mathbb{R} :  \rho(X+c)=\rho(X)+c$.
\item[$\star$] A risk measure $\rho$ is called  \textit{positive homogeneous} if: $\forall X \in \mathcal{X}$, $a \in [0, +\infty)$, such that  $aX\in \mathcal{X}:  \rho(aX)=a\rho(X)$.
\item[$\star$] A risk measure $\rho$ is called  \textit{subadditive} if: $\forall X,Y, X+Y \in \mathcal{X}$ it holds that  $\rho(X+Y)\le \rho(X)+\rho(Y)$.
\end{itemize}
A risk measure $\rho$ is coherent when it is monotone, translation invariant, positive homogeneous and subadditive. 

Theorem \ref{GenPropDistRisk} investigates risk measures for the distorted random variable $X_{\sD}$. In establishing some of these statements, we can rely on results of \cite{LiuSchiedWang2021}. 
The proof of the theorem is provided in Appendix \ref{App: RisksComposition}.\\

\begin{theorem}\label{GenPropDistRisk}
Denote by $D$ a distortion function and by $\rho$ a risk measure defined on $\mathcal{X}$, the  linear vector space of random variables.
For two random variables $X$ and $Y$ in the domain of $\rho$ we then have:
\begin{itemize}
	\item[(i)] If $\rho$ is positive homogeneous, then also $\rho((aX)_{\sD}) = a\rho(X_{\sD})$ for $a\in[0,\infty)$.
	\item[(ii)] If $\rho$ is translation invariant, then also $\rho((X+c)_{\sD}) = \rho(X_{\sD})+c$ for $c\in\mathbb{R}$.
	\item[(iii)] If $\rho$ is monotone and $X \leqas Y$ then $\rho(X_{\sD}) \leq \rho(Y_{\sD})$. 
	\item[(iv)] Suppose $\rho_{\delta}$ is coherent. If $D$ is convex, then 
     the composition (distortion) risk measure is also coherent (and thus in particular subadditive), i.e. for all  $X$ and $Y$, we have $\rho((X+Y)_{\sD}) \leq \rho(X_{\sD}) + \rho(Y_{\sD})$.
	\item[(v)] Denote by $D_1$ and $D_2$ two distortion functions such that $D_1(u) \leq D_2(u)$ for all $u\in[0,1]$. If $\rho$ is monotone, then $\rho(X_{\sD_2})\leq \rho(X_{\sD_1})$. Specifically for $D_1(u) \leq u \leq D_2(u)$ we have $\rho(X_{\sD_2}) \leq \rho(X)\leq \rho (X_{\sD_1})$.
	\item[(vi)] Suppose that the risk measure $\rho$ depends on a parameter $\gamma$, and denote the risk measure by $\rho_{\gamma}$. If $\rho_{\gamma}$ is monotone increasing in $\gamma$, then so is the composition risk measure $\rho_{\gamma}(X_{\sD})$. 
\end{itemize}
\end{theorem}

Risk measures defined as the minimizer of an expected loss minimization procedure as in \eqref{eq:argmin:Exp} avoid unjustified risk-loading under general conditions.
If $\ell$ is such that $\arg\min_{c \in \mathbb{R}} \ell(x,c)=x$ for all $x\in\mathbb{R}$ then we have for an almost surely constant random variable $X\sim\ind{[b,\infty)}$, i.e., $X\eqas b$ for some $b\in\mathbb{R}$, that $\rho(X) = \arg\min_{c\in\mathbb{R}}\Exp[\ell(X,c)] = \arg\min_{c\in\mathbb{R}} \ell(b,c) = b$.
This property also carries over to the distorted version $X_{\sD}$, where we have $\sD\circ\ind{[b,\infty)} = \ind{[b,\infty)}$ due to $D(0)=0$ and $D(1)=1$.
As such we have $X_{\sD}\eqas b$ and hence also $\rho(X_{\sD}) = b$.

\section{Real data example: US disaster data 1980-2023}\label{sec: realdata}

As an illustration we include some analysis of  a dataset from \cite{disasterData2023}. The data concern measurements on  weather and climate disasters in the U.S. since 1980, restricted to disasters that led to an overall damage/costs of more than or equal to \$1 billion (Consumer Price Index (CPI) adjusted to 2023). There are 360 such disasters amounting to a total cost of \$2575,7 billion. Each disaster event has been classified into one of seven classes. Some summary statistics regarding the  different (disaster) events are in Table \ref{tab:disasterSummary}.
Our analysis mainly focusses on the  three disaster types with the largest number of observations: severe storms, tropical cyclones and flooding. Per disaster type, we treat the entries as outcomes of i.i.d.  nonnegative random variables $X_1,\cdots, X_{n}$. The left panels of Figure \ref{fig:histogramAndQuantilesESV2} display histograms and scatter plots of the data.

\begin{table}[h]
\centering
\caption{US disaster data. Summary statistics. Columns indicated with $^*$: numbers are in billions of US dollars.}
\label{tab:disasterSummary}
\vspace*{0.2 cm}

\noindent
\resizebox{\textwidth}{!}{%
\begin{tabular}{l|llllll}
\hline \hline
Disaster Type    & Events & Events/Year & Total Costs$^*$ & Cost/Event$^*$& Cost/Year$^*$ & Variance/Event$^*$ \\ \hline
Drought          & 30     & 0.7         & \$334,8    & \$11,2    & \$7,6    & 12,2         \\
Flooding         & 41     & 0.9         & \$190,2    & \$4,6     & \$4,3    & 7,2           \\
Freeze           & 9      & 0.2         & \$36,0     & \$4,0     & \$0,8    & 2,2           \\
Severe Storm     & 177    & 4.0         & \$423,1    & \$2,4     & \$9,6    & 1,9           \\
Tropical Cyclone & 60     & 1.4         & \$1359,0   & \$22,7    & \$30,9   & 38,5          \\
Wildfire         & 21     & 0.5         & \$135,5    & \$6,5     & \$3,1    & 7,7          \\
Winter Storm     & 22     & 0.5         & \$97,1     & \$4,4     & \$2,2    & 5,6           \\ \hline
All Disasters    & 360    & 8.2         & \$2575,7   & \$7,2     & \$58,5   & 17,9      \\ \hline \hline
\end{tabular}%
}
\end{table}
\begin{figure}[h!]
    \centering
    \includegraphics[width=0.40\textwidth]{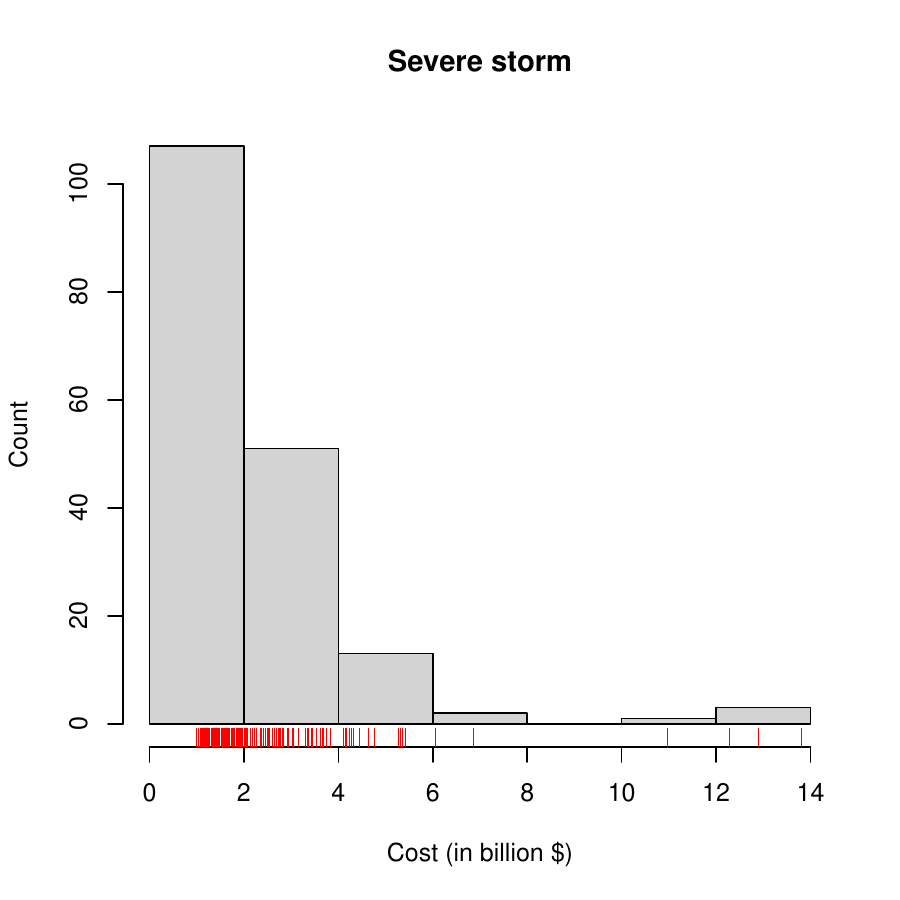}
\includegraphics[width=0.40\textwidth]{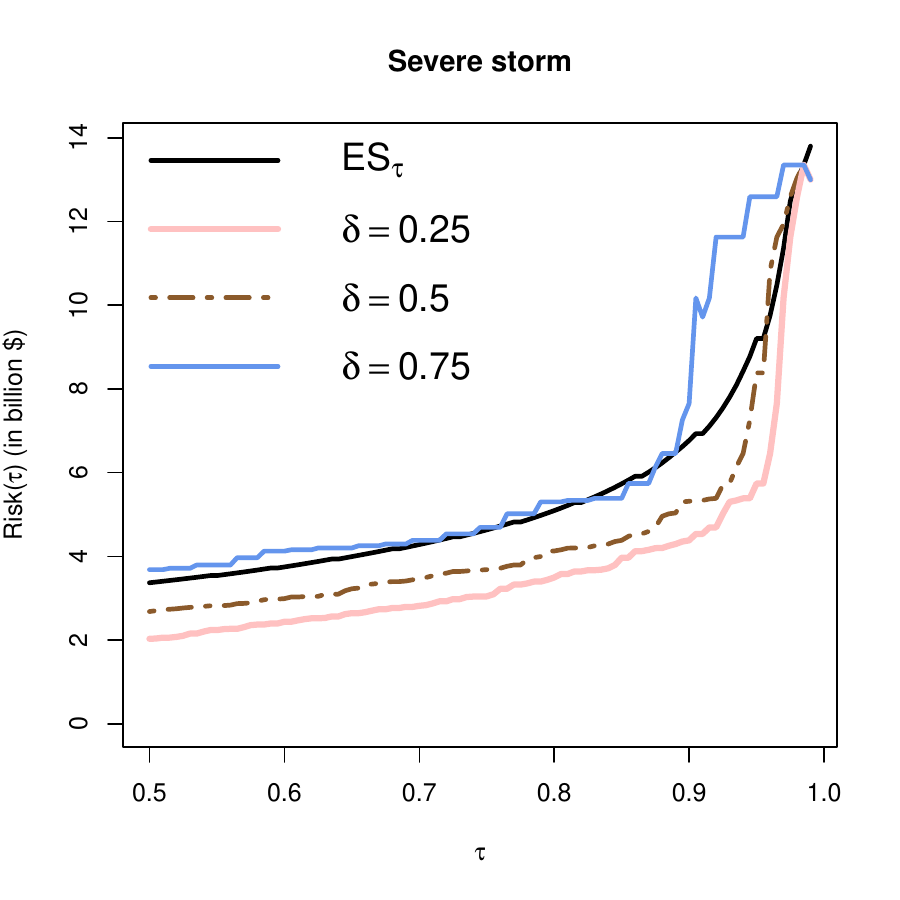} \\
 \includegraphics[width=0.40\textwidth]{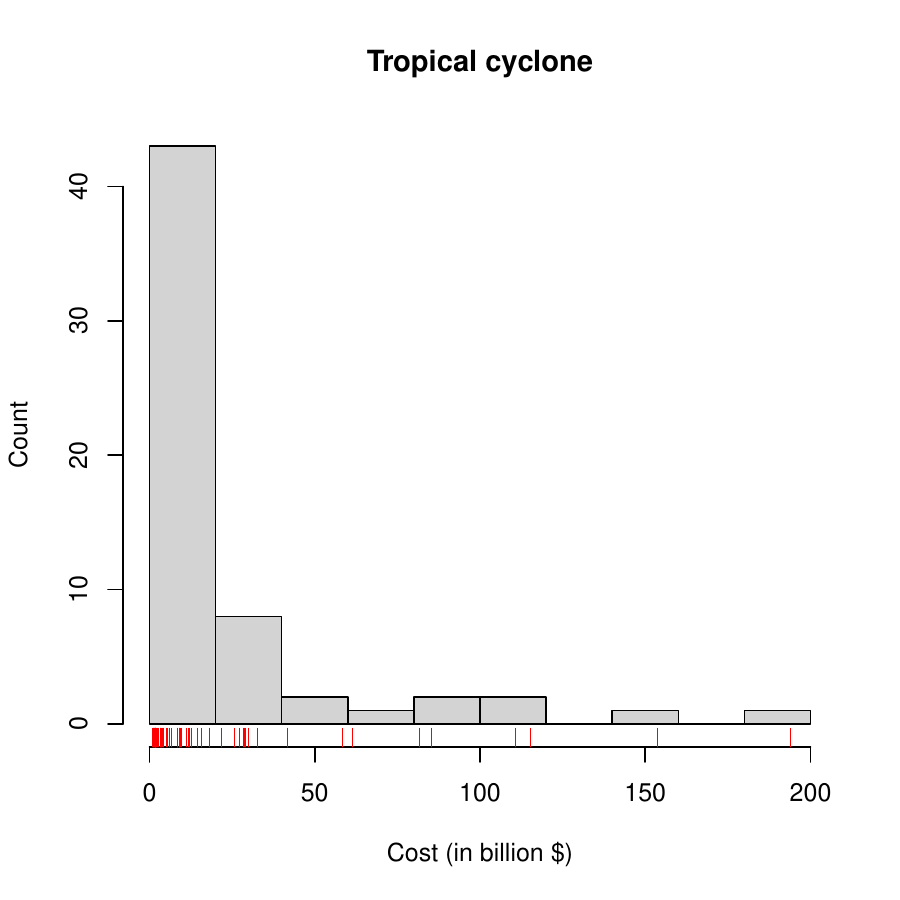} 
\includegraphics[width=0.40\textwidth]{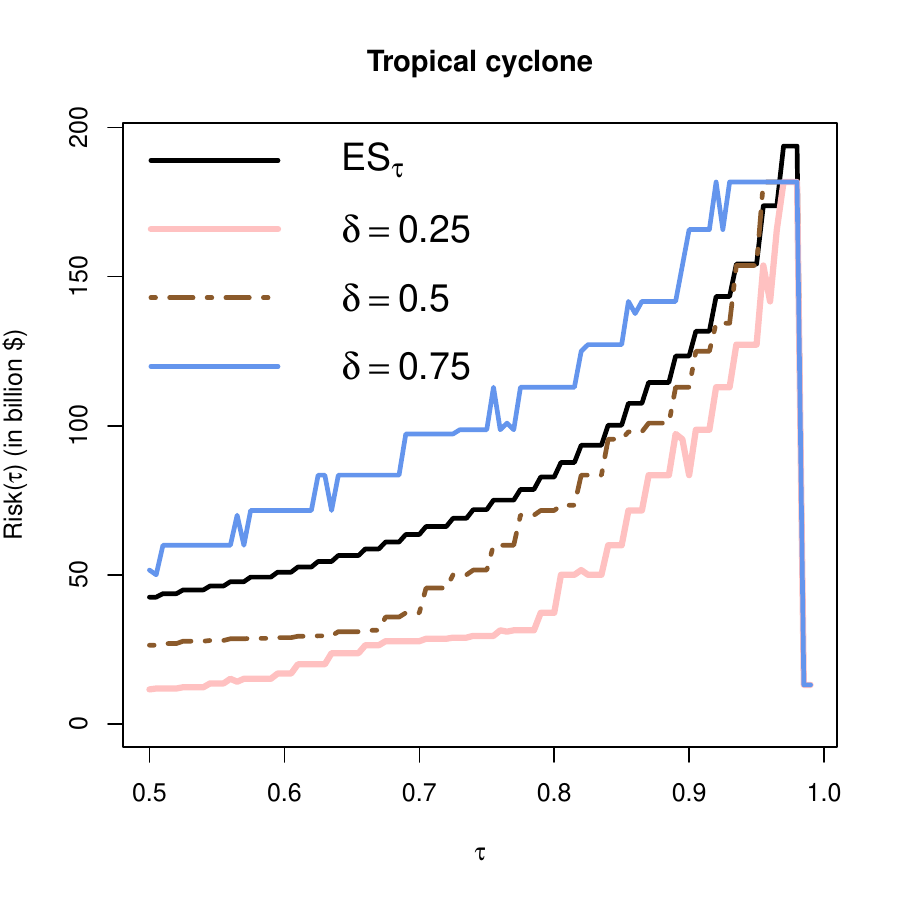} \\  
 \includegraphics[width=0.40\textwidth]{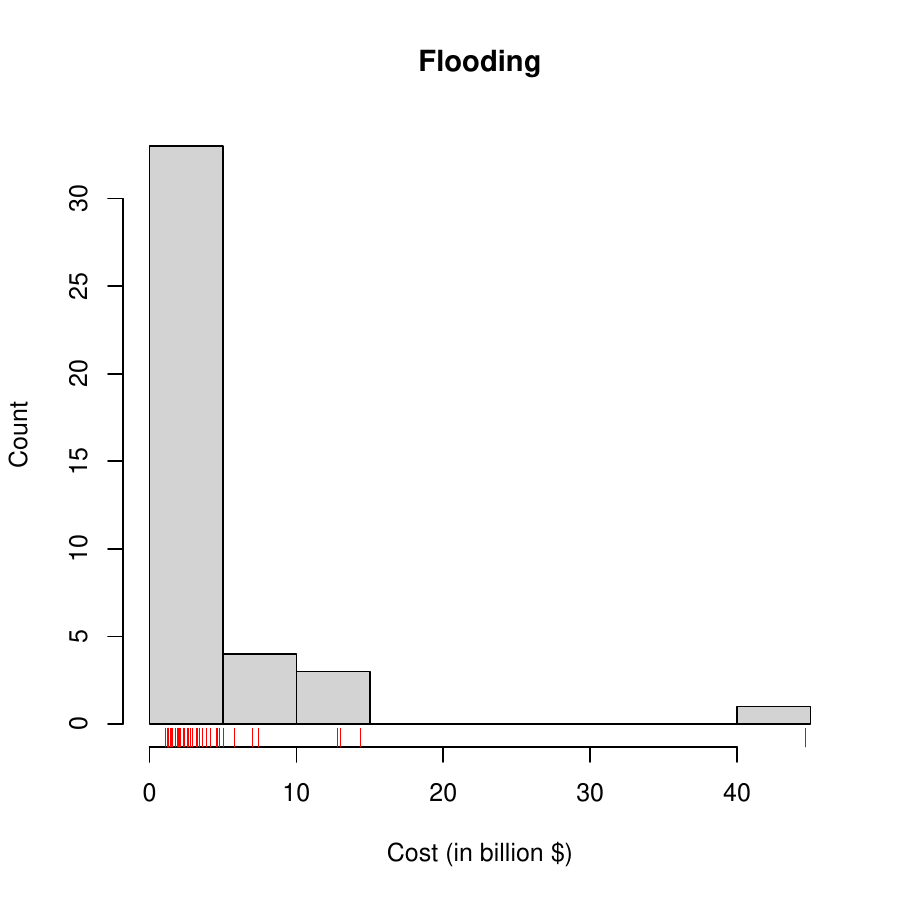} 
   \includegraphics[width=0.40\textwidth]{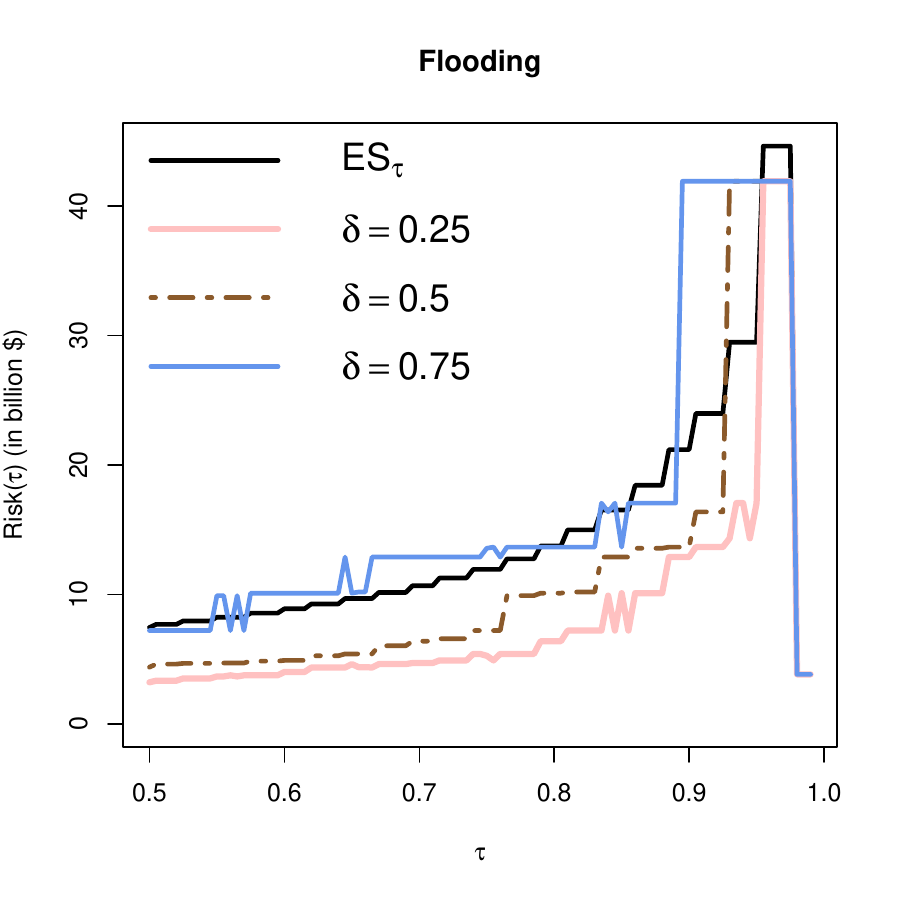}  
      \vspace*{-0.52 cm}
    
    \noindent
    \caption{US disaster data.  Left: histograms and scatter plots of different disaster types. Right: estimates of the mean and various ($\delta=0.25,0.5,0.75$) quantiles of $X_{\sD_{\btau}}$ with $D_{\btau}$ corresponding to expected shortfall.}
    \label{fig:histogramAndQuantilesESV2}
\end{figure}

\begin{figure}[h]
\centering 
         \includegraphics[width=0.48\textwidth]{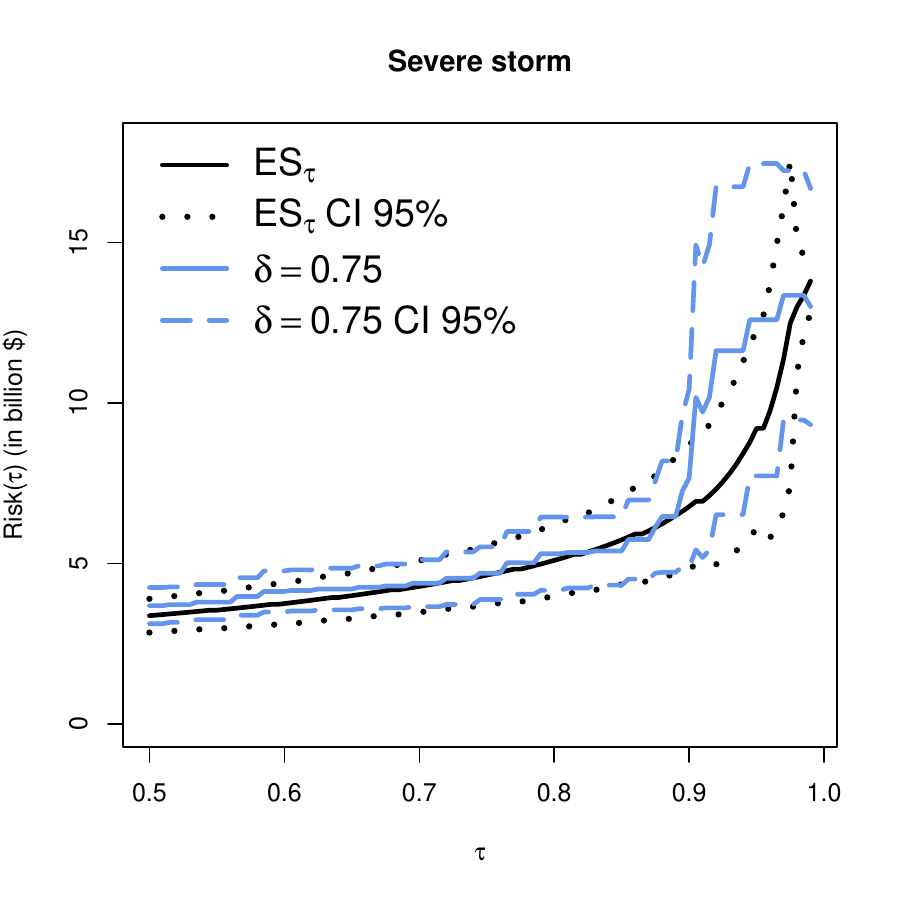}
       \includegraphics[width=0.48\textwidth]
{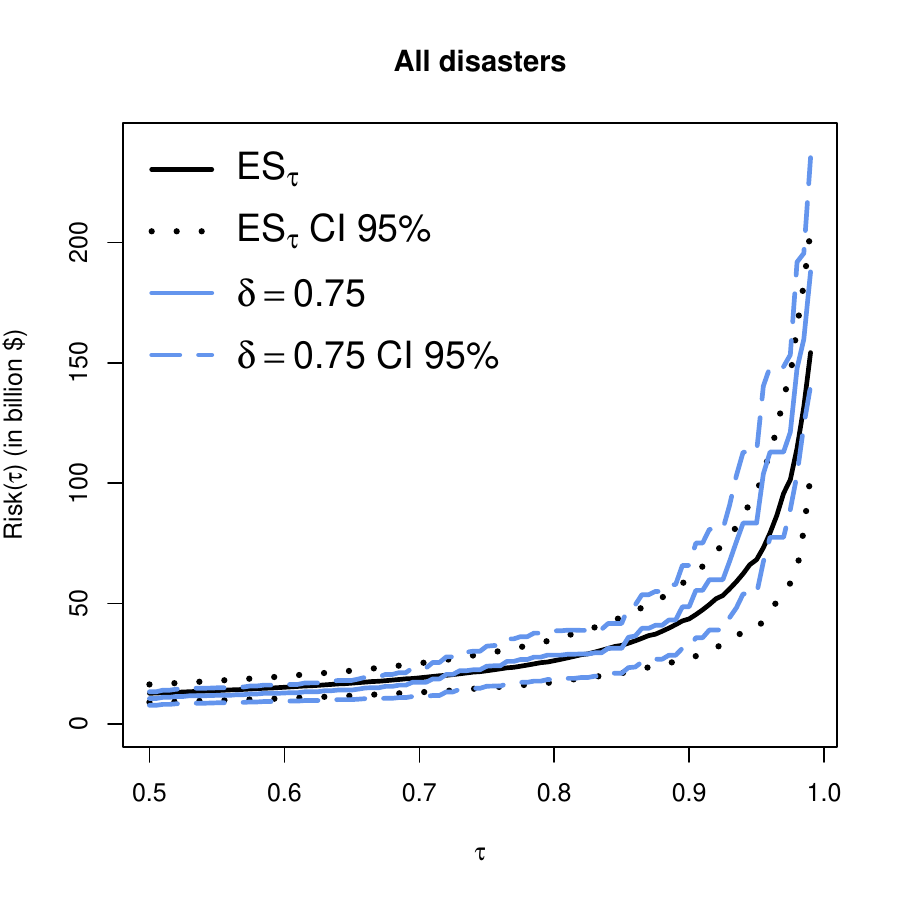}
    \vspace*{-0.68 cm}
    
    \noindent
    \caption{US disaster data. Estimated $\mbox{ES}_{\btau}$ and $0.75$th quantile curves, together with 95\% asymptotic pointwise confidence intervals (in respectively dotted and dashed lines). Left panel: for  severe storms. Right panel: for  all disasters.}
    \label{fig:confidenceIntervalsSevereStorms}
\end{figure}
\begin{figure}[h!]
    \centering
\includegraphics[width=0.48\textwidth]{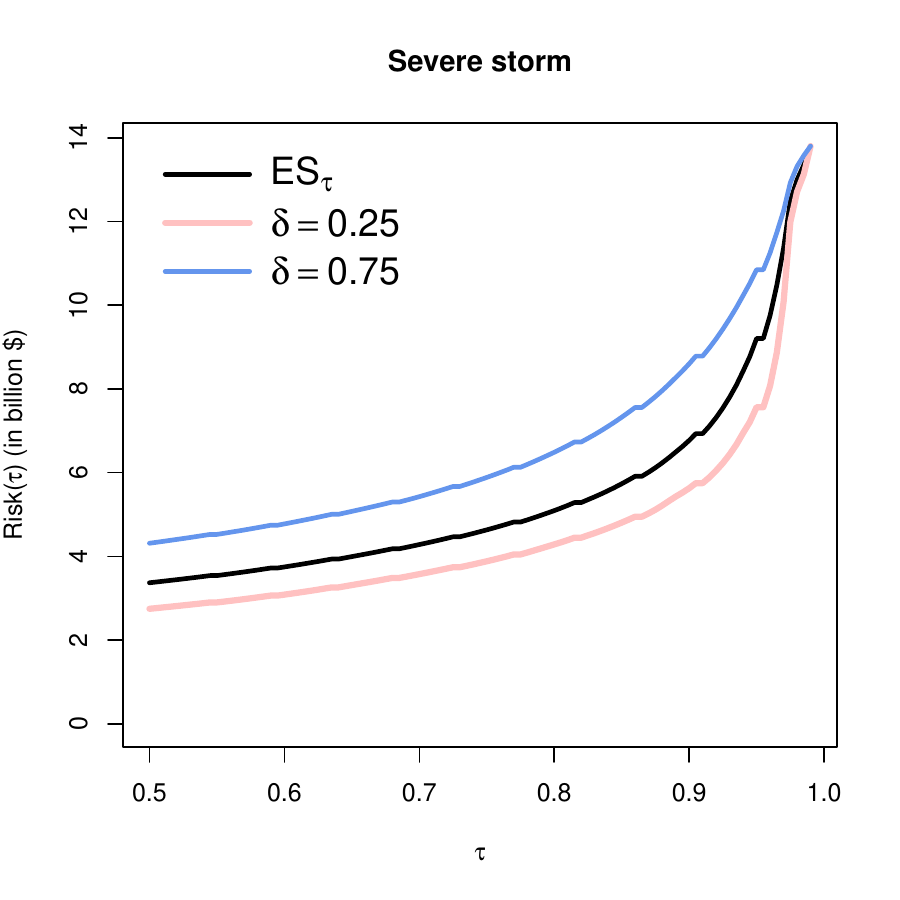}
       \includegraphics[width=0.48\textwidth]
{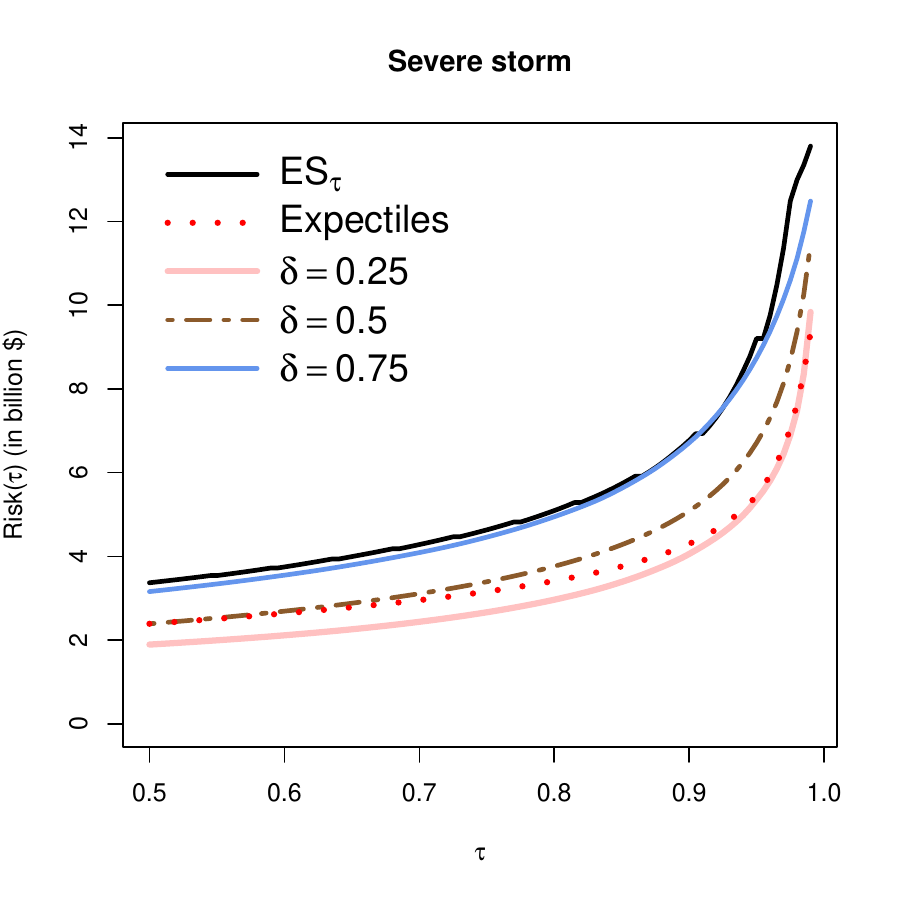}
      \vspace*{-0.68 cm}
    
    \noindent
    \caption{US disaster data. Severe storms. Left: estimates of various ($\delta=0.25,0.75$) expectiles of $X_{\sD_{\btau}}$ with $D_{\btau}$ corresponding to expected shortfall. Right: estimates of various ($\delta=0.25,0.5,0.75$) expectiles of $X_{\sD_{\btau}}$ with $D_{\btau}=K_{\btau}$, i.e. corresponding to extremiles. In dotted line, estimates of the usual expectiles. }
    \label{fig:ExpectilesESandExpectilesExtremilesV2}
\end{figure}

We firstly consider the distorted random variable $X_{\sD_{\btau}}$ using $D_{\tau}(u)= (1-\tau)^{-1} (u-\tau) 1_{\{ u \ge \tau \}}$, the Expected Shortfall distortion function (see entry 5 in Table \ref{DistrFunctions}). 
We evaluate the risks $\rho_{\delta}$ attributed to $X_{\sD_{\btau}}$, using (i) the mean risk (denoted by $\mbox{ES}_{\tau}$ in the figures)--using square loss; and (ii) the risk associated to the $\delta$th quantile --using quantile loss. 
The right panels of Figure \ref{fig:histogramAndQuantilesESV2} display estimates of these risks, as a function of $\tau$. Note that the estimates break down, i.e., become zero, at $\tau=n/(n+1)$. This is to be expected from the fact that $D_{\tau}$ is the cumulative distribution function corresponding to $\mbox{U}[\tau,1]$, and the lack of ranked data beyond $n (n+1)^{-1}$. For the three different disaster types the estimates of $\text{Median}\left ( X_{\sD_{\btau}} \right )$ lie below the estimates of $\text{Mean}\left ( X_{\sD_{\btau}} \right )$ (denoted $\mbox{ES}_{\btau}$ on the plots)  when $\btau\leq 0.9$. This means that the distribution of  $X_{\sD_{\btau}} $ is right-skewed. For the cases of floodings and severe storms, the estimates of the $\delta=0.75$ quantiles are close to the estimates of $\text{ES}_{\btau}$, 
although floodings are roughly two and a half times as risky/costly.  

To illustrate the uncertainty that comes with the estimates, we calculated also approximate 95\%  pointwise confidence intervals for the severe storm category,  relying on the asymptotic results established in Sections  \ref{sec:estimation--GeneralCase} and \ref{sec:estimation--ApplicationsANResult}. 
In Figure \ref{fig:confidenceIntervalsSevereStorms} (left panel) we depict these approximate pointwise 95\%  confidence intervals, for both  $\mbox{ES}_{\tau}$  and the   0.75th quantile.
The construction of these approximate confidence intervals relies
on estimators for the asymptotic variance. For the estimate of  $\mbox{ES}_{\tau}$, we
obtain  the approximate variance by substituting the empirical cumulative distribution
function into expression \eqref{eq:asymptoticVarianceSquaredLoss}. For the asymptotic variance
of the estimated $\delta$th quantile, we use expression \eqref{eq:quantileLossEstimatorAsymptoticVariance}. In our implementation we apply the kernel density estimator from the R package \texttt{stats} in
combination with the Sheather \& Jones bandwidth selection method. More specifically, we use the command \texttt{density} with bandwidth selector 
\texttt{bw.SJ}.
To prevent division by zero in expression \eqref{eq:quantileLossEstimatorAsymptoticVariance}, we need to ensure that $\widehat{f}_X(T_n^*)\not =  0$. To
achieve this, we increase the Sheather \& Jones bandwidth at each point where
we estimate the density, ensuring that the neighbourhood determined by the
bandwidth contains at least a fixed proportion of observations. We set the
threshold at 10\% of the total number of observations. For the severe storm category, this
corresponds to a minimum of 18 observations. 
 Similar plots for the other two values of $\delta$ (quantile curves) are to be found in Figure S3 in Section S5 of the Supplementary Material. 
 In the right panel of Figure \ref{fig:confidenceIntervalsSevereStorms} we display the estimated $\mbox{ES}_{\tau}$ and 0.75th quantile curves together with the approximate pointwise 95\% confidence intervals, but now for all disasters together. A histogram and scatter plot for data on all disasters (all seven categories) is depicted in the left panel of Figure \ref{fig:allDisaster}. The discrepancy between the estimated $\mbox{ES}_{\btau}$ curve and the quantile curve, in Figure \ref{fig:confidenceIntervalsSevereStorms},  is less pronounced when considering all disasters together.

It is also interesting to evaluate the estimated risk curves keeping in mind the theoretical results established in Theorem \ref{GenPropDistRisk}. Note that 
since $D_{\btau}$ is decreasing in $\btau$, the risk $\rho_{\delta}$ is expected to be an increasing function in $\tau$, according to Theorem \ref{GenPropDistRisk}(v).
As can be seen from Figure \ref{fig:histogramAndQuantilesESV2} (right panels) there is no guarantee that the estimated curves are increasing (although the overall increasing trend is clear).   
Furthermore, since the quantile risk is monotone increasing in $\delta$, Theorem \ref{GenPropDistRisk}(vi) ensures that this is also the case for the risk applied to the distorted random variable $X_{\sD_{\btau}}$.  From visual inspection of the estimated curves in Figure \ref{fig:histogramAndQuantilesESV2}  it seems that there is no issue of crossing quantile curves in this example (at least not before the occurence of the breakdown point).  
Developing estimation procedures that would guarantee that estimators  also present the associated theoretical properties is an open research question.  

Secondly, we estimate expectile risk measures associated to $X_{\sD_{\btau}}$, meaning that we take the expectile loss function. 
In the left panel of Figure \ref{fig:ExpectilesESandExpectilesExtremilesV2}, we display estimates of various expectiles ($\delta=0.25,0.75$) of $X_{\sD_{\btau}}$. In this figure and in Figure \ref{fig:allDisaster}, we display the black curve corresponding to the estimate $\mbox{ES}_{\btau}$ as a reference curve. This in fact corresponds to the estimated risk curve for the $\delta=0.5$th expectile risk. Within each disaster type the curves exhibit a high degree of similarity. 
The parameter $\delta$ allows us to obtain more (or less) conservative risk estimates. As in our previous risk analysis, there is a similar behaviour for the disasters of floodings and severe storms. We only present the plots for the disaster severe storms here. Similar plots for tropical cyclone and flooding disaster categories are displayed in Section S5 of the Supplementary Material. 
By Theorem \ref{GenPropDistRisk}(v), we know that the population risk $\rho_{\delta}$ is an increasing function in  $\tau$, for each expectile level $\delta$. The estimated risk curves also clearly exhibit this increasing behaviour. Theorem \ref{GenPropDistRisk}(vi) also 
states that there should be no crossings of the population risk curves  
for different $\delta$-values. Also the estimated risk curves do not show crossings (before the  breakdown point).
\begin{figure}[h!]
    \centering
\includegraphics[width=0.48\textwidth]{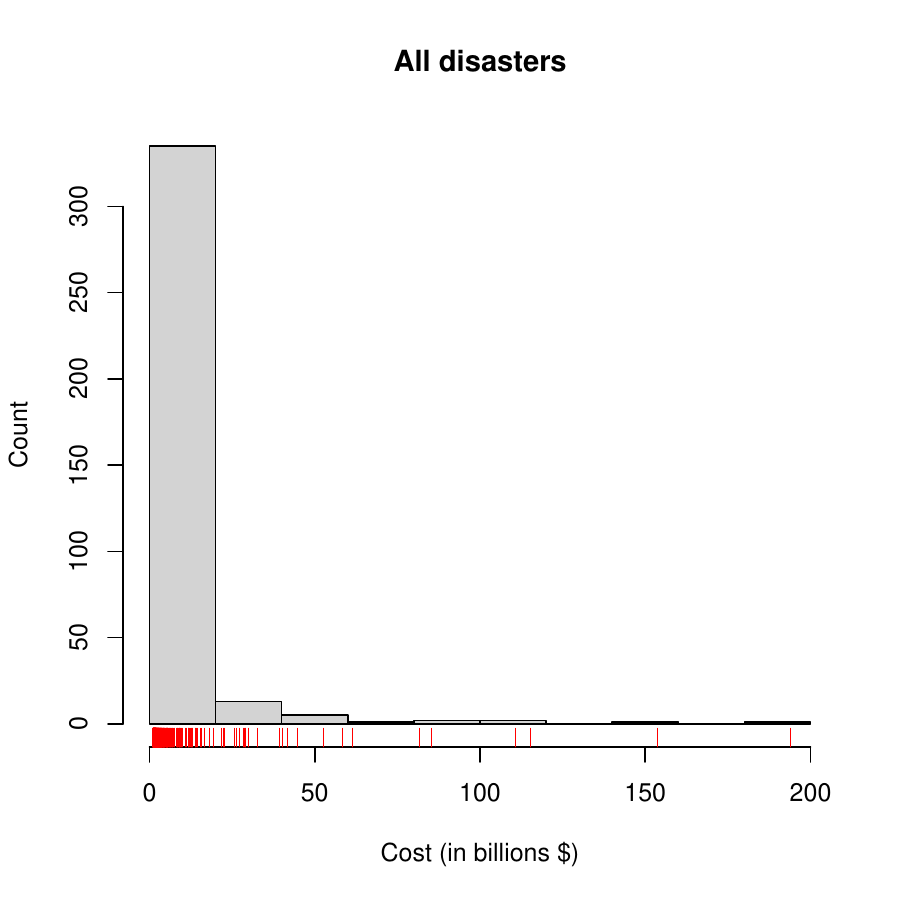} \includegraphics[width=0.48\textwidth]{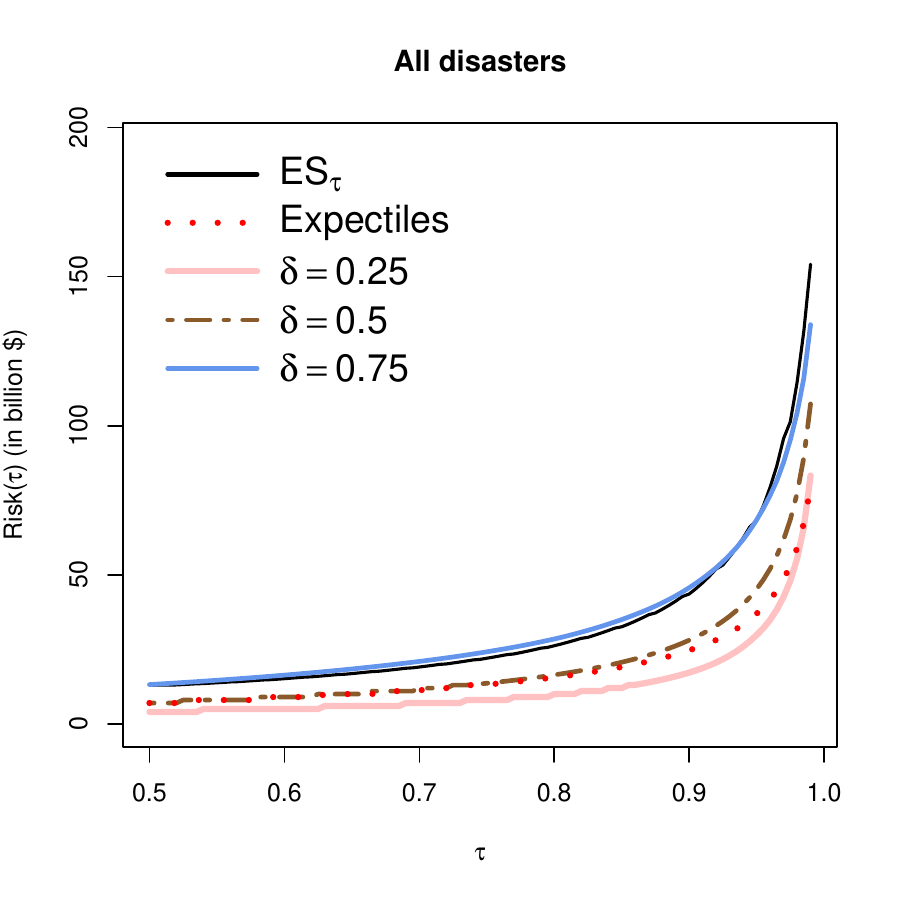}
    \vspace*{-0.44 cm}
    
    \noindent
    \caption{US disaster data. Top left: histogram and scatterplot of all disasters. Top right: estimates of various ($\delta=0.25,0.5,0.75$) expectiles of $X_{\sD_{\btau}}$ with $D_{\btau}=K_{\btau}$, i.e. corresponding to extremiles. In  dotted line, estimates of the usual expectiles.}
    \label{fig:allDisaster}
\end{figure}
\medskip 

Thirdly, we look into extremiles as risk measures. Therefore we take $D_{\btau}=K_{\btau}$, and consider expectile risks of the distorted random variable 
$X_{\sD_{\btau}}$. 
The right panel of Figure  \ref{fig:ExpectilesESandExpectilesExtremilesV2} displays estimates of different expectile risk curves ($\delta=0.25,0.5,0.75$) of $X_{\sD_{\btau}}$ with $D_{\btau}=K_{\btau}$, which gives estimated extremile risk curves for $X$.  We only present the estimates for the disaster category severe storms. 
In black, we also display estimates of the usual expectile curve (as function of $\tau$). In contrast to all previous estimated curves in this example, we now observe smooth estimated curves. In addition, there is no breakdown point. We note that the $\delta=0.5$th expectile of $X_{\sD_{\btau}}$ coincides with the usual extremile. Note again that the estimated curves exhibit a high shape similarity. 
For low $\tau$-values the estimates of the usual expectiles (in black) and the usual extremiles almost coincide. For larger $\btau$-values, expectiles are more conservative than extremiles. The cumulative distribution function   $D_{\tau}=K_{\tau}$ is decreasing in $\tau$. Hence, again by 
Theorem \ref{GenPropDistRisk}(v), we know that the population risk $\rho_{\delta}$ is an increasing function in  $\tau$. The estimated risk curves also show this monotone behaviour. Also the non-crossingness property stated in Theorem \ref{GenPropDistRisk}(vi) seems to hold for the estimated curves. 
Figure S4 in the Supplementary Material presents the estimated expectile and extremile curves for floodings, tropical storms as well as for all disasters together.
\begin{table}[htb]
\centering
\caption{Estimated risks for severe storms, floodings and tropical storms, for various risk measures.}
\label{tab:USDisasterComparison}
\vspace*{0.2 cm}

\noindent
\hspace*{-0.6 cm}
\begin{tabular}{|l|rrrr|rrrr|}
\hline
   & \multicolumn{4}{c}{$\tau=0.85$} & \multicolumn{4}{c|}{$\tau=0.95$}\\[1.2 ex]
\multicolumn{1}{|c|}{pair $(D_{\btau}, \ell_{\delta})$}&   Severe & Tropical  & Flooding & All &  Severe & Tropical  & Flooding & All\\ 
   & storm & cyclone &  & disasters &   storm & cyclone & & disasters\\
 \hline
$D_\tau \sim U[\tau,1]$  &              &                  &          &             &&&& \\ \cline{1-1}
$\quad$ Squared loss   &5.73         & 100.20           & 16.53    & 32.61                                & 9.21         & 154.23           & 29.49    & 68.23         \\[1.2 ex]
$\quad$ Quantile loss $\delta=0.25$   & 3.97         & 59.91            & 10.10    & 11.88                  & 5.74         & 127.20           & 17.06    & 35.99         \\
$\quad$ Quantile loss $\delta=0.5$    &  4.38         & 95.48            & 12.89    & 17.30                   & 8.38         & 153.75           & 41.92    & 48.63         \\
$\quad$ Quantile loss $\delta=0.75$    &5.39         & 127.20           & 13.66    & 31.39                 & 12.59        & 181.75           & 41.92    & 83.43         \\[1.2 ex]
$\quad$ Expectile loss $\delta=0.25$    &4.81         & 81.65            & 12.80    & 22.12               & 7.56         & 138.60           & 21.92    & 51.64         \\
$\quad$ Expectile loss $\delta=0.75$    &7.32         & 122.84           & 23.56    & 49.81           & 10.85        & 170.06           & 37.07    & 91.34         \\[1.4 ex] \cline{1-1}
$D_\tau=K_\tau$          &              &                  &          &              &&&& \\ \cline{1-1}
$\quad$ Expectile loss $\delta=0.25$     &3.39         & 41.36            & 7.39     & 12.40          & 5.36         & 89.60            & 14.82    & 29.83         \\
$\quad$ Expectile loss $\delta=0.5$      & 4.28         & 63.08            & 10.86    & 20.72               & 6.71         & 114.35           & 20.41    & 45.24         \\
$\quad$ Expectile loss $\delta=0.75$      &5.63         & 90.73            & 16.43    & 35.04              & 8.74         & 140.32           & 29.13    & 68.45         \\[1.4 ex] \cline{1-1}
$D_\tau\sim U[0,1]$                   &&&&    &              &                  &          &               \\[1.4 ex] \cline{1-1}
$\quad$ Expectile loss $\delta=0.95$     &5.46         & 88.48            & 16.84    & 37.02        & 5.46         & 88.48            & 16.84    & 37.02   \\
\hline     
\end{tabular}
\end{table}

\medskip 

In a final analysis, we again combine the data of all seven classes of disasters. 
In Figure \ref{fig:allDisaster}, and Figure S4 (bottom panel) in the Supplementary Material, we display similar estimated risk curves, now based on all 360 observations. Table \ref{tab:disasterSummary} shows that over half of the total cost is due to tropical cyclones. Nevertheless, the shape of the estimated risk curves of all disasters together appear to be somewhat similar to the estimated risk curves for  floodings and severe storms. 
To compare various risks estimates, we present in Table \ref{tab:USDisasterComparison} the values of the estimated risks for each of the separate disaster categories, as well as the estimated risks for all disasters together. We quantify the various risks for two specific values of $\tau$, namely $\tau=0.85$, and $\tau=0.95$. Note that, for all risk estimates (i.e., looking row by row), the risks are systematically highest for tropical cyclones, and lowest for severe storms. The estimated risks for flooding are about a factor 3 larger than these for severe storms, and a  factor 6 to 7 smaller than these for tropical cyclones.
Our framework not only allows to estimate the various risks, it also provides approximate confidence intervals for the risks. In Table \ref{tab: TableRisksAndCI} we provide this information for some of the risk measures for severe storms.  
As expected the approximate CIs are wider for larger values of $\tau$. For a given $\tau$, the CIs are wider for higher-level risk measures (i.e. higher values of $\delta$). 
\begin{table}[h!]
\centering
 \caption{Some estimated risks for severe storms, with the approximate 95\% confidence interval (CI) and its length.}
\label{tab: TableRisksAndCI}
\vspace*{0.2 cm}

\noindent
\hspace*{-0.6 cm}
\begin{tabular}{|l|ccc|ccc|}
\hline
  & \multicolumn{3}{c}{$\tau=0.85$} & \multicolumn{3}{c|}{$\tau=0.95$} \\[1.2 ex]
 \multicolumn{1}{|l|}{pair $(D_{\btau}, \ell_{\delta})$}  &   risk  &  &  &    risk  &  &  \\
$\quad D_\tau \sim U[\tau,1]$ & estimate & CI & length CI & estimate & CI & length CI  \\
 \hline
$\quad$ Squared loss   &5.73 & $[4.35, 7.12]$ & 2.77 & 9.21 & $[6.04, 12.37]$ & 6.33          \\[1.2 ex]
$\quad$ Quantile loss $\delta=0.25$   & 3.97 & $[3.39, 4.55]$ & 1.16 & 5.74 & $[4.51, 6.97]$ & 2.46 \\
$\quad$ Quantile loss $\delta=0.5$    &   4.38 & $[3.65, 5.11]$ & 1.46 & 8.38 & $[4.83, 11.94]$ & 7.11\\
$\quad$ Quantile loss $\delta=0.75$    &5.39  & $[4.32, 6.45]$ & 2.22 & 12.59 & $[7.73, 17.46]$ & 9.73       
 \\[1.2 ex]
\hline 
\end{tabular}
 \end{table}
               

\section{Concluding remarks and further discussion} \label{ConclusionDiscussion}

This paper studies general functionals, called generalized extremiles, that are solutions to an optimization problem involving as crucial ingredient a pair $(D_{\btau}, \ell_{\delta})$ with $D_{\btau}$ a cumulative distribution function with support $[0,1]$, and with $\ell_{\delta}$ a loss function. The main contribution in the paper consists of providing estimators for the generalized extremile, and proving their consistency and asymptotic normality. 
By allowing for a general pair $(D_{\btau}, \ell_{\delta})$ several functionals studied in the literature fall under this framework, and hence we also provide statistical inference for all these. Cases for which estimators and asymptotic normality results for them are available in the literature are limited, and for these cases we showed that our general theorem includes these results as special cases. 

Several interesting research questions can be raised. A first question is that one can wonder under which conditions the pair $(D_{\btau}, \ell_{\delta})$ would lead to a coherent risk measure itself. Some partial answers are to be found in Section \ref{sec: GenExtremilesProperties}. Furthermore, in Section~\ref{sec:risks of distorted rvs} we derive properties such as e.g. monotonicity or coherence of distorted risk measures based on properties of the base risk measure and the utilized distortion function.
We provide an expression for the asymptotic variance for the estimator of the generalized extremile. Such a result allows to construct approximate confidence intervals for it. This however requires an estimator for the asymptotic variance. Although plug-in estimators are straightforward, it still remains to be established which are sufficient conditions for such a plug-in estimator to be consistent.
Another interesting aspect is the behaviour of the asymptotic variance as a function of the parameters $\btau$ and $\delta$. For certain functionals it would be interesting to look into such behaviour. Lastly regarding the asymptotic normality result, this has been established under a set of sufficient conditions. From the simulation study however we know that the result continues to 
hold in cases even when some of these conditions are violated. It would be interesting to look for a minimal set of conditions.
A final issue concerns the qualitative properties to be expected from a generalized extremile, such as monotonicity as a function of $\btau$, or non-crossingness for various values of $\delta$. See Section \ref{sec: realdata} and the discussion therein. 
An interesting open research question involves estimation of the generalized extremile which would guarantee these properties to hold for the estimators. 
\vspace*{0.4 cm}

\noindent
\textbf{Acknowledgement}
This work was initiated while the third author was visiting the KU Leuven.
The first and second author gratefully acknowledge support from the Research
Fund KU Leuven [C16/20/002 project].
The third author acknowledges support from NSERC through Discovery Grant RGPIN-2020-05784.


\renewcommand{\theequation}{\Alph{section}.\arabic{equation}}
\renewcommand{\thesection}{\Alph{section}} 
\renewcommand{\thesubsection}{\Alph{section}.\arabic{subsection}} 
\renewcommand{\thetable}{\Alphabetic{table}}  
\renewcommand{\thefigure}{\Alphabetic{figure}}
\setcounter{example}{0}
\setcounter{section}{0}
\setcounter{subsection}{0}
\setcounter{equation}{0}
\setcounter{figure}{0}
\setcounter{table}{0}

\vspace*{0.38 cm}
\newpage
\noindent
\textbf{{\Large Appendices}}
\section{Proofs of propositions in Section \ref{sec: GenExtremilesProperties}}\label{App: ProofsPropositions}
\subsection*{Proof of Proposition \ref{prop:signSymmetry}}
 We have $\widetilde{d}_{\btau}(t)=d_{\btau}(1-t)$ and by continuity of $F_X$ also $F_X(y)=$ $1-F_{-X}(-y)$ for all $y\in \mathbb{R}$. Defining
$$
L_{\btau,\delta}\left(c ; Y, D_{\btau},\ell_{\delta}\right)=\mathbb{E}\left[d_{\btau}\left(F_Y(Y)\right)(\ell_{\delta}(Y, c)-\ell_{\delta}(Y, 0))\right]
$$
and setting $Y=-X$ we have
$$
\begin{aligned}
L_{\btau,\delta}\left(c ; Y, \widetilde{D}_{\btau}, \ell_{\delta}\right) &=\mathbb{E}\left[\widetilde{d}_{\btau}\left(F_{-X}(-X)\right)(\ell_{\delta}(-X, c)-\ell_{\delta}(-X, 0))\right] \\
&=\mathbb{E}\left[d_{\btau}\left(1-F_{-X}(-X)\right)(\ell_{\delta}(X,-c)-\ell_{\delta}(X, 0))\right] \\
&=\mathbb{E}\left[d_{\btau}\left(F_X(X)\right)(\ell_{\delta}(X,-c)-\ell_{\delta}(X, 0))\right] \\
&=L_{\btau,\delta}\left(-c ; X, D_{\btau}, \ell_{\delta}\right).
\end{aligned}
$$
The minimizer $e_{\btau,\delta}\left(-X ; \widetilde{D}_{\btau}, \ell_{\delta}\right)$ of $L_{\btau,\delta}\left(c ;-X, \widetilde{D}_{\btau}, \ell_{\delta} \right)$ is therefore the negative of $e_{\btau,\delta}\left(X ; D_{\btau}, \ell_{\delta} \right)$ which minimizes $L_{\btau,\delta}\left(c ; X, D_{\btau}, \ell_{\delta} \right)$.
\qed

\subsection*{Proof of Proposition \ref{prop:transformations}}
Setting $Y=a X+b$, with $a>0$,  we have $F_Y(x)=F_X((x-b) / a)$ and hence $F_Y(Y)=F_X(X)$. Defining $L(c)=\mathbb{E}[d_{\btau}(F_Y(Y))(\ell_{\delta}(Y, c)-\ell_{\delta}(Y, 0))]$ we then get, based on the assumed properties of $\ell_{\delta}$, that
\begin{align*}
L(a c+b) & =\mathbb{E}[d_{\btau}(F_Y(Y))(\ell_{\delta}(Y, a c+b)-\ell_{\delta}(Y, 0))] \\
& =\mathbb{E}[d_{\btau}(F_X(X))(\ell_{\delta}(a X+b, a c+b)-\ell_{\delta}(a X+b, 0))] \\
& =\mathbb{E}[d_{\btau}(F_X(X))(a^k \ell_{\delta}(X, c)-\ell_{\delta}(a X+b, 0)+\ell_{\delta}(a X, 0)-\ell_{\delta}(a X, 0))]\\
&=a^k \mathbb{E}[d_{\btau}(F_X(X))(\ell_{\delta}(X,c)-\ell_{\delta}(X,0))]+\mathbb{E}[d_{\btau}(F_X(X))(\ell_{\delta}(aX,0)-\ell_{\delta}(aX+b,0))].
\end{align*}

Based on our assumptions the last equation (and hence $L(ac+b)$) is finite and by definition minimized for $c_0=e_{\btau,\delta}(X;D_{\btau},\ell_{\delta})$.  Hence, 
the minimal value of $L(c)$ is $L(ac_0+b)$. This value is attained for $c=ac_0+b.$ This shows 
$e_{\btau,\delta}(Y;D_{\btau},\ell_{\delta})=a e_{\btau,\delta}(X;D_{\btau},\ell_{\delta})+b.$\\

For $a<0$ set $\widetilde{a}=-a >0$. Applying the first part of the proposition then yields $e_{\btau,\delta}(\widetilde{a}(-X)+b; D_{\btau},\ell_{\delta})=\widetilde{a}e_{\btau,\delta}(-X;D_{\btau},\ell_{\delta})+b.$  Under the additional assumption we use Proposition \ref{prop:signSymmetry} to find 
$e_{\btau,\delta}(-X;D_{\btau},\ell_{\delta})=-e_{\btau,\delta}(X;\widetilde{D}_{\btau},\ell_{\delta})$, which gives $\widetilde{a}e_{\btau,\delta}(-X;D_{\btau},\ell_{\delta})+b= ae_{\btau,\delta}(X;\widetilde{D}_{\btau},\ell_{\delta})+b$.
\qed 

\subsection*{Proof of Proposition \ref{prop:kTransformation}}
Consider $\alpha(\cdot)$ a strictly increasing function. The quantity 
    $e_{\btau,\delta}(\alpha(X);D_{\btau},\ell_{\delta})$ is the minimizer of 
    $$\mathbb{E}[d_{\btau}(F_{\alpha(X)}(\alpha(X))(\ell_{\delta}(\alpha(X),c)-\ell_{\delta}(\alpha(X),0))]=\mathbb{E}[d_{\btau}(F_X(X))(\widetilde{\ell}_{\delta}(X,c)-\widetilde{\ell}_{\delta}(X,0)].$$
    We used here that $F_{\alpha(X)}(z)=F_X(\alpha^{-1}(z))$. Note that the latter expected value is by definition minimized by $e_{\btau,\delta}(X;D_{\btau},\widetilde{\ell}_{\delta})$ which proves the first statement.\\

    For the second statement, when considering $\alpha(\cdot)$ being strictly decreasing, we note that
    \begin{align*}
        \mathbb{E}[d_{\btau}(F_{\alpha(X)}(\alpha(X))(\ell_{\delta}(\alpha(X),c)-\ell_{\delta}(\alpha(X),0))]&=\mathbb{E}[d_{\btau}(1-F_{X}(X))(\ell_{\delta}(\alpha(X),c)-\ell_{\delta}(\alpha(X),0))]\\
        &=\mathbb{E}[\widetilde{d}_{\btau}(F_{X}(X))(\widetilde{\ell}_{\delta}(X,c)-\widetilde{\ell}_{\delta}(X,0))].
    \end{align*}
    For the first equality we used that since $\alpha$ is strictly decreasing and $F_X$ is continuous, it follows that  $F_{\alpha(X)}(z)=1-F_X(\alpha^{-1}(z))$. The first expected value is minimized when $c=e_{\btau,\delta}(\alpha(X);D_{\btau},\ell_{\delta})$. The latter expected value is minimized by $e_{\btau,\delta}(X;\widetilde{D}_{\btau},\widetilde{\ell_{\delta}})$.   
\qed 

\subsection*{Proof of Proposition \ref{prop:MeanSymmetry}}
 Mean symmetry is a standard property for minimizers of expected loss, see, for example, Proposition 2.1 in \cite{HerrmannEtAl2020} where the arguments only need to be adapted slightly to cover more general shift invariant loss functions. This leads directly to $$\frac{1}{2}(e_{\btau,\delta}(X ; D_{\btau}, \ell_{\delta})-e_{\btau,\delta}(-X ; D_{\btau}, \ell_{\delta}))=\mathbb{E}[X].$$
 Under the additional assumptions on $\ell_{\delta}$ and $F_X$ we use Proposition \ref{prop:signSymmetry} to obtain the second statement.
\qed

\subsection*{Proof of Proposition \ref{prop:ComonotonicAdditivity}}
In case of the absolute value loss, the generalized extremile is a fixed quantile of $X+Y$, where we have for comonotonic random variables $F_{X+Y}^{-1}(t)=$ $F_X^{-1}(t)+F_Y^{-1}(t)$. In case of the squared loss we can use the fact that
$$e_{\btau,\delta}(X;D_{\btau},\ell_{\delta})=\mathbb{E}[X_{\sD_{\btau}}]=\int_0^1 F_X^{-1}(u)d_{\btau}(u)\mathrm{d}u.$$
Indeed, this gives
\begin{align*}
e_{\btau,\delta}(X+Y ; D_{\btau}, \ell_{\delta})  =\int_0^1 F_{X+Y}^{-1}(u) d_{\btau}(u)\mathrm{d}u  &=\int_0^1(F_X^{-1}(u)+F_Y^{-1}(u)) d_{\btau}(u)  \mathrm{d}u \\
& =\int_0^1 F_X^{-1}(u) d_{\btau}(u) \mathrm{d} u+\int_0^1 F_Y^{-1}(u) d_{\btau}(u) \mathrm{d}u \\
& =e_{\btau,\delta}(X ; D_{\btau}, \ell_{\delta})+e_{\btau,\delta}(Y ; D_\tau, \ell_{\delta}).
\end{align*}
\vspace*{-1 cm}

\noindent
\hfill{\qed}

\subsection*{Proof of Proposition \ref{prop:injectiveTransformation}}
    Applying Proposition \ref{prop:equivalentDefinition} we know that 
    $$L(c):=\mathbb{E}[\ell_{\delta}(X_{\sD_{\btau}},c)-\ell_{\delta}(X_{\sD_{\btau}},0)]$$
    is minimal in $c=e_{\btau,\delta}(X;D_{\btau},\ell_{\delta})$.
    Hence, $$L^*(c):=\mathbb{E}[\ell_{\delta}(X_{\sD_{\btau}},\beta^{-1}(c))-\ell_{\delta}(X_{\sD_{\btau}},0)]$$
    is minimal for $c$ satisfying $\beta^{-1}(c)=e_{\btau,\delta}(X;D_{\btau},\ell_{\delta})$. Or thus $L^*(c)$ is minimal for $c=\beta(e_{\btau,\delta}(X;D_{\btau},\ell_{\delta}))$. Observe further that
    $$L^*(c)=\mathbb{E}[\ell_{\delta}(X_{\sD_{\btau}},\beta^{-1}(c))-\ell_{\delta}(X_{\sD_{\btau}},\beta^{-1}(0))]+\mathbb{E}[\ell_{\delta}(X_{\sD_{\btau}},\beta^{-1}(0))-\ell_{\delta}(X_{\sD_{\btau}},0)].$$
    The right-hand side is well defined and is by definition minimized by $e_{\btau,\delta}(X;D_{\btau},\widetilde{\ell}_{\delta})$, concluding the proof.
\qed
\smallskip

\section{Proofs of results in Section \ref{sec:estimation}}\label{App: ProofsConsistencyANresults}
\subsection{Proof of Theorem \ref{TheorySquareLossEstimation}}\label{App: ProofsConsistencyANresultsSquareLoss}
\subsubsection*{Preliminary results}
Before proving Theorem \ref{TheorySquareLossEstimation}, we establish a  theorem that is  an adaptation of results in \cite{Wellner1977}, which is needed in our context, where we avoid evaluation of $d_{\btau}$ in 0 or 1. The result in \hyperref[theorem:WellnerReplacement]{Theorem}   \ref{theorem:WellnerReplacement} is also of independent interest. 


\begin{theorem}\label{theorem:WellnerReplacement}
    Let $X_1,\cdots,X_n$ be a sample from $X$, where $\mathbb{E}[\lvert X \rvert^\kappa]<\infty$ for some $\kappa>0$. Suppose that 
    $$\lvert J^*(t)\rvert \leq M^*[t(1-t)]^{-1+\frac{1}{\kappa}+\gamma}, \quad 0<t<1$$
    for some $\gamma>0$ and $M^*>0$, where $J^*:(0,1)\to \mathbb{R}$ is continuous almost everywhere. Then 
    $$T_n^*:= \frac{1}{n}\sum_{i=1}^n J^*\left (\frac{i}{n+1}\right )X_{i:n} \xrightarrow[]{a.s.}\mu^*:=\int_0^1 J^*(t)F_{X}^{-1}(t)\mathrm{d}t, \quad \text{ as } n\to \infty.$$
\end{theorem}
\begin{proof}
    We start by applying Theorem 3 by \cite{Wellner1977} with $g=F_X^{-1}$, $b_1=b_2=1-\frac{1}{\kappa}-\gamma$. For each $n \in \mathbb{N}$ we define a function $J_n$ on $[0,1]$:  let  $J_n(t)=J^*(\frac{i}{n+1})$ for $\frac{i-1}{n}<t\leq \frac{i}{n}$, for $i=1, \dots, n$, and $J_n(0)=J^*(\frac{1}{n+1})$.  To apply Theorem 3 of \cite{Wellner1977}, we need to verify the following assumptions for a constant $M>0$
    \begin{itemize}
        \item[(a)] $\lvert F_X^{-1}(t)\rvert \leq M[t(1-t)]^{-\frac{1}{\kappa}}$ for all $t\in (0,1)$;
        \item[(b)] $\lvert J_n(t)\rvert \leq M[t(1-t)]^{-1+\frac{1}{\kappa}+\gamma}$ for all $t\in (0,1)$;
        \item[(c)] $\int_0^1 M [t(1-t)]^{\frac{1}{\kappa}+\frac{\gamma}{2}}\mathrm{d}F_X^{-1}(t)<\infty$.
    \end{itemize}
   For item (a), we distinguish some cases: $X$ having as support a bounded interval, the whole real line or a half-real line, say. Firstly, suppose that $X$ is bounded, i.e. has a finite lower and upper bound. Then $F_X^{-1}(t)$ is bounded for every $t\in (0,1)$. This implies $t\lvert F_X^{-1}(t)\rvert^\kappa \xrightarrow[]{}0$ as $t\xrightarrow[]{}0$. Similarly, $(1-t)\lvert F_X^{-1}(t)\rvert^\kappa \xrightarrow[]{}0$ as $t\xrightarrow[]{}1$. Secondly, suppose that $X$ has support the whole real line, i.e. is unbounded (no finite lower and upper bound). Then $\lvert F_X^{-1}(t)\rvert \xrightarrow[]{} +\infty$ as $t\xrightarrow[]{}0$, hence by using $\mathbb{E}[\lvert X \rvert^\kappa]<\infty$, we obtain
   \begin{equation}\label{eq:quantileBound}
    0 \le   t\lvert F_X^{-1}(t)\rvert^\kappa \leq \int_0^t \lvert F_X^{-1}(s)\rvert^\kappa \mathrm{d}s = \int_{-\infty}^{F_X^{-1}(t)}\lvert u \rvert^\kappa \mathrm{d}F_X(u) \xrightarrow{} 0 \text{ as } t \xrightarrow[]{} 0. 
   \end{equation}
    A similar reasoning gives $(1-t)\lvert F_X^{-1}(t)\rvert^\kappa \xrightarrow[]{}0$ as $t\xrightarrow[]{}1$. Thirdly, when the support of $X$ is unbounded at one side, then the arguments from the previous two cases can be combined. In conclusion, for all cases, we proved item (a) holds true. \\
    
    We now show item (b). When $-1+\frac{1}{\kappa}+\gamma\geq0$ the result is immediate since in this case $J^*$ is bounded by a constant and hence also $J_n$. We thus consider $-1+\frac{1}{\kappa}+\gamma<0$.   
    By assumption we have $J^*(t)\leq J_B^*(t):=M^*[t(1-t)]^{-\theta}$ where we denote $\theta=1-\frac{1}{\kappa}-\gamma$ to ease notation. Based on $J_B^*$, we define for every $n\in \mathbb{N}$ a function $J_{n,B}(t)=J^*_B(\frac{i}{n+1})$ for $\frac{i-1}{n}<t\leq \frac{i}{n}$ and $J_{n,B}(0)=J^*_B(\frac{1}{n+1})$. It is clear that $J_n(t)\leq J_B^*(t)$ for $t\in (0,1)$. Hence, it suffices to show $\lvert J_{n,B}(t)\rvert \leq K [t(1-t)]^{-\theta}$ for some constant $K$. We show that this holds with  $K=4^\theta M^*$. \\

    The difficulty lies in bounding $\lvert J_{n,B}(t)\rvert$ for $t$ close to the boundaries of $(0,1)$. For $t$ within a closed neighbourhood of $\frac{1}{2}$, the result is immediate. 
 Take $t \in (0,1)$, there exists a sequence $(i_n)_n \in \mathbb{N}$ 
 with $1\leq i_n \leq n$ such that $t\in (\frac{i_n-1}{n},\frac{i_n}{n}]$ for every $n$. Suppose that $\frac{i_n}{n}\leq \frac{1}{2}$, then $J_{n,B}(t)=J^*_B(\frac{i_n}{n+1})\leq J^*_B(t)\leq 4^\theta M^*[t(1-t)]^{-\theta}$ for $t \in (\frac{i_n-1}{n},\frac{i_n}{n+1}]$. Now consider $t\in (\frac{i_n}{n+1},\frac{i_n}{n}]$ and note $4^\theta M^*[\frac{i_n}{n}(1-\frac{i_n}{n})]^{-\theta}\leq 4^\theta M^*[t(1-t)]^{-\theta}$. Hence it suffices to show $J_{n,B}(t)=J^*_B(\frac{i_n}{n+1})\leq 4^\theta M^*[\frac{i_n}{n}(1-\frac{i_n}{n})]^{-\theta}$. Note 
    \begin{align*}
      &J^*_B\left (\frac{i_n}{n+1}\right )=M^*\Big[\frac{i_n(n+1-i_n)}{(n+1)^2}\Big]^{-\theta} \\
      &4^\theta M^*\Big[\frac{i_n}{n}\Big (1-\frac{i_n}{n}\Big )\Big]^{-\theta}=4^\theta M^* \Big[\frac{i_n(n-i_n)}{n^2}\Big]^{-\theta},
    \end{align*}
    or thus we need
    $$\Big[\frac{(n+1)^2}{i_n(n+1-i_n)}\cdot \frac{i_n(n-i_n)}{n^2}\Big]^\theta \leq 4^\theta.$$
    This is immediate since 
    $$\frac{(n+1)^2}{i_n(n+1-i_n)}\cdot \frac{i_n(n-i_n)}{n^2}\leq \frac{(n+1)^2}{n^2}\leq 4.$$
    Now suppose $\frac{1}{2}<\frac{i_n-1}{n}$, then $J_{n,B}(t)\leq J^*_B(t)\leq 4^\theta M^*[t(1-t)]^{-\theta}$ for $t \in [\frac{i_n}{n+1},\frac{i_n}{n}]$. Consider $t\in (\frac{i_n-1}{n},\frac{i_n}{n+1}]$ and note $4^\theta M^*[\frac{i_n-1}{n}(1-\frac{i_n-1}{n})]^{-\theta}\leq 4^\theta M^*[t(1-t)]^{-\theta}$. It suffices to show $J_{n,B}(t)=J^*_B(\frac{i_n}{n+1})\leq 4^\theta M^*[\frac{i_n-1}{n}(1-\frac{i_n-1}{n})]^{-\theta}$.
    Note that $4^\theta M^*[\frac{i_n-1}{n}(1-\frac{i_n-1}{n})]^{-\theta}=4^\theta M^* \Big[ \frac{(i_n-1)(n-i_n+1)}{n^2} \Big]^{-\theta}$. We thus need
    $$\Big[ \frac{n^2}{(i_n-1)(n-i_n+1)}\cdot \frac{i_n(n-i_n)}{n^2} \Big]^{\theta}\leq 4^\theta.$$
    This is immediate since $\frac{n}{i_n-1}<2$, and by noting that consequently 
    $$  \frac{n^2}{(i_n-1)(n-i_n+1)}\cdot \frac{i_n(n-i_n)}{n^2} \leq \frac{2}{(n-i_n+1)}\frac{i_n(n-i_n)}{n}\leq \frac{2(n-i_n)}{(n-i_n+1)}\leq 2.$$\\

    We are left to show item (c). Using integration by parts we obtain 
    \begin{align*}
        \Bigg \lvert \int_0^1 M [t(1-t)]^{\frac{1}{\kappa}+\frac{\gamma}{2}}\mathrm{d}F_X^{-1}(t) \Bigg \rvert&= \Bigg \rvert M[t(1-t)]^{\frac{1}{\kappa}+\frac{\gamma}{2}}F_X^{-1}(t) \Big\lvert_0^1+M\int_0^1   F_X^{-1}(t) \frac{\mathrm{d}}{\mathrm{d}t}[t(1-t)]^{\frac{1}{\kappa}+\frac{\gamma}{2}}  \mathrm{d}t \Bigg \rvert \\
        &\leq 0+ M\int_0^1 [t(1-t)]^{-\frac{1}{\kappa}} \Big \lvert  \frac{\mathrm{d}}{\mathrm{d}t}[t(1-t)]^{\frac{1}{\kappa}+\frac{\gamma}{2}} \Big \rvert \mathrm{d}t\\
        &\leq M' \int_0^1 [t(1-t)]^{\frac{\gamma}{2}-1}\mathrm{d}t < \infty \text{ since } -1+\frac{\gamma}{2}>-1.
    \end{align*}
    
    All together, we proved that the assumptions required to apply Theorem 3 by \cite{Wellner1977} are satisfied.
    This implies $T_n-\int_0^1 J_n(t)F_X^{-1}(t)\mathrm{d}t \xrightarrow[]{a.s.} 0$, as $n \to \infty$.
    We are left to show that $\int_0^1 J_n(t)F_X^{-1}(t)\mathrm{d}t\xrightarrow[]{a.s.}\mu^*$. To this end, we claim that $J^*(t)=\lim_{n\to \infty} J_n(t)$ for almost all $t \in (0,1)$.  Take $t\in (0,1)$ arbitrary, we assume that $J^*$ is continuous in $t$. 
    There exists a sequence $(i_n)_n \in \mathbb{N}$ 
 with $1\leq i_n \leq n$ such that $t\in (\frac{i_n-1}{n},\frac{i_n}{n}]$ for every $n$.
     Take $\epsilon>0$ arbitrary, by the continuity of $J^*$ in $t$, there exists some $n_0$ such that for all $n>n_0$ one has 
    $$ \lvert J^*(t)-J_n(t)\rvert =\left \lvert J^*(t)-J^*\left (\frac{i_n}{n+1}\right ) \right \rvert   < \epsilon.$$
    As $J^*$ is continuous almost everywhere, this demonstrates our assertion. We know that 
    $\lvert J_n(t) F_X^{-1}\rvert \leq M^2[t(1-t)]^{-1+\gamma},$
    which is integrable. Furthermore, we have shown that $J_n(t)F_X^{-1}(t)\xrightarrow[]{} J^*(t)F_X^{-1}(t)$ for almost all $t\in (0,1)$. 
    Hence, by Lebesgue's dominated convergence theorem (see \cite{Ziemer2017} p. 149)
    $$\int_0^1 J_n(t)F_X^{-1}(t)\mathrm{d}t \xrightarrow[]{} \mu^*<\infty \text{ as } n \to \infty.$$
   This concludes the proof.
\end{proof}
\medskip 

A second preliminary result that is helpful, is stated in Lemma 
\ref{lemma:WellnerReplacement}. This will serve directly for establishing asymptotic properties of the estimator $\widehat{e}_{\tau, n}^{\sL}$ defined in Section  \ref{SquareLossGeneralDtau}.

\begin{lemma}\label{lemma:WellnerReplacement}
     Let $X_1,\cdots,X_n$ be a sample from $X$, where $\mathbb{E}[\lvert X \rvert^\kappa]<\infty$ for some $\kappa>0$. Suppose that 
   $J^*:(0,1)\to \mathbb{R}$ is a (uniformly) bounded function that is continuous almost everywhere. Consider $T_n^*:= \frac{1}{n}\sum_{i=1}^n J^*(t_{i,n})X_{i:n},$
   where $t_{i,n}\in (\frac{i-1}{n+1},\frac{i}{n+1})$.
   Then $$T_n^*\xrightarrow[]{a.s.}\mu^*:=\int_0^1 J^*(t)F_{X}^{-1}(t)\mathrm{d}t, \quad \text{ as } n\to \infty.$$
\end{lemma}
\begin{proof}
We apply Theorem 3 by \cite{Wellner1977} with $g=F_X^{-1}$, $b_1=b_2=1-\frac{1}{\kappa}-\gamma$ for some $\gamma>0$. Define for every $n\in \mathbb{N}$ a function $J_n(t)=J^*(t_{i,n})$ for $\frac{i-1}{n}<t\leq \frac{i}{n}$ and $J_n(0)=J^*(t_{1,n})$. To apply Theorem 3 of \cite{Wellner1977}, we need to verify the following assumptions for a constant $M$
    \begin{itemize}
        \item[(a)] $\lvert F_X^{-1}(t)\rvert \leq M[t(1-t)]^{-\frac{1}{\kappa}}$ for all $t\in (0,1)$;
        \item[(b)] $\lvert J_n(t)\rvert \leq M[t(1-t)]^{-1+\frac{1}{\kappa}+\gamma}$ for all $t\in (0,1)$;
        \item[(c)] $\int_0^1 M [t(1-t)]^{\frac{1}{\kappa}+\frac{\gamma}{2}}\mathrm{d}F_X^{-1}(t)<\infty$.
    \end{itemize}
    Item (a) and (c) follow by the same reasoning as in the proof of Theorem \ref{theorem:WellnerReplacement}. Note that item (b), the boundedness of $J_n$ is now immediate since $J^*$ is bounded. We can complete the proof in the same way as in the proof of Theorem \ref{theorem:WellnerReplacement}.
    We only need a small adaptation in the argument for $J^*(t)=\lim_{n\to \infty} J_n(t)$ for almost every $t\in (0,1)$. Take $t\in (0,1)$ arbitrary, we assume that $J^*$ is continuous in $t$. There exists a sequence $(i_n)_n \in \mathbb{N}$ 
 with $1\leq i_n \leq n$ such that $t\in (\frac{i_n-1}{n},\frac{i_n}{n}]$ for every $n$. Take $\epsilon>0$ arbitrary. By continuity of $J^*$ in $t$, there exists a $\delta>$ such that for all $y \in (0,1)$ for which $\lvert t-y\rvert<\delta$ we have $\lvert J^*(t)-J^*(y)\rvert<\epsilon$.
    Note that $J_n(t)=J^*(t_{i_n,n})$ for $t\in (\frac{i_n-1}{n},\frac{i_n}{n}]$, now 
    $$\lvert t- t_{i_n,n}\rvert \leq \Big\lvert \frac{i_n}{n}-\frac{i_n-1}{n+1}\Big\rvert \leq \frac{2}{n+1}.$$
    Hence, we can find $n_0$ such that for all $n>n_0$ we have $\lvert t- t_{i_n,n}\rvert<\delta$ and thus $\lvert J^*(t)-J_n(t)\rvert<\epsilon$.
    As $J^*$ is continuous almost everywhere, this demonstrates our assertion and concludes the proof.
\end{proof}
\medskip 

Equipped with these preliminary results we prove Theorem \ref{TheorySquareLossEstimation}.

\subsubsection*{Proof of Theorem \ref{TheorySquareLossEstimation}}
 We consider the three estimators separately. 
 
 For $\widehat{e}_{\btau,n}^{\sLM}$ the consistency result follows directly from Theorem \ref{theorem:WellnerReplacement}. For the asymptotic normality we can use Theorem 1 (ii) from p. 664 in \cite{ShorackWellner2009}.\\

Next we study the estimator $\widehat{e}_{\btau,n}^{\sM}$. Since $d_{\btau}$ is integrable, the Riemann sum
    $(n+1)^{-1}\sum_{i=1}^{n+1} d_{\btau}(\frac{i}{n+1})$
    converges to $\int_{0}^{1}d_{\btau}(u)\mathrm{d}u=1$, as $n \to \infty$.
    As $d_{\btau}$ is a density, we can assume $d_{\btau}(1)=0$, this implies 
    also  $(n+1)^{-1}\sum_{i=1}^{n} d_{\btau}(\frac{i}{n+1})$ converges to $1$ and hence also $\frac{1}{n}\sum_{i=1}^{n} d_{\btau}(\frac{i}{n+1})=\frac{n+1}{n}\cdot \Big(\frac{1}{n+1}\sum_{i=1}^{n} d_{\btau}(\frac{i}{n+1})\Big) $ converges to $1$.
    Using 
     $$\widehat{e}^{\sM}_{\btau,n}=\frac{\widehat{e}^{\sLM}_{\btau,n}}{\frac{1}{n} \sum_{i=1}^n d_{\btau}(\frac{i}{n+1})},$$
     we obtain the almost sure convergence of $\widehat{e}^{\sM}_{\btau}$ as a fraction of almost sure converging sequences (see p. 24 in \cite{Serfling1980}).\\

     Note that $\widehat{e}^{\sLM}_{\btau,n}$ converges in distribution, and that $n^{-1} \sum_{i=1}^n d_{\btau}(\frac{i}{n+1})$ converges in probability to a constant. Hence we obtain the asymptotic normality by Slutsky's  Theorem (see p. 19 in \cite{Serfling1980}).\\

    Finally, we turn to the estimator  $\widehat{e}_{\btau,n}^{\sL}$.
    Denote $$\widehat{e}_{\btau,n}^{\sL}=\frac{n}{n+1}\cdot \Big(\frac{1}{n} \sum_{i=1}^n w_{i, n} X_{i:n}\Big) \qquad \text{ with } w_{i,n}=(n+1) \Big( D_{\btau}(i/(n+1))-D_{\btau}((i-1)/(n+1)) \Big).$$
    Since $D_{\btau}$ is continuously differentiable on $(0,1)$ with derivative $d_{\btau}$, we have by the mean value theorem
    $$\forall i \in \{1,\cdots,n\}, \exists t_{i,n} \in \Big(\frac{i-1}{n+1},\frac{i}{n+1}\Big), \quad \mbox{such that} \,\,  w_{i,n}=d_{\btau}(t_{i,n}).$$
    In case $d_{\btau}$ is bounded, we can apply Lemma \ref{lemma:WellnerReplacement} to obtain consistency for $n^{-1} \sum_{i=1}^n w_{i, n} X_{i:n}$. It is clear this implies consistency for $\widehat{e}_{\btau,n}^{\sL}$. The asymptotic normality follows by Theorem 1(ii) on p. 664 in \cite{ShorackWellner2009}. 
    The assumptions required for this latter theorem are satisfied. Indeed, Assumptions 1 on p. 662 in \cite{ShorackWellner2009} is satisfied since $d_{\btau}$ is bounded. Furthermore, Assumption 2' on p. 664 in \cite{ShorackWellner2009} is satisfied using the possibility to extend the interval mentioned by the authors.\\

    Suppose on the other hand that $d_{\btau}$ is Lipschitz, with constant $A >0$. 
    Consider 
    \begin{align}
        \Big\lvert \widehat{e}_{\btau,n}^{\sLM}- \frac{1}{n} \sum_{i=1}^n w_{i, n} X_{i:n}\Big \rvert &=\left \lvert \frac{1}{n}\sum_{i=1}^n \left [\, d_{\btau}\left  (\frac{i}{n+1}\right )-d_{\btau}(t_{i,n})\, \right ]X_{i:n}  \right \rvert \label{eq:squareLossConsistencyProof}\\
        &\leq \frac{A}{n}\sum_{i=1}^n \frac{\lvert X_{i:n} \rvert }{n+1}\leq \frac{A}{n+1}\Big( \frac{1}{n}\sum_{i=1}^n \lvert X_{i:n}\rvert \Big)\xrightarrow[]{a.s.} 0,\nonumber
    \end{align}
if $\mathbb{E}[\lvert X \rvert]<\infty$. As before, this implies consistency of  $\widehat{e}_{\btau,n}^{\sL}$. The asymptotic normality can be shown in the same way by adding a factor $\sqrt{n}$ to (\ref{eq:squareLossConsistencyProof}).\\

For the consistency results, we require $\kappa>1$, this ensures $\left | e_{\btau}(X; D_{\btau})\right | <\infty$, i.e. the existence of the target quantity. To see this note that 
\[
    \lvert e_{\btau}(X; D_{\btau}) \rvert = \Big \lvert \int_{0}^{1}F_X^{-1}(t)d_{\btau}(t)\mathrm{d}t\Big \rvert 
    \leq \int_0^1 M^2 [t(1-t)]^{-1+\gamma}\mathrm{d}t<\infty.
\]
We used that $\mathbb{E}[\lvert X\rvert^\kappa]<\infty$ implies $\lvert F_X^{-1}(t)\rvert \leq M[t(1-t)]^{-\frac{1}{\kappa}}$ (see (\ref{eq:quantileBound}) in the proof of Theorem \ref{theorem:WellnerReplacement}).\\

For the asymptotic normality results, we require $\kappa>2$. By Lemma 1 on p. 663 in \cite{ShorackWellner2009} this ensures that $\sigma_{\btau}^2<\infty$.
\qed

\subsection{Proof of theoretical results in Section \ref{sec:estimation--GeneralCase}}\label{App: ProofsConsistencyANresultsGenLoss}
\subsubsection*{Proof of Lemma \ref{lemma:consistency}}
Given that $c\mapsto \ell_\delta(x,c)$ is convex for every $x$, we know that $c \mapsto \ell_\delta'(x,c)$ is non-decreasing.
By definition, it follows that also $\lambda_F(c),\lambda_{F,n}(c)$ and $\lambda_{\widehat{F}_n}^*(c)$ are non-decreasing in $c$. Note
\begin{align*}
    \lvert \lambda_{\widehat{F}_n}^*(c)-\lambda_{F,n}(c) \rvert &= \Big\lvert \frac{1}{n}\sum_{i=1}^n \ell_{\delta}'(X_i,c)[d_{\btau}(\widehat{F}_n(X_i))-d_{\btau}(F_X(X_i))]\Big\rvert\\
    &\leq  \frac{1}{n}\sum_{i=1}^n \Big\lvert  \ell_{\delta}'(X_i,c) \Big\rvert \cdot \lvert d_{\btau}(\widehat{F}_n(X_i))-d_{\btau}(F_X(X_i)) \rvert \\
    &\leq \frac{1}{n}\sum_{i=1}^n \Big\lvert  \ell_{\delta}'(X_i,c) \Big\rvert \cdot \sup_{x\in \mathbb{R}} \lvert d_{\btau}(\widehat{F}_n(x))-d_{\btau}(F_X(x)) \rvert \\
    &\leq \frac{1}{n}\sum_{i=1}^n \Big \lvert  \ell_{\delta}'(X_i,c) \Big \rvert \cdot  M \sup_{x\in \mathbb{R}} \lvert \widehat{F}_n(x)-F_X(x) \rvert\\
    &\xrightarrow{\text{a.s.}} M\cdot \mathbb{E}\Big [ \Big \lvert \ell_{\delta}'(X,c) \Big \rvert \Big] \cdot 0 = 0
\end{align*}
Herein we used the Lipschitz continuity of $d_{\btau}$ (with Lipschitz constant $M >0$), and the strong law of large numbers.
This implies for each $c$ that $\lambda_{\widehat{F}_n}^*(c)-\lambda_{F,n}(c)\xrightarrow{\text{a.s.}} 0$. Furthermore, by the finiteness of $\lambda_F(c)$ we obtain, by again applying the strong law of large numbers, that $\lambda_{F,n}(c)\xrightarrow{\text{a.s.}} \lambda_F(c)$. In total, we demonstrated that
$$\lambda_{\widehat{F}_n}^*(c)- \lambda_{F}(c)=
(\lambda_{\widehat{F}_n}^*(c)-\lambda_{F,n}(c))
+(\lambda_{F,n}(c)-\lambda_{F}(c))
\xrightarrow{\text{a.s.}} 0, \qquad \mbox{as }  n \to \infty.$$

Now take $\epsilon>0$ arbitrary.  Observe that  $\lambda_F(t_0-\epsilon)\leq 0
\leq \lambda_F(t_0+\epsilon)$. By the above we have almost surely
\begin{align*}
    &\lambda_{\widehat{F}_n}^*(t_0-\epsilon) \xrightarrow{} \lambda_F(t_0-\epsilon)\leq 0\\
    &\lambda_{\widehat{F}_n}^*(t_0+\epsilon) \xrightarrow{} \lambda_F(t_0+\epsilon)\geq 0.
\end{align*}
Meaning 
\begin{equation*}
\lim_{n\to \infty}P( \lambda_{\widehat{F}_m}^*(t_0-\epsilon) \leq 0 \leq  \lambda_{\widehat{F}_m}^*(t_0+\epsilon), \text{ all }m\geq n  )=1.    
\end{equation*}
By the monotonicity of $\lambda_{\widehat{F}_n}^*(c)$, this means that for almost all solution sequences $T_n^*$ we have from some $n$ on that 
$$t_0-\epsilon \leq T_n^* \leq t_0+\epsilon.$$
Or thus we get
\[ \lim_{n\to \infty}P( t_0-\epsilon \leq T_m^*\leq t_0+\epsilon, \text{ all }m\geq n  )=1\qquad 
\Longleftrightarrow \qquad \lim_{n\to \infty}P(\lvert T_m^*-t_0\rvert \leq \epsilon, \text{ all }m\geq n  )=1, 
\]
which shows $T_n^*$ converges almost surely to $t_0$. By the convexity assumption, we know by Proposition \ref{prop: convexLossDeriviativeZero} that $t_0=e_{\btau,\delta}(X;D_{\btau},\ell_{\delta})$ (keeping also in mind Assumption (\textbf{A1})2.).
\qed
\medskip

\subsubsection*{Proof of Lemma \ref{lemma:almostSureConvergence}} 
When $\widehat{F}_n=F_n$ the empirical cumulative distribution function, i.e. $F_n(x)=(n+1)^{-1}\sum_{i=1}^n 1_{\{X_i\leq x\}}$ we can write 
    $$\lambda_{F_n}^*(c)=\frac{1}{n}\sum_{i=1}^n d_{\btau}\left(\frac{i}{n+1}\right)h_1(X_{i:n},c)-\frac{1}{n}\sum_{i=1}^n d_{\btau}\left(\frac{i}{n+1}\right)h_2(X_{i:n},c).$$
    
If $\mathbb{E}[\lvert h_j(X,c) \rvert^r]<\infty$ for $j=1,2$ for some $r>0$, then Theorem B.1. gives for $j=1,2$
     \begin{equation*}
             \frac{1}{n}\sum_{i=1}^n d_{\btau}\left(\frac{i}{n+1}\right)
             h_j(X_{i:n},c)
             \xrightarrow[]{a.s.} \int_0^1 d_{\btau}(t)F_{h_j(X,c)}^{-1}(t)\mathrm{d}t,  \quad \mbox{as } n \to \infty ,
     \end{equation*}
    where $F_{h_j(X,c)}^{-1}(t)$ is the quantile function of the random variable $h_j(X,c)$ for $j=1,2$.\\
     
Since we assume left continuity of $x\mapsto h_j(x,c)$ we have by p. 136 in \cite{Hosseini2009} that $F_{h_j(X,c)}^{-1}(t)=h_j(F_X^{-1}(t),c)$ for $j=1,2$. This implies 
   $$\lambda_{F_n}^*(c) \xrightarrow[]{a.s.} \int_0^1 d_{\btau}(t)[h_1(F_X^{-1}(t),c)-h_2(F_X^{-1}(t),c)]\mathrm{d}t=\lambda_F(c), \quad \mbox{as } n \to \infty.$$
   We can complete the proof in exactly the same way as  in Lemma \ref{lemma:consistency}. 
\qed 
\medskip

\subsubsection*{Proof of Lemma \ref{lemma:asympLambda} }  
As before, we will use the fact that we can write $\lambda_{F_n}^*(c)$ as
$$\lambda_{F_n}^*(c)=\frac{1}{n}\sum_{i=1}^n d_{\btau}\left( \frac{i}{n+1} \right)\ell_\delta'(X_{i:n},c).$$
     By Theorem 1(ii) on p. 664 in \cite{ShorackWellner2009} we obtain asymptotic normality of $\lambda_{F_n}^*(t_0)$, with asymptotic variance
    \begin{equation}\label{eq:lemmaAsympLambdaEq1}
        \int_0^1 \int_0^1 [\min\{x,y\}-xy]d_{\btau}(x)d_{\btau}(y)\mathrm{d}\ell_{\delta}'(F_X^{-1}(x),t_0)\mathrm{d}\ell_{\delta}'(F_X^{-1}(y),t_0).
    \end{equation}
    It is clear that this expression is equivalent to (\ref{eq: sigmat01}). 
    Since $x\mapsto h_j(x,c)$ is non-decreasing and left continuous we have $F_{h_j(X,t_0)}^{-1}(x)=h_j(F_X^{-1}(x),t_0)$ for $j=1,2$ (see p. 136 in \cite{Hosseini2009}). This implies
\[
        \ell_{\delta}'(F_X^{-1}(x),t_0)= h_1(F_X^{-1}(x),t_0)-h_2(F_X^{-1}(x),t_0)
        =F_{h_1(X,t_0)}^{-1}(x)-F_{h_2(X,t_0)}^{-1}(x).
\]
The above shows that, with denoting  $k(x,y)=[\min\{x,y\}-xy]d_{\btau}(x)d_{\btau}(y)$, 
\begin{align*}
    \int_0^1\int_0^1 k(x,y)\mathrm{d}\ell_{\delta}'(F_X^{-1}(x),t_0)\mathrm{d}\ell_{\delta}'(F_X^{-1}(y),t_0)&=\int_0^1 \int_0^1 k(x,y)\mathrm{d}F_{h_1(X,t_0)}^{-1}(x)\mathrm{d}F_{h_1(X,t_0)}^{-1}(y)\\
    &\hspace{1cm}+\int_0^1 \int_0^1 k(x,y)\mathrm{d}F_{h_2(X,t_0)}^{-1}(x)\mathrm{d}F_{h_2(X,t_0)}^{-1}(y)\\&\hspace{1cm}-\int_0^1 \int_0^1 k(x,y)\mathrm{d}F_{h_1(X,t_0)}^{-1}(x)\mathrm{d}F_{h_2(X,t_0)}^{-1}(y)\\
    &\hspace{1cm}-\int_0^1 \int_0^1 k(x,y)\mathrm{d}F_{h_2(X,t_0)}^{-1}(x)\mathrm{d}F_{h_1(X,t_0)}^{-1}(y).
\end{align*}
After last transformations, involving changing the integration variables $F_{h_2(X,t_0)}^{-1}(x)$ and $_{h_1(X,t_0)}^{-1}(y)$, to (say) $u$ and $v$, it follows that the asymptotic variance formula given in (\ref{eq:lemmaAsympLambdaEq1}) is also equal to (\ref{eq: sigmat02}). Lemma 1 on p. 663 in \cite{ShorackWellner2009} guarantees that $\sigma_{t_0}^2<\infty$.
\qed

\medskip

\subsubsection*{Proof of Lemma \ref{lemma:asympLambdaSecondVersion}} 
Theorem 2 on p. 1452 in \cite{Wozabal2009} leads to the desired result, by applying  this theorem with $\mathbb{M}=\mathbb{R}$. In the notation of \cite{Wozabal2009}, we take $c_{ni}(m)=d_{\btau}\left(\frac{i}{n+1}\right)$ for every $m\in \mathbb{M}$ and $h(t)=\ell_\delta'(t,c)$. As we consider $m\mapsto c_{ni}(m)=d_{\btau}\left(\frac{i}{n+1}\right)$  to be a constant function in $m$, we obtain $\lVert c_{ni}\rVert_{\mathbb{M}}:=\sup_{m\in \mathbb{M}}\lvert c_{ni}(m)\rvert = \left 
\lvert d_{\btau}\left(\frac{i}{n+1}\right)\right \rvert < \infty $ (by Assumption \textbf{(A8)}). \\

Next, we define for every $m\in \mathbb{M}$ the function $J_n^{(m)}(t)=d_{\btau}\left(\frac{i}{n+1}\right)$ for $\frac{i-1}{n}<t\leq \frac{i}{n}$ for $i=1,\cdots,n$ with $J_{n}^{(m)}(0)=d_{\btau}\left( \frac{1}{n+1} \right)$. By further defining $J=d_{\btau}$, it is noted that $J$ and $J_n^{(m)}$ are uniformly bounded for every $m\in \mathbb{M}$, hence the first item of Assumption 1 in \cite{Wozabal2009} is satisfied with $b_1=b_2=0$. Also item 2 of Assumption 1 in \cite{Wozabal2009} is satisfied, due to Assumptions \textbf{(A3)} and \textbf{(A4)}.\\

By our choices we obtain $\lVert J(s)-J(t)\rVert_{\mathbb{M}}=\lvert d_{\btau}(s)-d_{\btau}(t)\lvert$ and hence the equicontinuity condition in Theorem 2 by \cite{Wozabal2009} is satisfied by of the continuity of $d_{\btau}$, except at possibly a finite number of points.  Finally, the  uniform convergence assumption in \cite{Wozabal2009} is satisfied by our Assumption \textbf{(A8)}. Indeed, take $\epsilon>0$ and $x\in [0,1]$ arbitrary (but excluding possibly a finite number of points). We can assume there exists a neighborhood $(x_0,x_1)$ around $x$ such that $d_{\btau}$ is uniformly continuous in it. Now consider a neighborhood $(u_0,u_1)$ which is a strict subset of $(x_0,x_1)$ and contains $x$. By the uniform continuity of $d_{\btau}$, there exists a $\delta>0$ such that $\forall y,z\in (u_0,u_1): \lvert y-z\rvert <\delta \implies \lvert d_{\btau}(y)-d_{\btau}(z)\rvert < \epsilon$. Define $n_0\in \mathbb{N}$ such that $\frac{1}{n_0}<\min\{u_0-x_0,x_1-u_1,\delta\}$. Take $v\in(u_0,u_1)$ arbitrary. There exists a sequence $(i_n)_n\in \mathbb{N}$ with $1\leq i_n\leq n$ such that $v \in (\frac{i_n-1}{n},\frac{i_n}{n}]$ for every $n$. 
Now observe that for every $n\in \mathbb{N}$ with $n>n_0$ 
$$\lvert J_n^{(m)}(v)-J(v)\rvert =\left \lvert d_{\btau}\left( \frac{i_n}{n+1} \right) -d_{\btau}(v) \right \rvert <\epsilon, $$
hence the local uniform convergence is satisfied.
This shows all the assumptions required for Theorem 2 in \cite{Wozabal2009} are satisfied.\\

We end this proof by noting that Lemma 1 on p. 663 in \cite{ShorackWellner2009} guarantees $\sigma^2_{t_0}<\infty$.
\qed 
\medskip

\section{Proof of Theorem \ref{GenPropDistRisk}} \label{App: RisksComposition}

\begin{itemize}
\item[(i)]
We start with the positive homogeneity. First, we show that $(aX)_{\sD} \stackrel{d}{=} a X_{\sD}$. This follows directly 
    by noting that
    $$F_{(aX)_{\sD}}(y)=D(F_{aX}(y))=D\left (F_X\left (y/a\right )\right ),$$
    and
   $$F_{a X_{\sD}}(y)=F_{X_{\sD}}(y/a)=D(F_X(y/a)).$$
    Using this, we know that    $\rho((aX)_{\sD})=\rho(a X_{\sD})$.
    By the assumption that $\rho$ is positive homogeneous, the claim is proven.    
\item[(iii)]
For the monotonicity, we recall that $X \leq Y$ means $P(X \leq Y)=1$ and is equivalent with $\forall x \in \mathbb{R}: F_{X}(x)\geq F_{Y}(x).$ As $D$ is an increasing function it follows directly 
    that $\forall x \in \mathbb{R}: D(F_{X}(x))\geq D(F_{Y}(x))$ and hence that $X_{\sD} \leq Y_{\sD}$.
\item[(iv)]
Consider $\rho$ a coherent risk measure, then it follows by Proposition 5 on p. 1501 in \cite{LiuSchiedWang2021} that also the composition risk measure is coherent when $D$ is convex. 
     This concludes the proof.
\item[(ii)]
For the translation invariance, it suffices to note
    $$F_{(X+c)_{\sD}}(y)=D(F_{X+c}(y))=D(F_X(y-c)),$$
    and
    $$F_{X_{\sD}+c}(y)=P(X_{\sD}+c\leq y)=D(F_X(y-c)).$$
\item[(v)]
To prove this property, it is enough to demonstrate that if $D_1 \le D_2$, then $X_{\sD_2} \leq X_{\sD_1}$. Note that $X_{\sD_2} \leq X_{\sD_1}$  is ensured when $D_{2}(F_X(x))\geq D_{1}(F_X(x)) $ for all $x\in \mathbb{R}$. This holds by the fact that $D_1 \le D_2$.
\item[(vi)]
The proof of this property is trivial. 

\end{itemize}


\newpage
\noindent
\renewcommand{\theequation}{S\arabic{section}.\arabic{equation}}
\renewcommand{\thesection}{S\arabic{section}} 
\renewcommand{\thesubsection}{S\arabic{section}.\arabic{subsection}}
\renewcommand{\theexample}{S\arabic{example}} 
\renewcommand{\thetable}{S\arabic{table}}  
\renewcommand{\thefigure}{S\arabic{figure}}
\setcounter{page}{1}
\setcounter{example}{0}
\setcounter{section}{0}
\setcounter{subsection}{0}
\setcounter{equation}{0}
\setcounter{figure}{0}
\setcounter{table}{0}
\begin{center}
{\large {Supplementary Material} \\
to }\\
{\Large{\bf Generalized extremiles and risk measures of distorted random variables
}}
\vspace*{0.4 cm}

\noindent
by 
\vspace*{0.24 cm}

\noindent
{\large Dieter Debrauwer,  Ir\`ene Gijbels and Klaus Herrmann}
\end{center}

\section{Some additional loss functions} \label{supsec: LossFunctions} 
In this section we briefly discuss additional loss functions. For each we indicate to which quantity of the distorted random variable $X_{\sD_{\btau}}$ they lead to. 
\begin{itemize}
\item[$\bullet$]
Consider the loss function $\ell(x,c)=(1-c)^2 1_{\{ x\geq0 \}}x -c^2 x 1_{\{x<0 \}}$. Then 
\[
    \int \ell_{\delta}(x,c)\mathrm{d} D_{\btau} (F_X(x))=(1-c)^2 \int_{0}^{+\infty} x \mathrm{d} D_{\btau} (F_X(x))-c^2 \int_{-\infty}^{0} x \mathrm{d} D_{\btau} (F_X(x)).
\]
Hence 
\[
   \frac{\mathrm{d}}{\mathrm{d} c}\Big(  \int \ell_\delta(x,c)\mathrm{d} D_{\btau} (F_X(x))\Big) = 2(c-1)\int_{0}^{+\infty} x \mathrm{d} D_{\btau} (F_X(x))-2c\int_{-\infty}^{0} x \mathrm{d} D_{\btau} (F_X(x)). 
\]
Equating to zero results in the observation that the generalized extremile is equal to the value $c$ for which 
$$c=\frac{\mathbb{E}[X_{\sD_{\btau}} \cdot 1_{\{X_{\sD_{\btau}}\ge 0\}}]}{\mathbb{E}[X_{\sD_{\btau}} \cdot 1_{\{X_{\sD_{\btau}} \ge 0\}}]-\mathbb{E}[X_{\sD_{\btau}} \cdot 1_{\{X_{\sD_{\btau}}  <  0\}}]}=\frac{\mathbb{E}[X_{\sD_{\btau}}\cdot 1_{\{X_{\sD_{\btau}}\ge 0\}}]}{\mathbb{E}[\lvert X_{\sD_{\btau}} \rvert ]}.$$ 
Note that the function $c\mapsto \ell_\delta(x,c)$ is convex in $c$, for every $x$.
    Furthermore, the map $c\mapsto \ell_\delta(x,c)$ is  differentiable for all $c\neq x$. The derivative is equal to $\ell'_\delta(x,c)=-2(1-c)x1_{\{x\geq 0\}}-2cx\delta 1_{\{x<0\}}$.
\item[$\bullet$]
Consider $\ell(x,c)=2 1_{\{x<c\}}(c-x)+(x-c)^2$. Then 
\[
  \int \ell_\delta(x,c)\mathrm{d} D_{\btau} (F_X(x))=2 \int_{-\infty}^c  (c-x)\mathrm{d} D_{\btau} (F_X(x))
  + \int_{-\infty}^{+\infty} (x-c)^2 \mathrm{d} D_{\btau} (F_X(x)).  
\]
Hence 
\begin{align*}  
\frac{\mathrm{d}}{\mathrm{d} c}\Big(  \int \ell_\delta(x,c)\mathrm{d} D_{\btau} (F_X(x))\Big)=&   2 \int_{-\infty}^{c}\mathrm{d} D_{\btau} (F_X(x)) \\& \quad -2 \int_{-\infty}^{+\infty} x \mathrm{d} D_{\btau} (F_X(x))+2c \int_{-\infty}^{+\infty}\mathrm{d}D_{\btau} (F_X(x)).
\end{align*}
Equating to zero results in the observation that the generalized extremile is equal to the value $c$ for which 
$$F_{X_{\sD_{\btau}}}(c)-\mathbb{E}[X_{\sD_{\btau}}]+c=0
\quad \Longleftrightarrow \quad D_{\btau}\left (F_X(c) \right ) -\mathbb{E}[X_{\sD_{\btau}}]+c=0.$$ 
To be observed is that the  $c\mapsto \ell_\delta(x,c)$ is convex in $c$, for every $x$. Furthermore, the map $c\mapsto \ell_\delta(x,c)$ is  differentiable for all $c\neq x$. The derivative  with respect to $c$ equals $\ell'_\delta(x,c)=21_{\{x<c\}}-2(x-c)$.
\item[$\bullet$]
\cite{CaiAndWang2019} consider  a loss function (for positive random variables) of the form 
$$\ell_\delta(x,c)=g(x) \phi_1((I(x)-c)_+)+h(x)\phi_2(x-I(x)+c).$$
The notation $(x)_+=\max\{x,0\}$ is used. Furthermore,  $\phi_1,\phi_2,g,h$ are functions $\mathbb{R}^+ \to [0,\infty)$, and $I$ is a measurable function satisfying 
$$0\leq I(x)\leq x \text{ for any } x\geq 0.$$
The qualitative properties of $\ell_{\delta}(x,c)$ are determined by the functions
$\phi_1,\phi_2,g,h$ and $I$ involved. 
\end{itemize} 

\section{Regarding the equicontinuity condition in the asymptotic normality result of Theorem \ref{theorem: asymptoticNormalityEmpirical}} 
\subsection{Squared loss}
For square loss, we know $\ell_{\delta}'(x,c)=-2(x-c)$, and $t_0=\mathbb{E}[X_{\sD_{\btau}}]$. This gives
\begin{align*}
& \sup_{\lvert c - t_0 \rvert \leq \omega_n} \sqrt{n} \lvert \lambda_{F_n}^*(c)-\lambda_{F_n}^*(t_0)-\lambda_F(c) \rvert \\
& \qquad   = \sup_{\lvert c - t_0 \rvert \leq \omega_n} \left \lvert \frac{1}{\sqrt{n}} \sum_{i=1}^n d_{\btau}(F_n(X_i))[-2(X_i-c) \right .\\
&\qquad \qquad \left . +2(X_i-t_0)]+2 \sqrt{n}\int d_{\btau}(F_X(x))(x-c)\mathrm{d}F_X(x) \right \rvert \\
&\qquad = \sup_{\lvert c - t_0 \rvert \leq \omega_n} \left \lvert \frac{1}{\sqrt{n}} \sum_{i=1}^n d_{\btau}(F_n(X_i)) (2(c-t_0)) + 2\sqrt{n} \int d_{\btau}(F_X(x)) x \mathrm{d}F_X(x)- 2\sqrt{n}c \, \right \rvert \\
 &\qquad = \sup_{\lvert c - t_0 \rvert \leq \omega_n} \left \lvert \frac{1}{\sqrt{n}} \sum_{i=1}^n d_{\btau}(F_n(X_i)) (2(c-t_0)) +  2\sqrt{n}(t_0-c) \right \rvert \\
&\qquad \leq \sup_{\lvert c - t_0 \rvert \leq \omega_n} 2\sqrt{n} \lvert c-t_0 \rvert \cdot \left \lvert \frac{1}{n}\sum_{i=1}^n d_{\btau}\left (\frac{i}{n+1}\right )-1 \right \rvert \\
&\qquad \leq 2 \sqrt{n} \omega_n \left \lvert  \frac{1}{n}\sum_{i=1}^n d_{\btau}\left (\frac{i}{n+1}\right )-1 \right \rvert.
\end{align*}

Suppose $d_{\btau}$ is Lipschitz continuous with constant $M$. Then
\begin{align*}
    \left \lvert \frac{1}{n} \sum_{i=1}^n \left (d_{\btau}\left (\frac{i}{n+1}\right)-1 \right ) \right \rvert =  \left \lvert \frac{1}{n} \sum_{i=1}^n \left (d_{\btau}\left (\frac{i}{n+1}\right )-\int_{0}^1 d_{\btau}(x)\mathrm{d}x\right) \right \rvert 
    &=\Big \lvert \sum_{i=1}^n \Big(\frac{1}{n}d_{\btau}\left (\frac{i}{n+1}\right)- \int_{\frac{i-1}{n}}^{\frac{i}{n}} d_{\btau}(x)\mathrm{d}x \Big)  \Big \rvert\\
    &\leq \sum_{i=1}^n \left \lvert  \frac{1}{n}d_{\btau}\left (\frac{i}{n+1}\right )- \int_{\frac{i-1}{n}}^{\frac{i}{n}} d_{\btau}(x)\mathrm{d}x\right \rvert  \\
    &\leq \sum_{i=1}^n \int_{\frac{i-1}{n}}^{\frac{i}{n}}  \Big \lvert d_{\btau}\left (\frac{i}{n+1}\right)-d_{\btau}(x)  \Big \rvert \mathrm{d}x\\
    &\leq \sum_{i=1}^n \frac{M}{n^2}= \frac{M}{n}.
\end{align*}

Hence 
$$\sup_{\lvert c - t_0 \rvert \leq \omega_n} \sqrt{n} \lvert \lambda_{F_n}^*(c)-\lambda_{F_n}^*(t_0)-\lambda_F(c) \rvert \leq \Big \lvert   2\omega_n \frac{M}{\sqrt{n}} \Big \rvert \to 0 \text{ when } \omega_n \to 0.$$\\

\subsection{Loss function in entry G2 of Table \ref{LossFunctions2A}}
Consider the loss function $\ell_{\delta}(x,c)=-c\lvert x - b \rvert^\delta +\frac{c^2}{2}$ corresponding to $t_0=\mathbb{E}[\lvert X_{\sD_{\btau}} -b \rvert^\delta]$. Here, $\ell_{\delta}'(x,c)=-\lvert x-b\rvert^\delta +c$.  We have  
\begin{align*}
    \sqrt{n} \lvert \lambda_{F_n}^*(c)-\lambda_{F_n}^*(t_0)-\lambda_F(c) \rvert =\Big \lvert \frac{1}{\sqrt{n}} \sum_{i=1}^n d_{\btau}(F_n(X_i))(c-t_0) &\hspace{0.2cm}+  \sqrt{n} \int_{-\infty}^{+\infty} \lvert x -b\rvert^\delta d_{\btau}(F_X(x))f_X(x)\mathrm{d}x -\sqrt{n}  c\Big \rvert\\
    &=\Big \lvert \frac{1}{\sqrt{n}} \sum_{i=1}^n d_{\btau}(F_n(X_i))(c-t_0)-\sqrt{n}(c-t_0) \Big \rvert\\
    &= \sqrt{n}\lvert c-t_0 \rvert \cdot \Big \lvert \frac{1}{n} \sum_{i=1}^n d_{\btau}(F_n(X_i))-1 \Big \rvert.
\end{align*}
By the same reasoning as in the case of square loss, we obtain the equicontinuity of $\ell_{\delta}(x,c)$. This is under the assumption that $d_{\btau}$ is Lipschitz continuous.\\

\subsection{Loss function $\ell_{\delta}(x,c)=(1-c)^2 1_{\{x \geq 0\}}x-c^2 x 1_{\{x<0\}}$} \label{ThirdLossFunction}
This loss function corresponds  to $t_0=\frac{\mathbb{E}[X_{\sD_{\btau}} 1_{\{X_{\sD_{\btau}} \geq 0\}}]}{\mathbb{E}[ \lvert X_{\sD_{\btau}} \rvert]}$. Here, $\ell_{\delta}'(x,c)=2c \lvert x \rvert -2x 1_{\{x\geq 0\}}$. We have
\begin{align*}
    \sqrt{n} \lvert \lambda_{F_n}^*(c)-\lambda_{F_n}^*(t_0)-\lambda_F(c) \rvert &  =\Big \lvert \frac{1}{\sqrt{n}} \sum_{i=1}^n d_{\btau}(F_n(X_i))[2\lvert X_i \rvert (c-t_0)]\\ 
    &\qquad  -2 c \sqrt{n} \int_{-\infty}^{+\infty} \lvert x \rvert d_{\btau}(F_X(x))f_X(x)\mathrm{d} x -2\sqrt{n}  \int_0^{+\infty} x d_{\btau}(F_X(x))f_X(x)\mathrm{d}x  \Big \rvert\\
    & =\Big \lvert \frac{1}{\sqrt{n}} \sum_{i=1}^n d_{\btau}(F_n(X_i))[2\lvert X_i \rvert (c-t_0)]\\ 
    & \qquad -2 c \sqrt{n} \int_{-\infty}^{+\infty} \lvert x \rvert d_{\btau}(F_X(x))f_X(x)\mathrm{d}x -2\sqrt{n}  t_0 \int_{-\infty}^{+\infty} \lvert x \rvert d_{\btau}(F_X(x))f_X(x)\mathrm{d}x  \Big \rvert\\
    &=\Big \lvert \frac{1}{\sqrt{n}} \sum_{i=1}^n d_{\btau}(F_n(X_i))[2\lvert X_i \rvert (c-t_0)]-2 (c-t_0) \sqrt{n} \int_{-\infty}^{+\infty} \lvert x \rvert d_{\btau}(F_X(x))f_X(x)\mathrm{d}x  \Big \rvert\\
    & = 2\lvert c-t_0 \rvert \cdot  \Big \lvert \frac{1}{\sqrt{n}} \sum_{i=1}^n d_{\btau}(F_n(X_i))\lvert X_i \rvert - \sqrt{n} \int_{-\infty}^{+\infty} \lvert x \rvert d_{\btau}(F_X(x))f_X(x)\mathrm{d}x  \Big \rvert ,
\end{align*}
where we used the definition of $t_0$ in the second equality.
It is clear that 
\begin{align}
    & P \left ( \sup_{ \lvert c - t_0 \rvert \leq \omega_n} \sqrt{n} \lvert \lambda_{F_n}^*(c)-\lambda_{F_n}^*(t_0)-\lambda_F(c) \rvert \geq \epsilon \right ) \notag \\  
    &\leq P\left ( 2 \omega_n  \lvert \frac{1}{\sqrt{n}} \sum_{i=1}^n d_{\btau}(F_n(X_i))\lvert X_i \rvert - \sqrt{n} \int_{-\infty}^{+\infty} \lvert x \rvert d_{\btau}(F_X(x))f_X(x)\mathrm{d}x   \rvert \geq \epsilon \right ) \notag \\
    &\leq P \left ( 2\omega_n  \lvert  \frac{1}{\sqrt{n}} \sum_{i=1}^n (d_{\btau}(F_n(X_i))-d_{\btau}(F_X(X_i))\lvert X_i \rvert   \,  \rvert \geq \frac{\epsilon}{2}\right  ) \label{eq:equicontSpecialLossTerm1}\\
    &\quad +  P\left  ( 2\omega_n \sqrt{n} \lvert \frac{1}{n} \sum_{i=1}^n d_{\btau}(F_X(X_i)) \lvert X_i \rvert - \int_{-\infty}^{+\infty} \lvert x \rvert d_{\btau}(F_X(x))f_X(x)\mathrm{d}x  \,   \rvert \geq \frac{\epsilon}{2} \right ). \label{eq:equicontSpecialLossTerm2}
\end{align}

We consider the latter two terms separately. For the first term (\ref{eq:equicontSpecialLossTerm1}) we have under the assumption that $d_{\btau}$ is Lipschitz continuous
\begin{align*}
    & P\left ( 2\omega_n \Big \lvert  \frac{1}{\sqrt{n}} \sum_{i=1}^n (d_{\btau}(F_n(X_i))-d_{\btau}(F_X(X_i))\lvert X_i \rvert   \,  \rvert  \geq \frac{\epsilon}{2}\right )\\ 
    &\leq P\left ( 2\omega_n M \frac{1}{\sqrt{n}} \sum_{i=1}^n \lvert F_n(X_i)-F_X(X_i)\rvert \cdot \lvert X_i\rvert \geq \frac{\epsilon}{2} \right )\\
    &\leq P\left ( 2\omega_n M \sqrt{n} \sup_{x\in \mathbb{R}}\lvert F_n(x)-F_X(x)\rvert \cdot \frac{1}{n} \sum_{i=1}^n \lvert X_i\rvert \geq \frac{\epsilon}{2} \right ):=P(\lvert A_n \cdot B_n \rvert  \geq \frac{\epsilon}{2} ),
\end{align*}
where $A_n:=2\omega_n M \sqrt{n} \sup_{x\in \mathbb{R}}\lvert F_n(x)-F_X(x)\rvert $ and $B_n:=n^{-1}\sum_{i=1}^n \lvert X_i \rvert$.
We know by the weak law of large numbers that $B_n \xrightarrow[]{P} \mathbb{E}[\lvert X \rvert]$, if $\mathbb{E}[\lvert X \rvert]<\infty$. Hence, it suffices to show $A_n \xrightarrow[]{P} 0$.
Denote $F^*_n(x)=n^{-1} \sum_{i=1}^n 1_{\{X_i\leq x\}}$ and observe
\begin{equation}
    P(\lvert A_n \rvert \geq \epsilon)\leq P\left (2\omega_n M \sqrt{n} \sup_{x\in \mathbb{R}} \lvert F^*_n(x)-F_X(x)\rvert \geq \frac{\epsilon}{2}\right )+P\left (2\omega_n M \sqrt{n} \sup_{x\in \mathbb{R}} \lvert F^*_n(x)-F_n(x)\rvert \geq \frac{\epsilon}{2}\right ).
    \label{eq: TwoTerms}
\end{equation}
The first term goes to zero by Theorem A from p. 59 in \cite{Serfling1980}, i.e. we have
$$P\left (\sup_{x\in \mathbb{R}}\lvert 
    F^*_n(x)-F_X(x)\rvert\geq \frac{\epsilon}{2\omega_n M \sqrt{n}}\right )
    \leq C e^{-2n \cdot \frac{\epsilon^2}{4\omega_n^2 M^2 n}}. $$
For the second term, note that $\sup_{x\in \mathbb{R}} \lvert F^*_n((x)-F_n(x)\rvert=(n+1)^{-1}$, and therefore $2\omega_n M \sqrt{n} \sup_{x\in \mathbb{R}} \lvert F^*_n(x)-F_n(x)\rvert = 2 \omega_n M \sqrt{n}(n+1)^{-1}$ which tends to zero as $n \to \infty$. Hence, the second term in \eqref{eq: TwoTerms} is zero for $n$ sufficiently large. 
 In total we showed $P(\lvert A_n \rvert \geq \epsilon)\xrightarrow[]{}0$ as $n\to \infty$.\\

We now consider the second term (\ref{eq:equicontSpecialLossTerm2}). Denote 
$Y_n=n^{-1} \sum_{i=1}^n d_{\btau}(F_X(X_i))\lvert X_i \rvert$.
If $\text{Var}\{d_{\btau}(F_X(X))\lvert X \rvert \}<\infty$, then by the central limit theorem we obtain as $n\to \infty$
\begin{align*}
    &P\Big( 2\omega_n \sqrt{n} \Big \lvert \frac{1}{n} \sum_{i=1}^n d_{\btau}(F_X(X_i))\lvert X_i \rvert - \int_{-\infty}^{+\infty} \lvert x \rvert d_{\btau}(F_X(x))f_X(x)\mathrm{d}x  \Big \rvert \geq \frac{\epsilon}{2} \Big)\\
    &=P\Big( 2\omega_n \sqrt{n} \lvert Y_n - \mathbb{E}[Y_n]\lvert \geq \frac{\epsilon}{2} \Big) =P\Big( \sqrt{n}\lvert Y_n - \mathbb{E}[Y_n]\lvert \geq \frac{\epsilon}{4\omega_n} \Big) \rightarrow 0.
\end{align*}
Hence, also the second -term (\ref{eq:equicontSpecialLossTerm2}) goes to zero as $n\to \infty$.\\

In summary, we proved that for every sequence $\omega_n \to 0$
$$\sup_{\lvert c - t_0 \rvert \leq \omega_n}\sqrt{n} \lvert \lambda_{F_n}^*(c)-\lambda_{F_n}^*(t_0)-\lambda_F(c) \rvert \xrightarrow{P} 0,$$
and hence the equicontinuity Assumption \textbf{(A10)} is satisfied for this loss function, under the condition that $d_{\btau}$ is Lipschitz continuous with  $\mathbb{E}[\lvert X \rvert]<\infty$ and $\text{Var}\{d_{\btau}(F_X(X))\lvert X \rvert \}<\infty$.

\subsection{Loss function $\ell_{\delta}(x,c)=\lvert x \rvert (c-x)^2$}
For this loss function the generalized extremile equals $t_0=\frac{\mathbb{E}[X_{\sD_{\btau}}^2]}{\mathbb{E}[\lvert X_{\sD_{\btau}} \rvert]}$. Here, $\ell_{\delta}'(x,c)=2\lvert x \rvert (c-x)$.  We have  
\begin{align*}
    &\sqrt{n} \lvert \lambda_{F_n}^*(c)-\lambda_{F_n}^*(t_0)-\lambda_F(c) \rvert\\
    & \qquad  =\Big \lvert \frac{1}{\sqrt{n}} \sum_{i=1}^n d_{\btau}F_n(X_i))2\lvert X_i\rvert (c-t_0)\\ 
    & \qquad  \qquad -  2c\sqrt{n} \int_{-\infty}^{+\infty} \lvert x \rvert d_{\btau}(F_X(x))f_X(x)\mathrm{d}x + 2\sqrt{n} \int_{-\infty}^{+\infty}x^2 d_{\btau}(F_X(x))f_X(x)\mathrm{d}x\Big \rvert\\
    &\qquad  =\Big \lvert \frac{1}{\sqrt{n}} \sum_{i=1}^n d_{\btau}(F_n(X_i))2\lvert X_i\rvert (c-t_0)-  2c\sqrt{n} \int_{-\infty}^{+\infty} \lvert x \rvert d_{\btau}(F_X(x))f_X(x)\mathrm{d}x + 2\sqrt{n} t_0 \Big \rvert\\
    &\qquad  =2 \lvert c-t_0\rvert \cdot \Big \lvert \frac{1}{\sqrt{n}} \sum_{i=1}^n d_{\btau}(F_n(X_i)) \lvert X_i \rvert - \sqrt{n} \int_{-\infty}^{+\infty} \lvert x \rvert d_{\btau}(F_X(x))f_X(x)\mathrm{d}x  \Big \rvert.
\end{align*}
For the second equality, we used the definition of $t_0$. The result follows by the same reasoning as in the case of the loss function $\ell_{\delta}(x,c)=(1-c)^2 1_{\{x \geq 0\}}x-c^2 x 1_{\{x<0\}}$ in Section \ref{ThirdLossFunction}. Hence, we obtain equicontinuity under the assumption that $d_{\btau}$ is Lipschitz continuous;  and the assumptions that $\mathbb{E}[\lvert X \rvert]<\infty$  and $\text{Var}\{d_{\btau}(F_X(X))\lvert X \rvert \}<\infty$.\\

\subsection{Esscher loss function (entry 6 in Table \ref{LossFunctions1A})}
Consider the Esscher loss $\ell_{\delta}(x,c)=(c-x)^2 e^{\delta x}$, corresponding to $t_0=\frac{\mathbb{E}[X_{\sD_{\btau}} e^{\delta X_{\sD_{\btau}}}]}{\mathbb{E}[e^{\delta X_{\sD_{\btau}}}]}$. Here, $\ell_{\delta}'(x,c)=2(c-x)e^{\delta x}$. We have 
\begin{align*}
    &\sqrt{n} \lvert \lambda_{F_n}^*(c)-\lambda_{F_n}^*(t_0)-\lambda_F(c) \rvert \\
    & \qquad =\Big \lvert \frac{1}{\sqrt{n}} \sum_{i=1}^n d_{\btau}(F_n(X_i))2 e^{\delta X_i} (c-t_0)\\ 
    &\qquad \qquad -  2c\sqrt{n} \int_{-\infty}^{+\infty} e^{\delta x} d_{\btau}(F_X(x))f_X(x)\mathrm{d}x + 2\sqrt{n} \int_{-\infty}^{+\infty}x e^{\delta x}d_{\btau}(F_X(x))f_X(x)\mathrm{d}x\Big \rvert\\
    &\qquad  =\Big \lvert \frac{1}{\sqrt{n}} \sum_{i=1}^n d_{\btau}(F_n(X_i))2 e^{\delta X_i} (c-t_0)-  2c\sqrt{n} \int_{-\infty}^{+\infty} e^{\delta x}d_{\btau}(F_X(x))f_X(x)\mathrm{d}x + 2\sqrt{n} t_0 \Big \rvert\\
    &\qquad  =2 \lvert c-t_0\rvert \cdot \Big \lvert \frac{1}{\sqrt{n}} \sum_{i=1}^n d_{\btau}(F_n(X_i)) e^{\delta X_i} - \sqrt{n} \int_{-\infty}^{+\infty} e^{\delta x} d_{\btau}(F_X(x))f_X(x)\mathrm{d}x  \Big \rvert.
\end{align*}
We again used the definition of $t_0$. The result now follows by a similar reasoning as in the case of the loss function $\ell_{\delta}(x,c)=(1-c)^2 1_{\{x \geq 0\}}x-c^2 x 1_{\{x<0\}}$ in Section \ref{ThirdLossFunction}. We obtain equicontinuity under the assump\-tion that $d_{\btau}$ is Lipschitz continuous; and that $ \mathbb{E}[\lvert e^{\delta X} \rvert]<\infty$ and $\text{Var}\{d_{\btau}(F_X(X))e^{\delta X} \}<\infty$.\\

\subsection{Loss function $\ell_{\delta}(x,c)=\frac{1}{\delta} e^{\delta c}-e^{-\delta x }c$}
The generalized extremile associated to this loss function is $t_0=\frac{1}{\delta} \log(\mathbb{E}[e^{-\delta X_{\sD_{\btau}}}])$. Here, $\ell_{\delta}'(x,c)=e^{\delta c}-e^{-\delta x}$. We have 
\begin{align*}
    &\sqrt{n} \lvert \lambda_{F_n}^*(c)-\lambda_{F_n}^*(t_0)-\lambda_F(c) \rvert =\Big \lvert \frac{1}{\sqrt{n}} \sum_{i=1}^n d_{\btau}(F_n(X_i)) (e^{\delta c}-e^{\delta t_0})\\ &\hspace{1cm}-  \sqrt{n}e^{\delta c}+ \sqrt{n} \int_{-\infty}^{+\infty} e^{-\delta x} d_{\btau}(F_X(x))f_X(x)\mathrm{d}x \Big \rvert\\
    &=\Big \lvert \frac{1}{\sqrt{n}} \sum_{i=1}^n d_{\btau}(F_n(X_i)) (e^{\delta c}-e^{\delta t_0})-  \sqrt{n}e^{\delta c}+ \sqrt{n} e^{t_0 \delta} \Big \rvert \\
    &=\lvert e^{\delta c}-e^{\delta t_0} \rvert \cdot \Big \lvert \frac{1}{\sqrt{n}} \sum_{i=1}^n d_{\btau}(F_n(X_i)) -  \sqrt{n} \Big \rvert.
\end{align*}
Note that $\sup_{\lvert c-t_0 \rvert \leq \omega_n } \lvert e^{\delta c}-e^{\delta t_0} \rvert \leq \sup_{\lvert c-t_0 \rvert \leq \omega_n }  K \lvert \delta c -\delta t_0 \rvert \leq K \delta \omega_n$, for some constant $K$. Indeed, this is because any continuously differentiable function is locally Lipschitz. The result follows by a similar reasoning as in the case of square loss.

\setcounter{equation}{0}

\section{Section \ref{sec:estimation--ApplicationsANResult}: further explanation regarding Remark \ref{remark:expectileVariance}}
\label{supsec: ANExpectiles}
When $D_{\btau}$ is the cumulative distribution function corresponding to a uniform random variable, the variance formula given in Corollary \ref{cor:expectilesAsymptoticVariance} reduces to (\ref{eq:expectilesRemarkVariance}). Recall when $D_{\btau}(u)=u$ for $u\in [0,1]$ then $d_{\btau}(u)=1$ for $u \in [0,1]$ and $X_{\sD_{\btau}}=X$. Also $D_{\btau}\left ( F_X(t_0)\right)=F_X(t_0)$. Hence the denominator in Corollary \ref{cor:expectilesAsymptoticVariance} reduces to 
\[
\left ( - 4 \delta F_X(t_0) + 2 \delta + 2 F_X(t_0) \right )^2 
= 4 \left [ \delta \left ( 1- F_X(t_0)\right ) + (1-\delta) F_X(t_0)   \right ]^2.
\]
It remains to look at the  the numerator of (\ref{eq:expectilesRemarkVariance}).
Note that $$\mathbb{E}\left [ \left (  I_{\delta}(t_0,X)\right )^2\right ]=\text{Var}(I_{\delta}(t_0,X))+\left (\mathbb{E}[I_{\delta}(t_0,X)]\right )^2.$$
By Example 4 on p. 43 in \cite{BelliniEtAl2014} we know that the second term is equal to zero. Observe, 
\begin{align}
    \text{Var}(I_{\delta}(t_0,X))&=\delta^2 \text{Var}((X-t_0)1_{\{X \geq t_0\}})+(1-\delta)^2 \text{Var}((X-t_0) 1_{\{X<t_0\}} \label{eq:remarkVarianceExpectiles}\\ &\text{\quad}-2\delta(1-\delta)\text{cov}((X-t_0)1_{\{X \geq t_0\}},(X-t_0) 1_{\{X<t_0\}})\nonumber.
\end{align}
In our calculations, we will make use of two equality's (see pp. 116--117 in \cite{Shorack2000})
$$\text{Cov}(X,Y)=\int_{-\infty}^{+\infty} \int_{-\infty}^{+\infty} [F_{X,Y}(x,y)-F_X(x)F_Y(y)]\mathrm{d}x \mathrm{d}y $$
and
$$\text{Var}(X)=\int_{-\infty}^{+\infty} \int_{-\infty}^{+\infty} [F_{X}(\min\{x,y\})-F_X(x)F_X(y)]\mathrm{d}x \mathrm{d}y.$$
Note that $F_{(X-t_0)1_{\{X \geq t_0\}}}(u)=F_{X}(u+t_0)1_{\{u>0\}}$. Hence, we have
\begin{align*}
   \text{Var}((X-t_0)1_{\{X \geq t_0\}})&=\int_{-\infty}^{+\infty} \int_{-\infty}^{+\infty} [F_{X}(\min\{x,y\}+t_0)1_{\{\min\{x,y\}>0\}}\\
    &\hspace{2.5cm}-F_{X}(x+t_0)F_{X}(y+t_0)1_{\{x>0\}}1_{\{y>0\}}]\mathrm{d}x\mathrm{d}y\\
    &=\int_{-\infty}^{+\infty} \int_{-\infty}^{+\infty} [F_{X}(\min\{x-t_0,y-t_0\}+t_0) 1_{\{\min\{x-t_0,y-t_0\}>0\}}\\ &\hspace{2.5cm}-F_{X}(x)F_{X}(y)1_{\{x>t_0\}}1_{\{y>t_0\}}]\mathrm{d}x\mathrm{d}y\\
    &=\int_{t_0}^{+\infty} \int_{t_0}^{+\infty} [F_{X}(\min\{x,y\})-F_{X}(x)F_{X}(y)]\mathrm{d}x\mathrm{d}y.
\end{align*}
For the second equality we used a change of variables.\\

Similarly, by using $F_{(X-t_0)1_{\{X<t_0\}}}(u)=F_X(u+t_0)1_{\{u<0\}}+1_{\{u>0\}}$ we obtain
\begin{align*}
    \text{Var}((X-t_0)1_{\{X < t_0\}})&=\int_{-\infty}^{+\infty} \int_{-\infty}^{+\infty} \Big[ F_X(\min\{x,y\}+t_0)1_{\{\min\{x,y\}<0 \}}+1_{\{\min\{x,y\}>0 \}}\\ &\hspace{1cm} -(F_X(x+t_0) 1_{\{x<0\}}+1_{\{x>0\}})(F_X(y+t_0) 1_{\{y<0\}}+1_{\{y>0\}})\Big] \mathrm{d}x\mathrm{d}y\\
    &=\int_{-\infty}^{+\infty} \int_{-\infty}^{+\infty} \Big[ 1_{\{\min\{x,y\}>0\}}-1_{\{x>0\}}1_{\{y>0\}}\Big] \mathrm{d}x\mathrm{d}y\\ &\hspace{1cm}+ \int_{-\infty}^{+\infty} \int_{-\infty}^{+\infty} \Big[ F_X(\min\{x,y\}+t_0)1_{\{\min\{x,y\}<0 \}}\\
    &\hspace{2cm}-F_X(x+t_0)1_{\{x<0\}}F_X(y+t_0)1_{\{y<0\}}\Big] \mathrm{d}x\mathrm{d}y\\
    &\hspace{1cm}-\int_{-\infty}^{+\infty} \int_{-\infty}^{+\infty} F_X(x+t_0)1_{\{x<0\}}1_{\{y>0\}}\mathrm{d}x\mathrm{d}y\\ &\hspace{1cm}-\int_{-\infty}^{+\infty} \int_{-\infty}^{+\infty} F_X(y+t_0)1_{\{x>0\}}1_{\{y<0\}}\mathrm{d}x\mathrm{d}y\\
    &=\int_{-\infty}^{0} \int_{-\infty}^{0} \Big[F_X(\min\{x,y\}+t_0)-F_X(x+t_0)F_X(y+t_0)\Big]\mathrm{d}x\mathrm{d}y\\
    &=\int_{-\infty}^{t_0} \int_{-\infty}^{t_0} \Big[ F_X(\min\{x,y\})-F_X(x)F_X(y)\Big]\mathrm{d}x\mathrm{d}y.
\end{align*}
Most parts of the long expression in the second equality cancel out. The first term of the expression is equal to zero. The second term can be split into three parts. Two of those parts cancel out with the third and fourth term. The remaining part is given in the third equality. For the last equality, we again made use of a change of variables.\\

Observe that 
$$F_{(X-t_0)1_{\{X \geq t_0\}},(X-t_0)1_{\{X <t_0\}}}(x,y)=\begin{cases}
    0 &\hspace{0.5cm} \text{ if } x<0\\
    F_X(x+t_0) &\hspace{0.5cm} \text{ if } x,y>0\\
    F_X(y+t_0) &\hspace{0.5cm} \text{ if } x>0 \text{ and } y<0.
\end{cases}$$
Because of this, we obtain
\begin{align*}
    &\text{Cov}((X-t_0)1_{\{X \geq t_0\}},(X-t_0) 1_{\{X<t_0\}})\\
    &=\int_{0}^{+\infty} \int_{0}^{+\infty} F_X(x+t_0)\mathrm{d}x\mathrm{d}y + \int_{-\infty}^{0} \int_{0}^{+\infty} F_X(y+t_0)\mathrm{d}x\mathrm{d}y\\ &\hspace{2cm}-\int_{-\infty}^{+\infty} \int_{-\infty}^{+\infty} \Big[ F_X(x+t_0)1_{\{x>0\}}(F_X(y+t_0)1_{\{y<0\}}+1_{\{y>0\}})\Big] \mathrm{d}x\mathrm{d}y\\
    &=\int_{0}^{+\infty} \int_{0}^{+\infty} F_X(x+t_0)\mathrm{d}x\mathrm{d}y + \int_{-\infty}^{0} \int_{0}^{+\infty} F_X(y+t_0)\mathrm{d}x\mathrm{d}y\\ &\hspace{2cm}-\Big(\int_{-\infty}^{0} \int_{0}^{+\infty} 
    F_X(x+t_0)F_X(y+t_0)\mathrm{d}x\mathrm{d}y + \int_{0}^{+\infty} \int_{0}^{+\infty}
    F_X(x+t_0)\mathrm{d}x\mathrm{d}y \Big)\\
    &=\int_{-\infty}^{t_0} \int_{t_0}^{+\infty} [ F_X(y)-F_X(x)F_X(y)]\mathrm{d}x\mathrm{d}y.
\end{align*}
For the last equality, we again made use of a change of variables.\\

Plugging all parts into equation (\ref{eq:remarkVarianceExpectiles}) proves the stated equality.

\section{Simulation studies: additional table and figure}\label{supsec: Simulations}
\subsubsection*{Simulation study 1}
Table \ref{table:MSEtableNormalAndExp} summarizes the simulation results by providing  estimated bias, variance and MSE  of the estimator $T_n^*$. The results for samples drawn from a standard normal distribution are displayed in the first three rows of Table \ref{table:MSEtableNormalAndExp}; and these for samples drawn from a standard exponential distribution in the last three rows.
Note that the mean squared error increases as $\tau$ increases. The largest part of the MSE seems to be due to the variance, at least for the larger values of $\tau$, for which performance is also a bit worse for the exponential distribution. In Figure \ref{FigureAVarNormalVsExponential} the asymptotic variance of the estimator $T_n^*$ is depicted for each of the two underlying distributions, revealing indeed that for values of $\tau$ larger than (approximately) 0.30, the asymptotic variances for the exponential distribution are substantially larger than these for a normal distribution. 
\begin{table}[h]
\caption{Simulation study 1. Mean squared error (MSE), bias and variance of the estimator $T_n^*$. For accompanying results, see Figure \ref{fig:densityEstimates}.}
\label{table:MSEtableNormalAndExp}
\vspace*{0.2 cm}

\noindent
\resizebox{\textwidth}{!}{%
\begin{tabular}{l|lr|lr|lr}
                       & \multicolumn{2}{c|}{$\tau=0.1$} & \multicolumn{2}{c|}{$\tau=0.9$} & \multicolumn{2}{c}{$\tau=0.95$} \\  \cline{2-7} 
\multicolumn{1}{r|}{} & \multicolumn{1}{c}{$n=50$} & \multicolumn{1}{c|}{$n=800$} & \multicolumn{1}{c}{$n=50$} & \multicolumn{1}{c|}{$n=800$} & \multicolumn{1}{c}{$n=50$} & \multicolumn{1}{c}{$n=800$} \\ [0.2cm]
& \multicolumn{6}{c}{$\mathcal{N}(0;1)$} \\ \hline                 
MSE                    & $2.78 \cdot 10^{-2}$                & $1.71 \cdot 10^{-3} $                & $9.66\cdot 10^{-2}$                & $6.41 \cdot 10^{-3}$                & $1.59 \cdot 10^{-1}$                & $1.08 \cdot 10^{-2}$                  \\
Bias                   & $3.53 \cdot 10^{-2}$                 & $1.49 \cdot 10^{-3}$                & $-5.54 \cdot 10^{-2}$                & $2.30 \cdot 10^{-4}$              & $-1.31 \cdot 10^{-1}$                & $-6.43 \cdot 10^{-3}$               \\
Variance               &  $2.65 \cdot 10^{-2}$                & $1.70 \cdot 10^{-3}$                 &   $9.35 \cdot 10^{-2}$                &  $6.41 \cdot 10^{-3}$               & $1.41 \cdot 10^{-1}$                 & $1.08 \cdot 10^{-2}$           \\[0.2cm]
\hline \hline 
\multicolumn{7}{c}{}\\[0.02cm]
 & \multicolumn{6}{c}{$\mbox{Expo}(1)$} \\ \cline{2-7}         
MSE                    &  $4.77 \cdot 10^{-3}$                 & $3.03 \cdot 10^{-4}$                  & $6.89 \cdot 10^{-1}$                 & $4.42 \cdot 10^{-2}$                    & $1.23$                   & $8.71 \cdot 10^{-2}$                    \\
Bias                   &  $1.79 \cdot 10^{-2}$                  & $4.08 \cdot 10^{-4}$                  & $-1.63 \cdot 10^{-1}$                 & $-2.47 \cdot 10^{-2}$                   &  $-3.81 \cdot 10^{-1}$                  &  $-4.34 \cdot 10^{-2}$                  \\
Variance               &  $4.45 \cdot 10^{-3}$                 & $3.02 \cdot 10^{-4}$                  & $6.62 \cdot 10^{-1}$                   & $4.36 \cdot 10^{-2}$                    & $1.07 $                  &  $8.52 \cdot 10^{-2}$                  
\end{tabular}%
}
\end{table}
\vspace*{-0.8 cm}

\noindent
\begin{figure}[h!]
    \centering
\includegraphics[width=0.52\textwidth]{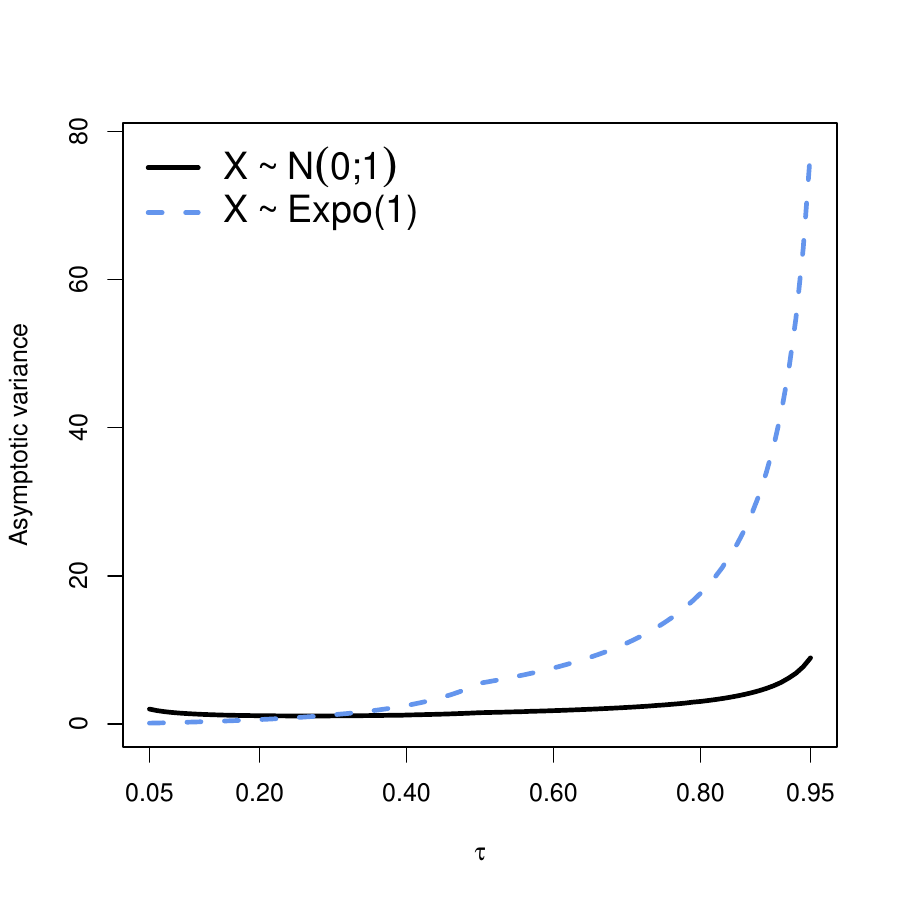}
\vspace*{-0.86 cm}

\noindent
\caption{Simulation study. The asymptotic variances of the estimator $T_n^*$, for an underlying normal distribution  and and an exponential distribution, as a function of $\tau$.}
\label{FigureAVarNormalVsExponential}. 

\end{figure}
\begin{figure}[htb]
    \centering
    \includegraphics[width=0.42\textwidth]{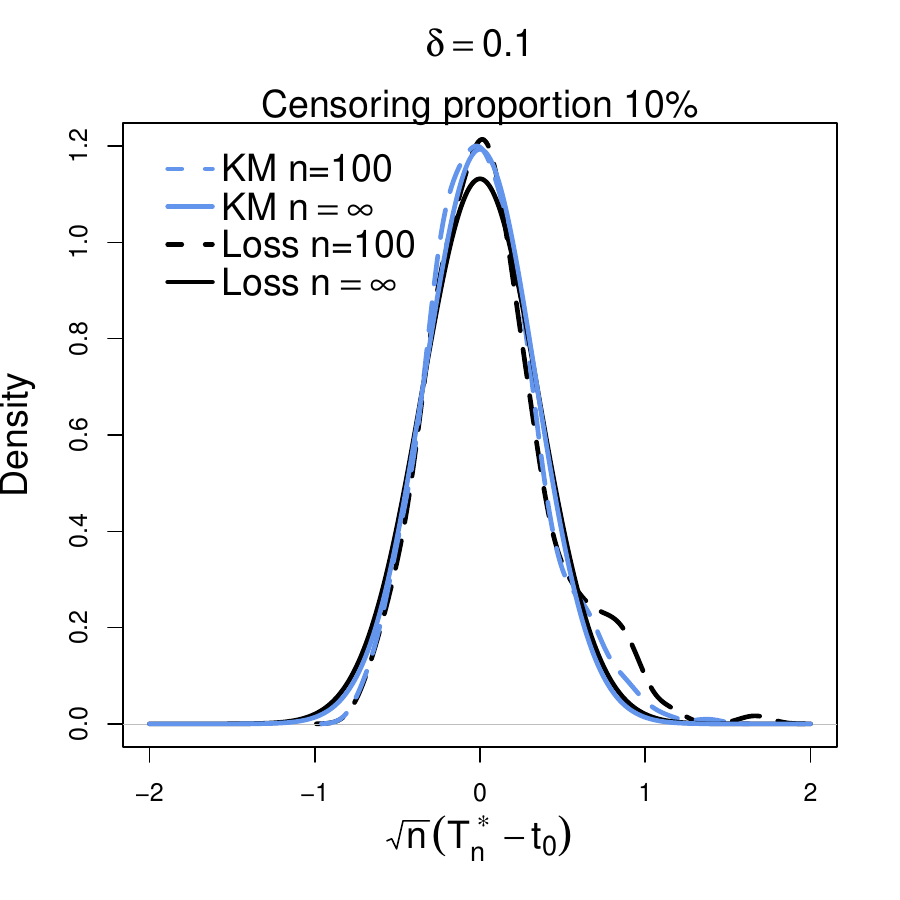}
    \includegraphics[width=0.42\textwidth]{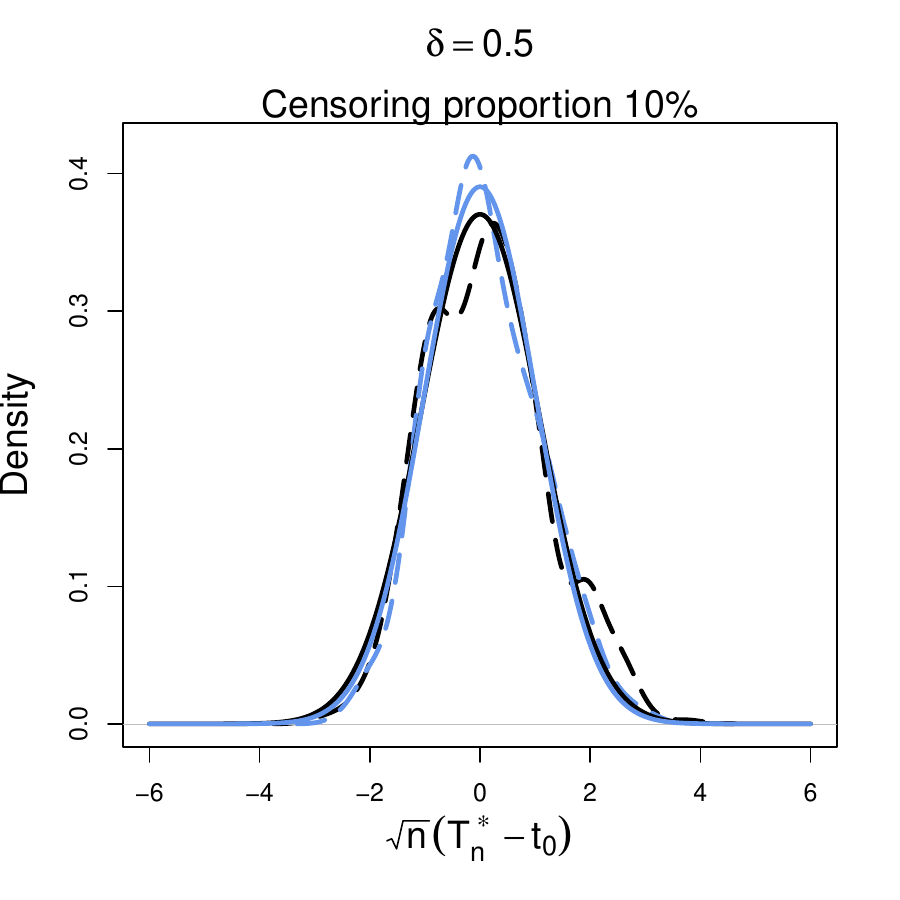}\\
        \includegraphics[width=0.42\textwidth]{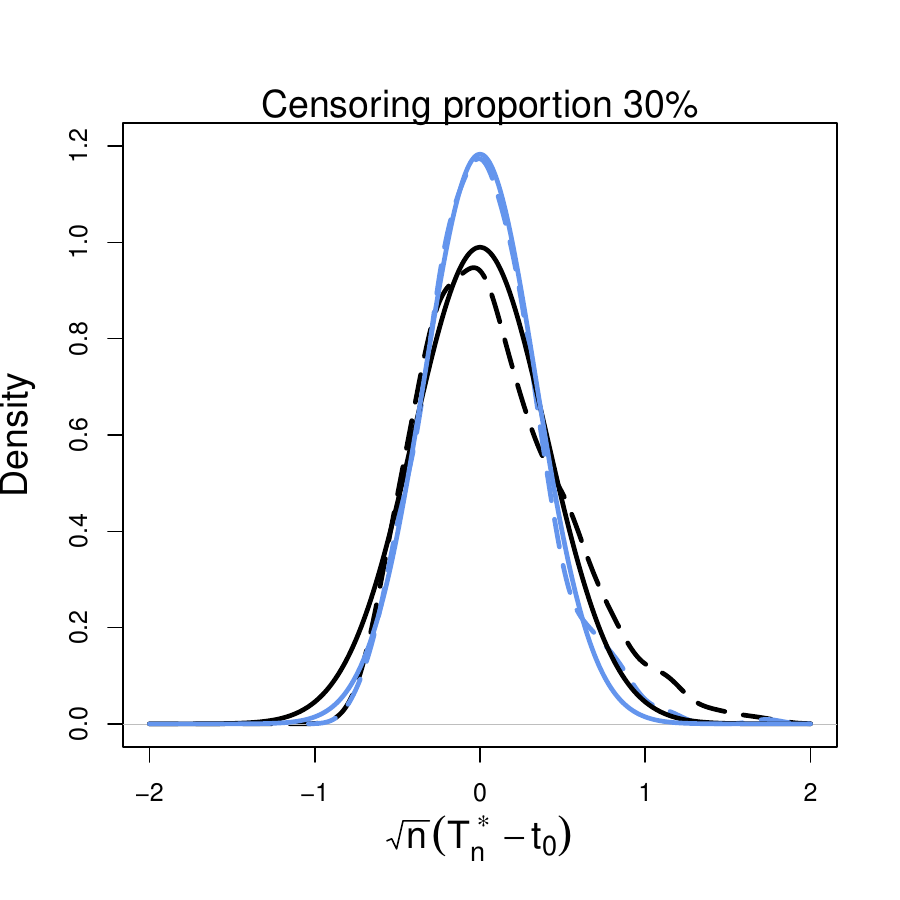}
   \includegraphics[width=0.42\textwidth]{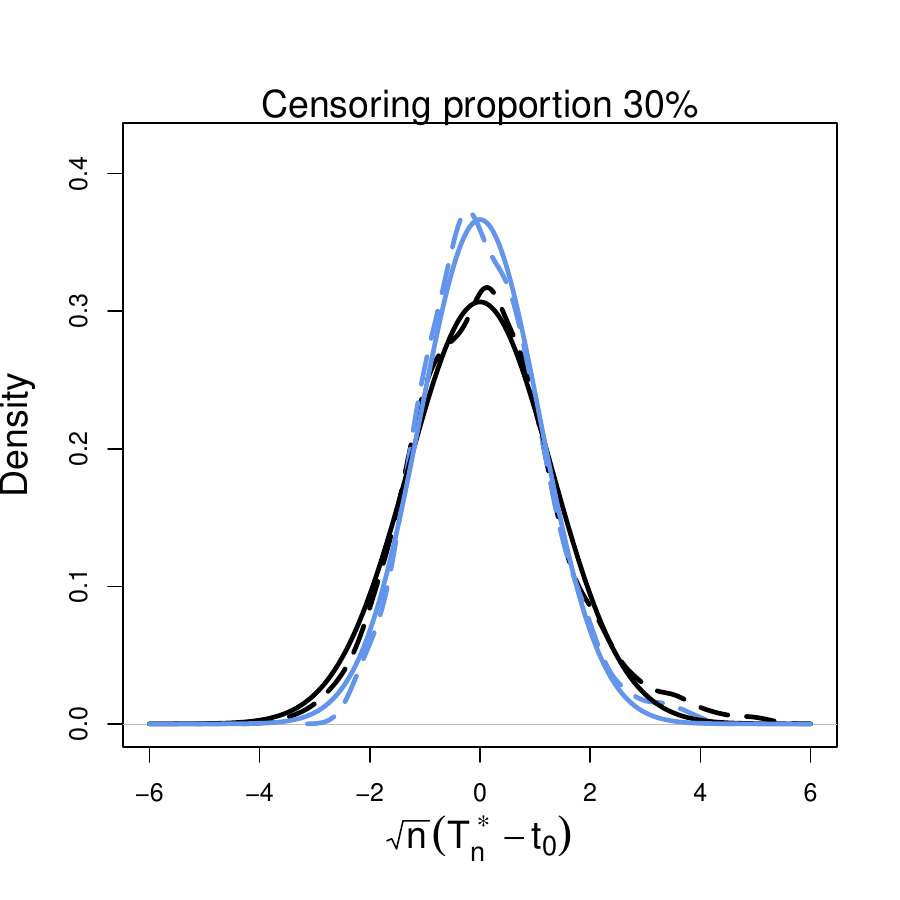}\\       
            \includegraphics[width=0.42\textwidth]{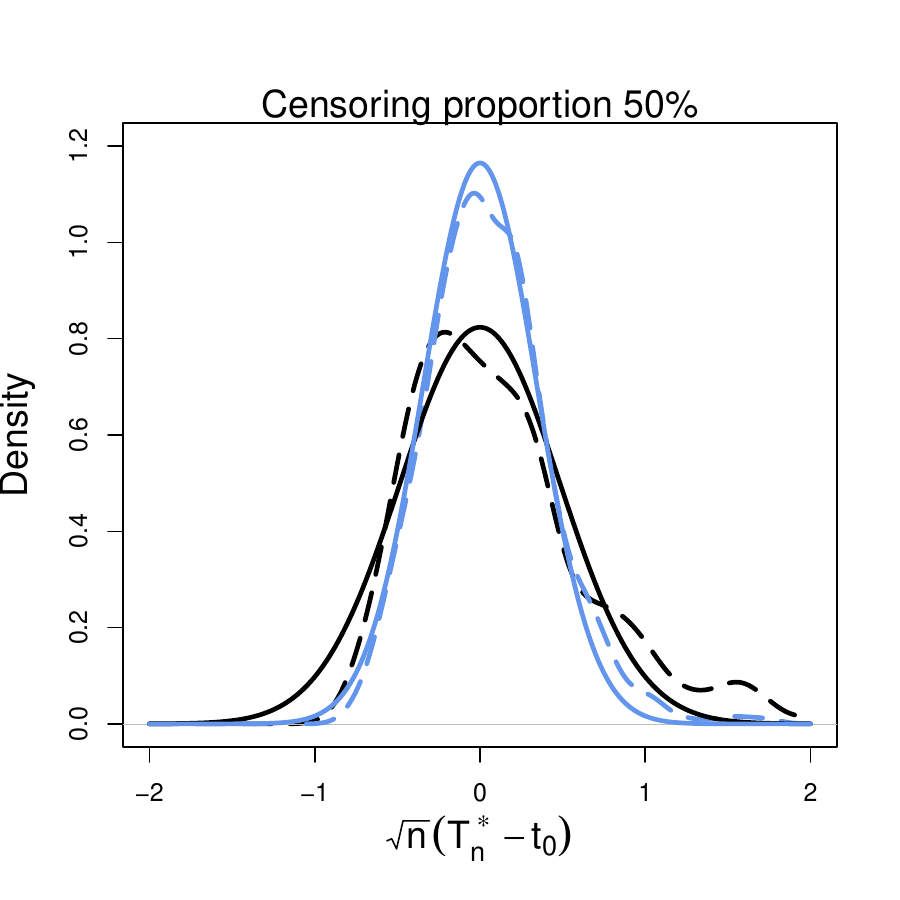}
            \includegraphics[width=0.42\textwidth]{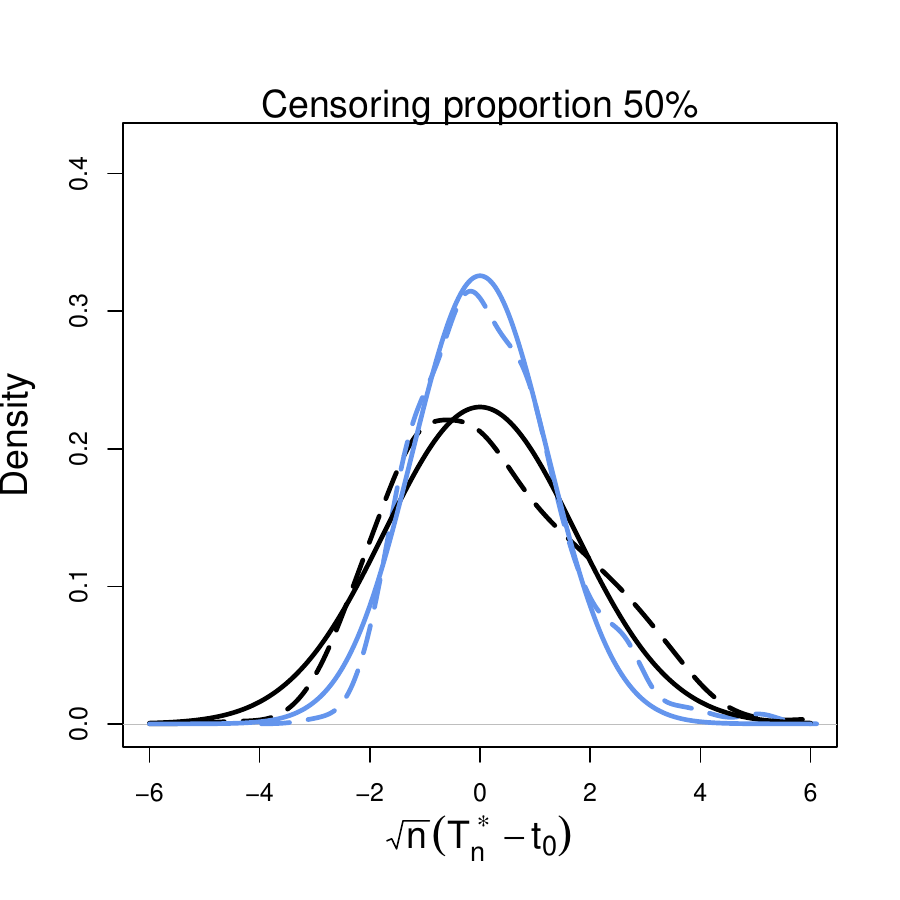}
    \vspace*{-0.44 cm}
    
    \noindent
    \caption{Simulation study 4. Finite-sample distribution of the (loss-based) estimator $T_n^*$ and of the KM-based estimator.
    Density estimates of 500 values of $\sqrt{n}(T_n^*-t_0)$ for $n=100$ in dashed black lines. The blue dashed lines are for the KM-based estimator. The solid black line is the normal density with variance equal to   (\ref{eq:asymptoticVarianceMyCensoring}). The solid blue line represents the asymptotic distribution of the KM-based estimator.}
    \label{fig:adaptiveQuantileLossDensityPlots}
\end{figure}
\subsubsection*{Simulation study 4}
Figure \ref{fig:adaptiveQuantileLossDensityPlots} displays in dashed black lines the density estimates (based on 500 values) of our loss-based estimator $\sqrt{n}(T_n^*-t_0)$ for $n=100$. The solid black line is the normal density with variance equal to (\ref{eq:asymptoticVarianceMyCensoring}).
The dashed line and solid blue line represent the finite-sample distribution and the asymptotic distribution of the KM-based estimator.
The left (respectively right) panels are results for estimating the  $0.1$th
(respectively $\delta=0.5$)  quantile. The different rows display results for different censoring proportions.

Since overall the dashed black lines are quite close to the solid black lines, it appears that our loss-based estimator is asymptotically normal with variance as in (\ref{eq:asymptoticVarianceMyCensoring}). For the different censoring proportions, the asymptotic variance of both estimators is larger when $\delta$ is $0.5$. The KM-based estimator shows a smaller (finite-sample) variance compared to the loss-based estimator. As the censoring proportion increases, the (finite-sample) distributions of the two different estimators become more distinct.  As to be expected, there is an increase in variance for both estimators when the censoring proportion increases.

\FloatBarrier
\section{Real data example: additional plots}
\label{supsec: RealDataExample}

\begin{figure}[htb]
    \centering
\includegraphics[width=0.48\textwidth]{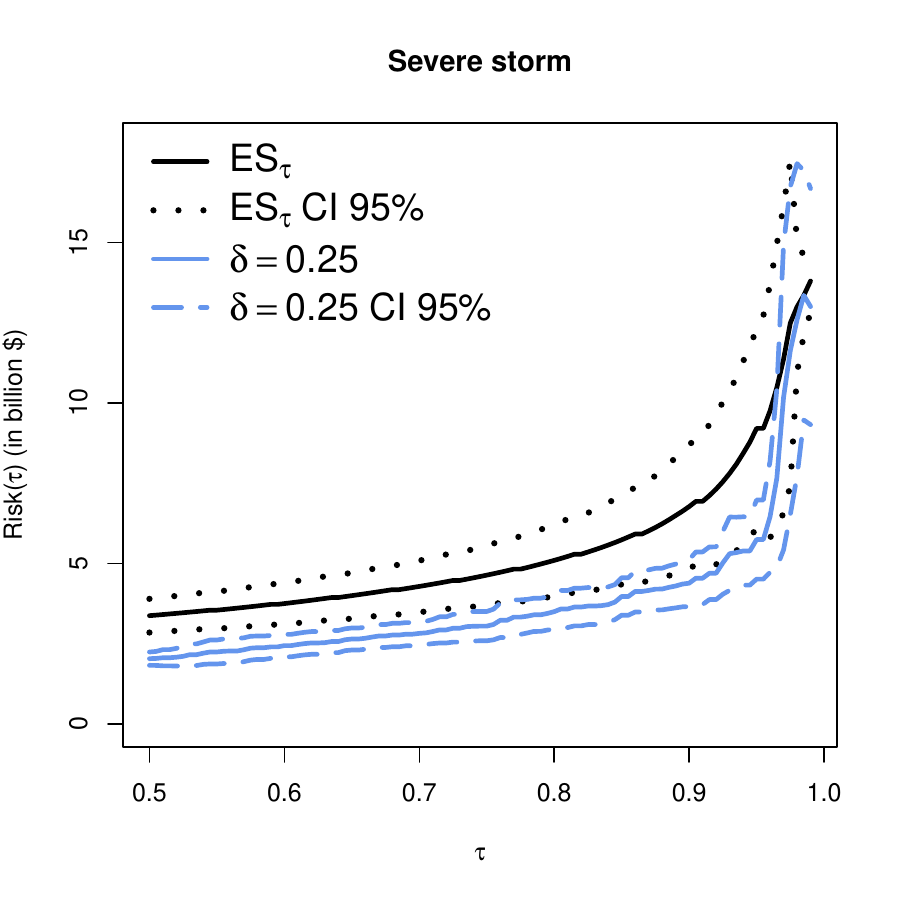}
\includegraphics[width=0.48\textwidth]{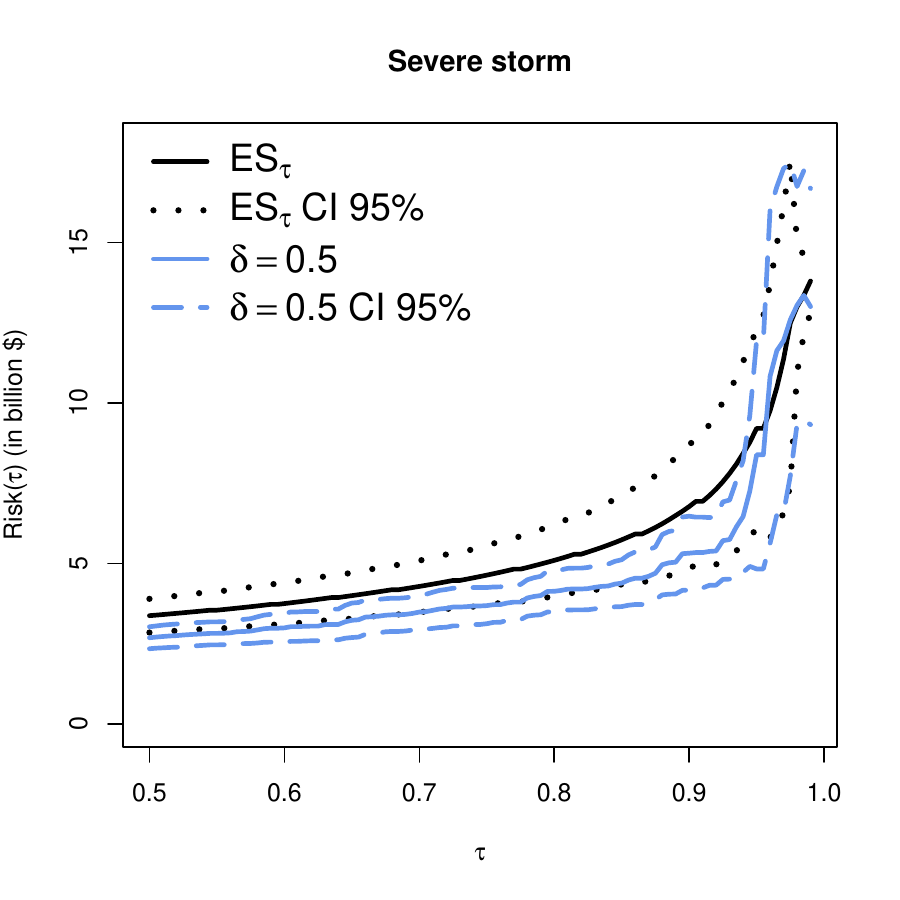}
\vspace*{-0.68 cm}

\noindent
    \caption{US disaster data. Estimates of severe storms, as in the top right panel of Figure \ref{fig:histogramAndQuantilesESV2}, together with 95\% asymptotic pointwise confidence intervals for the estimated $0.25$th and $0.50$h  quantile curves. The asymptotic pointwise confidence intervals are constructed using (\ref{eq:quantileLossEstimatorAsymptoticVariance}). }
    \label{fig:SevereStormsCISTwoOtherQuantiles}
\end{figure}
\begin{figure}[htb]
    \centering
    \includegraphics[width=0.44\textwidth]{tropicalCycloneQuantilesES.pdf}
        \includegraphics[width=0.44\textwidth]{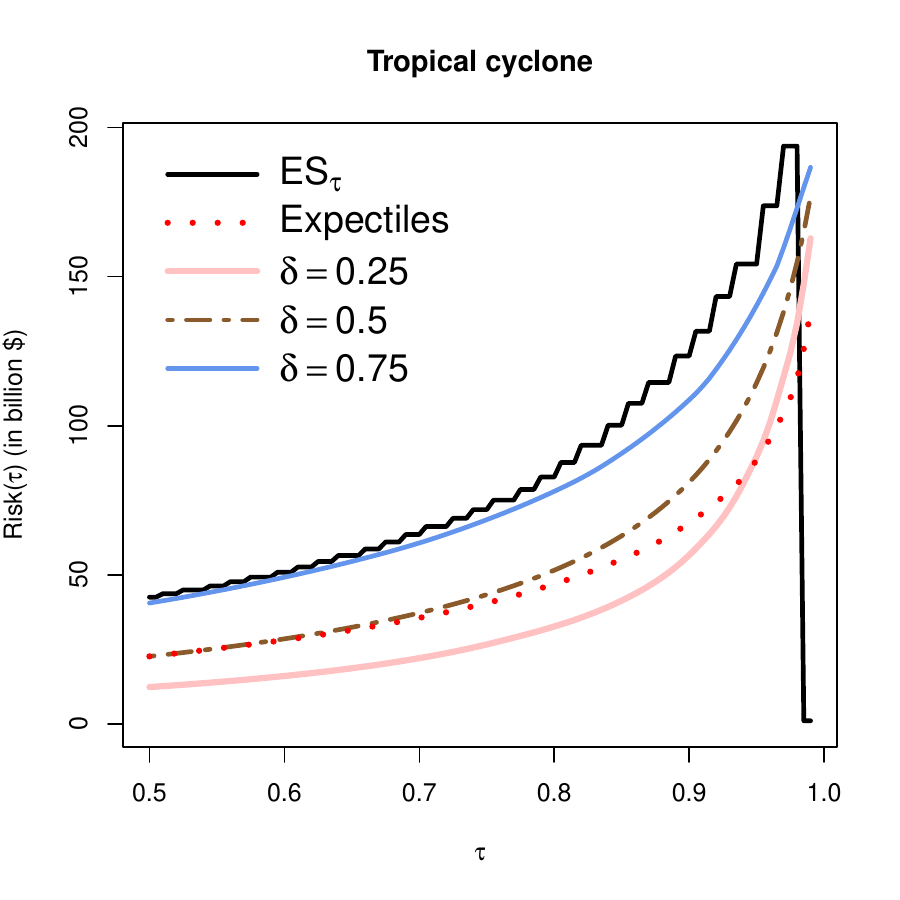}\\[0.02 ex]
 \includegraphics[width=0.44\textwidth]{floodingQuantilesES.pdf}
        \includegraphics[width=0.44\textwidth]{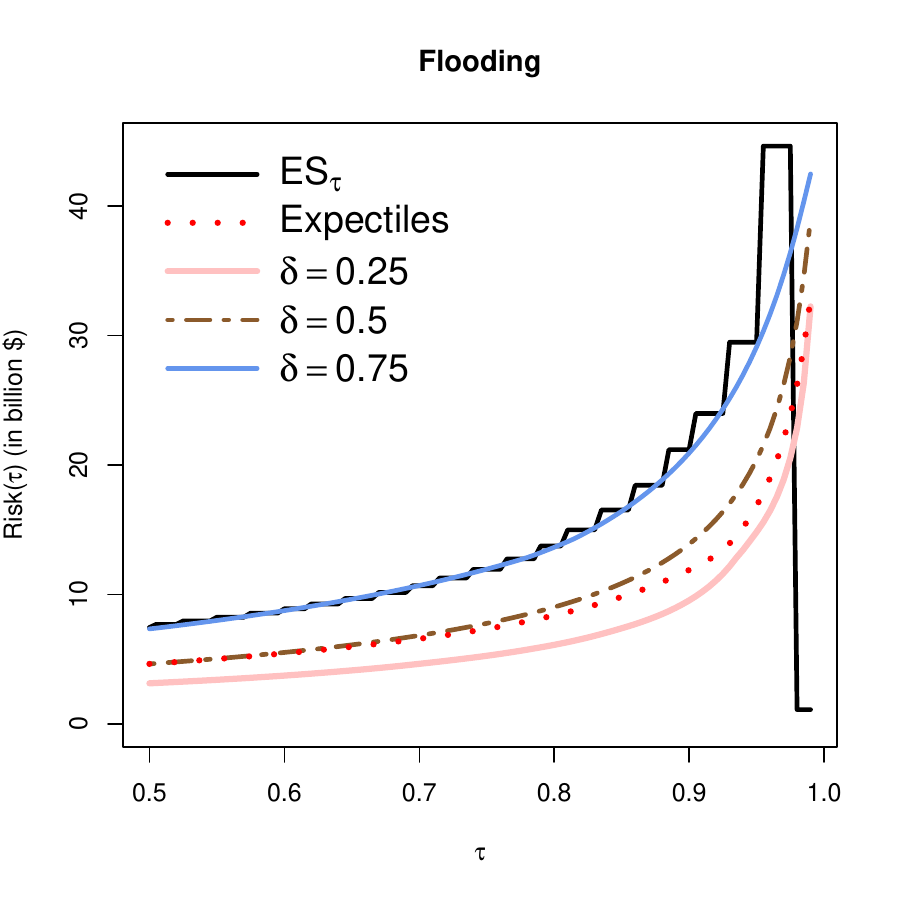}\\ [0.02 ex]
         \includegraphics[width=0.44\textwidth]{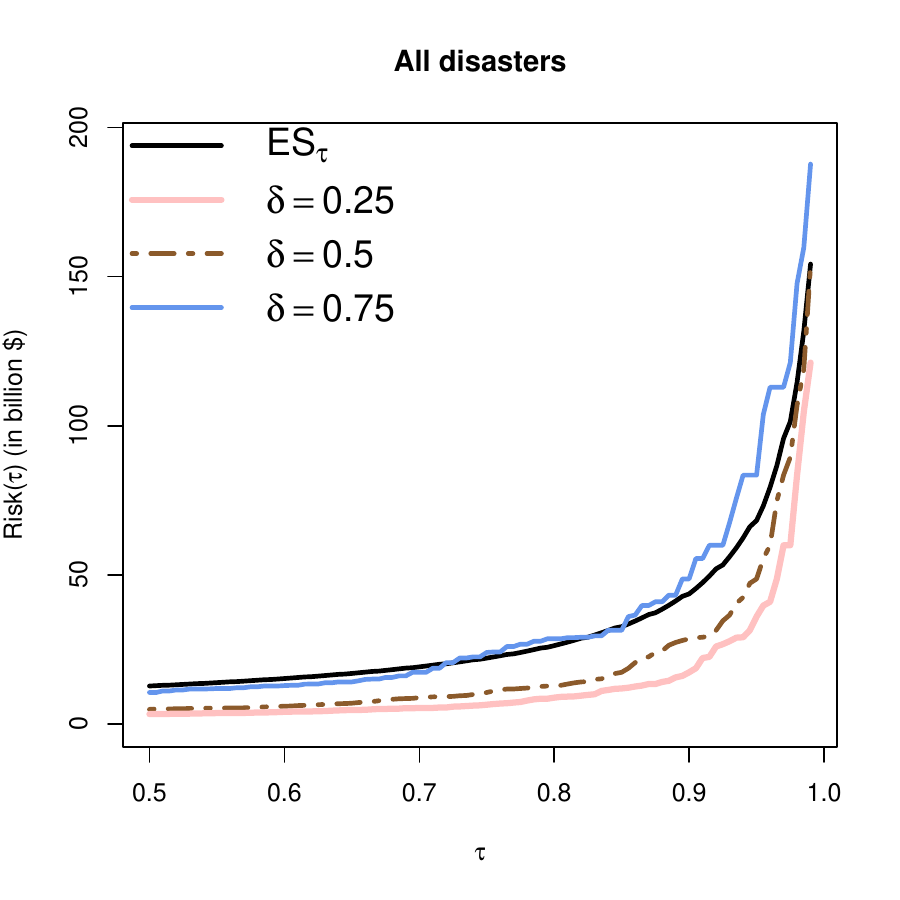}
 \includegraphics[width=0.44\textwidth]{allDisastersExpectilesExtremiles.pdf}       
      \vspace*{-0.86 cm}
    
    \noindent
    \caption{US disaster data. Tropical cyclones (top panels), Floodings (middle panels) and all disasters (bottom panels). Left: estimates of various ($\delta=0.25, 0.5, 0.75$) expectiles of $X_{\sD_{\btau}}$ with $D_{\btau}$ corresponding to expected shortfall. Right: estimates of various ($\delta=0.25,0.5,0.75$) expectiles of  $X_{\sD_{\btau}}$ with $D_{\btau}=K_{\btau}$, i.e. corresponding to extremiles. In black, estimates of the usual expectiles. }
    \label{fig:ExpectilesAndExtremilesESFlooding}
\end{figure}
\FloatBarrier

\end{document}